\pdfoutput=1
\documentclass[draft]{lmcs}
\usepackage{preamble}
\usepackage{jon-tikz}
\usepackage{iblock}
\usepackage{macros}
\usepackage{categories}
\usetikzlibrary{nfold}

\begin{document}

\author{Daniel Gratzer\lmcsorcid{0000-0003-1944-0789}}
\address{Aarhus University}
\email{gratzer@cs.au.dk}

\author{Jonathan Weinberger\lmcsorcid{0000-0003-4701-3207}}
\address{Chapman University}
\email{jweinberger@chapman.edu}

\author{Ulrik Buchholtz\lmcsorcid{0000-0002-5944-6838}}
\address{University of Nottingham}
\email{ulrik.buchholtz@nottingham.ac.uk}

\title{The Yoneda embedding in simplicial type theory}

\begin{abstract}
  \textcite{riehl:2017} introduced \emph{simplicial type theory} (\STT), a variant of homotopy type
  theory which aimed to study not just homotopy theory, but its fusion with category theory:
  $\prn{\infty,1}$-category theory. While notoriously technical, manipulating $\infty$-categories in
  simplicial type theory is often easier than working with ordinary categories, with the type theory
  handling infinite stacks of coherences in the background. We capitalize on recent work by
  \textcite{gratzer:2024} defining the $\prn{\infty,1}$-category of $\infty$-groupoids in \STT{} to
  define presheaf categories within \STT{} and systematically develop their theory. In particular,
  we construct the Yoneda embedding, prove the universal property of presheaf categories, refine the
  theory of adjunctions in \STT{}, introduce the theory of Kan extensions, and prove Quillen's
  Theorem A. In addition to a large amount of category theory in \STT{}, we offer substantial
  evidence that \STT{} can be used to produce difficult results in $\infty$-category theory at a
  fraction of the complexity.
\end{abstract}

\maketitle

\begin{center}
  \emph{Dedicated to the dear memory of Thomas Streicher (1958--2025)}
\end{center}

\section{Introduction}

\textcite{russell:1919} famously described two styles of formalizing mathematics as the difference
between \emph{theft} and \emph{honest toil}. Both approaches can be seen in the present use of
dependent type theory. Honest toil involves proceeding \emph{analytically}: treating types as
basic objects equivalent to sets and defining and reasoning about objects like the real numbers,
groups, and topological spaces as one would ordinarily. This is what is done in \eg{}, the Rocq proof
of the Odd Order Theorem~\parencite{gonthier:2013}. The more expeditious route of theft involves
treating type theory as a bespoke \emph{synthetic} language for a particular kind of mathematical
object and postulating their basic properties. This narrows the scope of type theory but, by the
same token, makes proofs about those particular objects far more concise. For instance,
\emph{homotopy type theory} (\HoTT{})~\parencite{hottbook} postulates various axioms that ensure that types
behave like spaces (up to homotopy), making it possible to prove theorems from algebraic topology
without ever introducing an explicit description of a space. In reality, the synthetic approach is
less akin to theft than a loan; one pays for the customized type theory
with a semantic model that interprets types as the intended objects and validates the additional
axioms.

In this work, we embrace theft and, along with it, the synthetic methodology to use type theory to study category
theory. In particular, we add various axioms to homotopy type theory in order to construct a system
where \HoTT{}'s slogan ``all types are spaces and all functions are continuous'' is replaced by
``(some) types are ($\infty$-)categories and all functions are functors''.\footnote{%
  In this paper, by $\infty$-category we mean $\prn{\infty,1}$-category.}
This extension of
type theory is called \emph{simplicial} type theory (\STT{}) and was introduced by
\textcite{riehl:2017}.

While knowledge of $\infty$-categories is not necessary to use our theory, rough intuition for them
is helpful for understanding \STT{}. We therefore recall the following fuzzy definition. An
$\infty$-category $C$ is a collection of objects with a \emph{space} of arrows between objects $c$
and $d$, $\Hom{c}{d}$, rather than a set, equipped with a continuous composition operation and assignment
of identity arrows. Crucially, the composition operation need only be associative and unital up to
homotopy, but with the constraint that those homotopies themselves satisfy \emph{coherence} laws in
the form of additional homotopies, and so on with coherences between coherences, \etc{} As a loose
analogy, just as a monoidal category relaxes monoids by allowing $\otimes$ to be associative up to
isomorphisms satisfying certain coherence equations, $\infty$-categories weaken ordinary categories
to allow for the category laws to only hold up to (infinitely coherent) isomorphisms.

Remarkably, essentially every theorem one might hope for of ordinary categories holds for
$\infty$-categories.\footnote{At least, as one of our reviewers remarked, provided one correctly calibrates one's hopes.} However, the
proofs are vastly more complex as they work with \emph{models} of $\infty$-categories (tools used to
organize and manage the tower of coherences~\parencite{bergner:2018}). The goal of \STT{} is to use
type theory to hide coherences from the user and to allow for proofs that are no more difficult than
the classical arguments for $1$-categories.

In this work, we provide substantial evidence of this hypothesis by developing a large swathe of
category theory---several of the main results of \emph{Categories for the Working
  Mathematician}~\parencite{maclane:1978}---purely within \STT{}.

\subsection{Simplicial type theory}

To construct a type theory for synthetic category theory, one may hope to interpret type theory into
the category of categories ($\infty$ or otherwise) to ensure that types realize categories. However,
the category $\CAT$ of small categories is too poorly behaved to form a model of Martin-L{\"o}f type
theory (\MLTT{}). Instead, \textcite{riehl:2017} enlarge $\CAT$ and embed it as a reflective
subcategory in the ($\infty$-)presheaf category on the simplex category $\PSH{\SIMP}$ which is rich
enough to model \HoTT{}. \STT{} then axiomatizes some of $\PSH{\SIMP}$ to isolate $\CAT$ as a
reflective subuniverse within the type theory~\parencite{rijke:2020}.

We will introduce the full suite of additions in Section~\ref{sec:prelims} (collected in Appendix~\ref{sec:axioms}
for convenience), but the most important among them is the postulated \emph{interval type}
$\Int : \Uni[0]$. We further assume that $\Int$ is a bounded linear order with endpoints
$0,1 : \Int$. Intuitively, $\Int$ is meant to capture the category $\brc{0 \to 1}$---it is
interpreted as such in $\PSH{\SIMP}$---and we may use this to define and probe the type of
\emph{synthetic morphisms} in an arbitrary type $X$: an arrow in $X$ corresponds to an ordinary
function $\Int \to X$ with evaluation at $0,1$ yielding the domain and codomain. For instance, the
identity arrow at $x : X$ is given by $\lambda \_.\,x$.

However, just as the intended model $\PSH{\SIMP}$ is strictly larger than $\CAT$, not all types in \STT{}
faithfully model categories. In particular, while one is always able to construct identity
morphisms, not all types enjoy a composition operator. Remarkably, however, composition operators
are unique when they exist and their existence for a type $X$ is captured by a relatively short
proposition (Definition~\ref{def:prelims:category}). With a composition operation for $X$ to hand, we can
define the type of isomorphisms in $X$ and we define a category to be a type where (1) the
composition operation exists uniquely up to homotopy, and (2) the type of isomorphisms in $X$ is equivalent to the identity type $=_X$.

\begin{rem}
  This last point hinges crucially on \emph{not} assuming the uniqueness of identity proofs (UIP) lest we
  accidentally forbid any synthetic category from having an object with a non-trivial
  automorphism. However, by assuming isomorphisms and identity proofs coincide, we are able to
  leverage type theory's support for replacing equals by equals to seamlessly transport proofs along
  isomorphisms. This is why working with \HoTT{}/intensional type theory when formulating synthetic
  category theory proves more convenient than extensional type theory, even if one is unconcerned
  with $\infty$-categories.
\end{rem}

\subsection{Category theory inside of \texorpdfstring{\STT{}}{STT}}

While some recent work has investigated \STT{} for its applications to programming
languages~\parencite{weaver:2020,gratzer:2024,weaver:phd,grodin:2026}, the majority of work on simplicial type theory has
focused on proving results from category theory inside of type
theory~\parencite{riehl:2017,riehl:2023,riehl:2026,bardomiano:2024a,bardomiano:2024b,weinberger:phd,buchholtz:2023,weinberger:twosided:2024,weinberger:sums:2024}.
To this end, the theory of adjunctions, discrete and Grothendieck fibrations, and (co)limits have
been introduced and studied within simplicial type theory. Some of these results, \eg{}, a
fibrational Yoneda lemma~\parencite{riehl:2017}, were subsequently
mechanized~\parencite{kudasov:2024}, as well as various results from triangluated type theory~\parencite{agda-synthetic-categories}.

Until recently, however, there were no closed types in \STT{} which represented non-trivial
categories. As a result, while an excellent definition of adjunctions is presented by
\textcite{riehl:2017}, no examples can be given. In previous work, we changed this by
extending \STT{} to construct $\Space$, the category of $\mathcal{S}$paces, which is the homotopical analog of $\SET$~\parencite{gratzer:2024}.
Objects of $\Space$ are elements of $\Uni[0]$ that encode $\infty$-groupoids and morphisms in
$\Space$ correspond to functions thereof. \emph{Op.~cit.} uses $\Space$ as a building block to
recover algebraic categories (groups, rings) as well as other examples (posets, the
simplex category, \etc{}).

Our extension of \STT{} employed various \emph{modalities} on top of \HoTT{} to construct $\Space$. Here we take $\Space$ wholesale, but some of the modalities we used are still critical for stating natural theorems
in category theory. Accordingly, we also work within a modal extension of \HoTT{} based on
\MTT{}~\parencite{gratzer:2020} within this paper.

\subsection{Contributions}
We revisit the basic category theory in light of the construction of $\Space$ and show that the
majority of classical results one encounters in category theory are now within reach of simplicial
type theory. For the first time, we show that \STT{} can be used to prove vital theorems in
$\infty$-category theory without recourse to complex models.  Many of these theorems
(\eg{}, fully-faithful essentially surjective functors are equivalences) do not explicitly mention
$\Space$, but crucially rely on the reasoning principles enabled by $\Space$. We prove two
workhorse results from presheaf categories $\PSH{C}$:
\begin{itemize}
\item We construct a fully-faithful function $\Yo : C \to \PSH{C}$.
\item We prove that $\PSH{C}$ is the ``free cocompletion of $C$''.
\end{itemize}
The key technical innovation for these is the \emph{twisted arrow category}, which we integrate into
\STT{} as a modality. We are then able to deduce various classical results, \eg{}:
\begin{itemize}
\item that pointwise invertible maps in $C \to D$ are invertible;
\item that pointwise left adjoints are left adjoints;
\item that (co)limits are computed pointwise in $C \to D$;
\item the theory and existence of pointwise Kan extensions;
\item Quillen's theorem A;
\item the properness of cocartesian fibrations;
\item smooth and proper base change.
\end{itemize}

The synthetic approach yields concise proofs for many of these theorems compared with classical
expositions in $1$-category theory, but our proofs apply to $\infty$-categories as well and there the
improvements are far more radical: it takes hundreds of pages for \textcite{lurie:2009} to prove
that $\Yo$ is fully-faithful and the proof that pointwise natural transformations are isomorphisms
takes nearly five pages of effort by \textcite{cisinski:2019}. By dividing work between a
construction within \STT{} and the already-existing model of \STT{}, we are able to avoid many of
these technicalities and give proofs more familiar to $1$-category theorists. In particular, we show
that just as homotopy type theory allowed type theorists to produce new arguments in algebraic
topology, simplicial type theory enables type theorists to do the same with $\infty$-category
theory.

\begin{rem}
  Given that \STT{} extends \HoTT{} with a number of axioms, it is natural to ask whether these axioms
  are \emph{complete} in any sense. Our present suite of axioms is not complete for the intended
  models of simplicial objects in an $\infty$-topos (though they are sound) but this is neither
  surprising nor undesirable: \HoTT{} itself is not complete for its intended models ($\infty$-topoi)
  and its exotic models are a source of considerable interest~\cite{rasekh:2025a}. Similarly, we expect \STT{} to have
  interesting exotic models~\cite{rasekh:2025b} and cannot reasonably hope for a finite set of axioms to be complete for
  standard models. What is far more important is whether these axioms suffice to derive the standard
  results in category theory, an empirical rather than a mathematical question. Indeed, in
  related synthetic approaches to domain theory~\parencite{hyland:1991}, differential
  geometry~\parencite{kock:2006}, and algebraic geometry~\parencite{cherubini:2023}, the precise
  axioms arose over the course of multiple years and several iterations. To this end, we view our
  results as providing firm evidence towards the expressivity of this axiom set.
\end{rem}

\subsection{Organization}
In Section~\ref{sec:prelims} we review the highlights of the basis of this work: homotopy type
theory, basic simplicial type theory, modal homotopy type theory, and their synthesis: \STT{}. In
Section~\ref{sec:tw}, we study the \emph{twisted arrow} category and use it to construct the Yoneda
embedding. We prove several increasingly sophisticated versions of the Yoneda lemma and conclude
with a fully functorial version (Theorem~\ref{thm:tw:strong-yoneda}). In Section~\ref{sec:adj} we
put the Yoneda lemma to work to revisit the theory of adjunctions given by \textcite{riehl:2017}. We
develop several tools for constructing adjunctions and use them to give the first non-trivial
examples of adjunctions in \STT{}. We also use this machinery to show that $\PSH{C}$ is the free
cocompletion of $C$ (Theorem~\ref{thm:adj:cc}). In Section~\ref{sec:kan} we develop the theory of
Kan extensions in \STT{} and prove several vital results: the existence of pointwise Kan extensions
(Theorems~\ref{thm:kan:left-kan-exists} and \ref{thm:kan:right-kan-exists}), Quillen's theorem A
(Theorem~\ref{thm:kan:theorem-a}), and the properness of cocartesian maps
(Theorem~\ref{thm:kan:cocartesian-is-proper}). Our proof of the last fact is particularly notable,
as our use of type theory led us to a far simpler proof than those we are aware of in the
literature.

\section*{Acknowledgments}
For interesting and helpful discussions and their comments around the material of this work we
would like to thank
Mathieu Anel,
Carlo Angiuli,
Steve Awodey,
Denis-Charles Cisinski,
Bastiaan Cnossen,
Nicola Gambino,
Andr{\'e} Joyal,
Emily Riehl,
Mike Shulman,
Sam Staton,
Jonathan Sterling,
Thomas Streicher,
and Dominic Verity.
We also thank the anonymous reviewers for the careful reading and comments, both for the abbreviated
version of this article in LICS'25 as well as the reviewers for this extended version.

Jonathan Weinberger is grateful to the Fletcher Jones Foundation for their generous financial
support and fellowship at Chapman University's Fowler School of Engineering. He also thanks the US
Army Research Office for the support of some stages of this work under MURI Grant
W911NF-20-1-0082, hosted through the Department of Mathematics at Johns Hopkins University.

\section{Modal and simplicial type theory}
\label{sec:prelims}

In this paper we take \STT{} largely for granted and focus on working within the theory. However, to
make this paper more self-contained, we devote this section to carefully explaining the novel
constructs of modal homotopy type theory and the axioms supplementing it which form simplicial type
theory.

\subsection{Homotopy type theory}
We begin by recalling the basic concepts and notation from homotopy type theory we use in this
paper. The canonical reference is the \HoTT{} book~\parencite{hottbook}. We work within intensional
Martin-L{\"of} type theory and note how \HoTT{} extends this.

\begin{nota}
  We write $a =_A b$ for the identity type (often suppressing $A$). Given $p : a =_A b$ and
  $B : A \to \Uni$, we write $p_!$ for the map $B\prn{a} \to B\prn{b}$.
\end{nota}

\begin{nota}
  On occasion, we write $B^A$ instead of $A \to B$ as a more compact notation for non-dependent
  functions.
\end{nota}

\begin{defi}
  We say that a function $f : A \to B$ is an equivalence if $f$ admits both a left and a right
  inverse:
  \[
    \IsEquiv\prn{f} =
    \Sum{g,h : B \to A} \prn{g \circ f = \ArrId{}} \times \prn{f \circ h = \ArrId{}}
  \]
  We write $A \Equiv B$ for the sum $\Sum{f : A \to B} \IsEquiv\prn{f}$.
\end{defi}
\HoTT{} is an extension of intensional type theory with a hierarchy of universes satisfying the
univalence axiom:
\begin{restatable}{axiom}{univalenceaxiom}\label{ax:univalence}
  There is a family of constants of the following type:
  \[
    \Con{univ}_i
    :
    \Prod{A, B : \Uni[i]} \IsEquiv\prn{
      \lambda p.\,\prn{p_!, \cdots} : A =_{\Uni[i]} B \to A \Equiv B
    }
  \]
\end{restatable}

We shall suppress the $i$ in $\Con{univ}_i$ and $\Uni[i]$ and ignore size issues unless they are
relevant. Univalence produces a great number of paths in $\Uni$ that are distinct from $\Refl$. We
are often interested in types that are trivial, have only trivial paths, or trivial paths between
paths, \etc{} These conditions are organized into a family of predicates referred to as the truncation
level ($-2,-1,0,\ldots$) of a type.
We will only use the first three levels, stating that a type is
contractible or a (homotopy) proposition or set:
\begin{mathparpagebreakable}
  \IsContr\prn{A} = \Sum{a : A} \Prod{b : A} a = b
  \and
  \IsProp\prn{A} = \Prod{a,b : A} \IsContr\prn{a = b}
  \and
  \IsSet\prn{A} = \Prod{a,b : A} \IsProp\prn{a = b}
\end{mathparpagebreakable}

\begin{prop}[\textcite{shulman:2019}, see also \textcite{riehl:2024}]
  \label{thm:prelims:hott-models}
  All type-theoretic model topoi
  \textup(and, therefore, Grothendieck $\infty$-topoi\textup) model \HoTT{}.
\end{prop}

We shall also have occasion to use various \emph{higher inductives types} (HITs). The semantics of
HITs is complex and not directly addressed by the above result~\parencite{lumsdaine:2019}. In
particular, while \textcite{shulman:2019} shows that the above model supports all higher inductive
types, he does not show that universes are strictly closed under these constructions. While it is
work-in-progress to obtain this result, it is easy to show that universes are \emph{weakly} closed
under these constructions. For instance, there exists a type $D : \Uni[0]$ such that
$D \Equiv A \Pushout{C} B$ whenever $A,B,C : \Uni[0]$.
Accordingly, we shall assume that our
universes are closed under higher inductive types, albeit only with propositional
$\beta$-rules.

\subsection{Simplicial type theory}
With \HoTT{} to hand, we turn to simplicial type theory. This is an extension of \HoTT{} by a handful
of axioms that allow us to treat (certain) types as $(\infty,1$)-categories, henceforth just referred to as categories. We will consequently drop the $(\infty,1)$- or $\infty$-prefix everywhere. First and most fundamentally, we
add the following:

\begin{restatable}{axiom}{intaxiom}
  \label{ax:int}
  There is a set\/ $\Int$ that forms a bounded distributive lattice $\prn{0,1,\lor,\land}$ such
  that $\Prod{i,j : \Int} i \le j \lor j \le i$ holds.
\end{restatable}
We view $\Int$ as a \emph{directed interval}, and \textcite{riehl:2017} use this to equip every type
with a notion of synthetic morphism:
\begin{defi}
  A synthetic morphism $f : \Hom[X]{x}{y}$ where $x,y : X$ is a function $f : \Int \to X$ together
  with propositional equalities $f\,0 =_X x$ and $f\,1 =_X y$.
\end{defi}

\begin{rem}
  In \textcite{riehl:2017}, the synthetic interval is defined as more primitive judgmental structure
  and $\Hom{x}{y}$ uses strict extension types. This yields more definitional equalities: $f\,0$ and
  $x$ would coincide definitionally when $f : \Hom{x}{y}$ (and similarly for $f\,1$ and $y$). However, the judgmental approach does not
  straightforwardly include $\Int$ as a normal type and its interactions with modalities
  (Section~\ref{sec:prelims:modal-stt}) are complex. For these reasons, we work with the simpler but less
  strict definition of $\Hom{x}{y}$. These two definitions agree in our intended model and in a
  system where both strict extension types and a first-class interval is available, the two notions
  are equivalent~\parencite[Section 2.4]{buchholtz:2023}.
\end{rem}

One can define the identity morphism $\ArrId{x} : x \to x$ as $\lambda \_.\,x$. Moreover, every
function $f : X \to Y$ automatically has an action on synthetic morphisms $\alpha : \Int \to X$ by
post-composition $f \circ \alpha : \Int \to Y$. In this case, we often write $f\prn{\alpha}$.

From $\Int$ we immediately obtain the $n$-cubes $\Int^n$ and from them we can isolate simplices
$\Delta^n$, boundaries $\partial \Delta^n$, and horns $\Lambda_k^n$. In particular, $\Delta^2 \to X$
represents an $2$-cell in $X$ witnessing the composite of two arrows,
and $\Lambda^2_1 \to X$ represents a pair of composable arrow (without a composite). We recall the definitions
of these types below:
\begin{mathparpagebreakable}
  \Delta^n = \Compr{\prn{i_1, \dots, i_n} : \Int^n}{i_1 \ge i_2 \ge \dots \ge i_n}
  \and
  \Lambda^2_1 = \Compr{\prn{i,j} : \Int^2}{i = 1 \lor j = 0}
\end{mathparpagebreakable}

\begin{nota}
  We write $i : \Delta^n$ ($0 \le i \le n$) as shorthand for the sequence of $i$ copies of $1$ followed
  by $0$: $\prn{1, 1, \dots, 0, \dots}$.
\end{nota}

\begin{figure}
  \begin{tikzpicture}[diagram]
    \node (Ax00) {\tiny $0$};
    \node[right = 0.75cm of Ax00] (Ax10) {\tiny $1$};
    \node[below = 0.75cm of Ax00] (Ax01) {\tiny $1$};
    \path[->] (Ax00) edge node[above] {\tiny $i$} (Ax10);
    \path[->] (Ax00) edge node[left] {\tiny $j$} (Ax01);
    \SpliceDiagramSquare{
      nw/style = {below right = -0.3cm and 1.3cm of Ax01},
      height = 1cm,
      width = 1cm,
      nw = {\bullet},
      ne = {\bullet},
      sw = {\bullet},
      se = {\bullet},
    }
    \fill[pattern=north west lines]
    ([xshift=-0.1cm,yshift=0.1cm]nw.south east) --
    ([xshift=0.1cm,yshift=0.1cm]ne.south west) --
    ([xshift=0.1cm,yshift=-0.1cm]se.north west) --
    ([xshift=-0.1cm,yshift=-0.1cm]sw.north east) -- cycle;
    \node[right = 1.5cm of ne] (T0) {$\bullet$};
    \node[right = 1cm of T0] (T1) {$\bullet$};
    \node[below = 1cm of T1] (T2) {$\bullet$};
    \path[->] (T0) edge (T1);
    \path[->] (T0) edge (T2);
    \path[->] (T1) edge (T2);
    \fill[pattern=north west lines]
    ([yshift=0.1cm,]T0.south east) --
    ([yshift=0.1cm,xshift=0.15cm]T1.south west) --
    ([xshift=0.15cm]T2.north west) -- cycle;
    \node[right = 1.5cm of T1] (H0) {$\bullet$};
    \node[right = 1cm of H0] (H1) {$\bullet$};
    \node[below = 1cm of H1] (H2) {$\bullet$};
    \path[->] (H0) edge (H1);
    \path[->] (H1) edge (H2);
  \end{tikzpicture}
  \caption{Visualization of $\Int^2$, $\Delta^2$, and $\Lambda^2_1$.}
  \label{fig:prelims:viz}
\end{figure}

A map $f : \Delta^2 \to X$ is said to witness that the composite of $f\prn{-,0}$
followed by $f\prn{1,-}$ is $\lambda i.\,f\prn{i,i}$. We emphasize that this is \emph{data}; there
can be many distinct $f$'s witnessing the same composition as $X$ may have many non-equivalent
$2$-cells with the same boundary. By the same token however, it is not always the case that a pair of
composable morphisms $\Lambda^2_1 \to X$ extends to a composition datum $\Delta^2 \to X$. This is
precisely because not every type in \STT{} can be regarded as a category; even though we have
defined $\Hom[X]{x}{y}$ for every $X$, there is no \emph{a priori} way of \emph{composing} these morphisms.
Precategories are types for which all composites exist:
\begin{defi}
  A \emph{precategory} is a type $X$ satisfying the Segal condition:
  the inclusion $\Lambda^2_1 \to \Delta^2$ induces an equivalence
  $\IsEquiv\prn{\prn{\Delta^2 \to X} \to \prn{\Lambda^2_1 \to X}}$.
\end{defi}

Roughly, the Segal condition ensures that every pair of composable morphisms in $X$ extends
(uniquely) to a $2$-cell witnessing their composition and, in particular, there is an induced
composition function $\Hom{x}{y} \times \Hom{y}{z} \to \Hom{x}{z}$. Uniqueness automatically ensures
that this operation is associative and unital. The definition of a category refines this
slightly. In a precategory $X$ we are able to define the type of isomorphisms $x \cong y$ between
$x,y : X$ and so there are two potentially distinct types of evidence for $x$ and $y$ being
identical: $x =_X y$ and $x \cong_X y$. A category is a precategory for which these two types are
canonical equivalent.
\begin{defi}
  $\alpha : \Hom{x}{y}$ is an isomorphism ($\IsIso\prn{\alpha}$) if there exist
  $\beta_0,\beta_1 : \Hom{y}{x}$ such that $\beta_0 \circ f =
  \ArrId{}$,$f \circ \beta_1 = \ArrId{}$.\footnote{This is precisely the \HoTT{} equivalence but recast
    into synthetic morphisms.} We write $\Iso{x}{y}$ or $x \cong y$ for the subtype of isomorphisms.
\end{defi}
\begin{defi}
  \label{def:prelims:category}
    A precategory $C$ is a \emph{category} if it satisfies the Rezk condition:
    \[
      \Prod{\DeclVar{x,y}{}{C}} \IsEquiv\prn{\DeclVar{\mathsf{idtoiso}}{}{\prn{x=y} \to \Iso{x}{y}}}
    \]
    where $\mathsf{idtoiso}(\mathsf{refl}) := \ArrId{}$.

    If additionally every morphism in $C$ is an isomorphism, then $C$ is a \emph{groupoid}.
\end{defi}

\begin{defi}
  A map $f : A \to B$ is said to be a \emph{categorical equivalence} if for every category $C$
  precomposition induces an equivalence $(B \to C) \Equiv (A \to C)$.
\end{defi}

\begin{lem}
  $C$ is a groupoid if and only if $\IsEquiv\prn{\DeclVar{C^{!_\Int}}{}{C \to C^\Int}}$
  \textup($C$ is \emph{$\Int$-null} \textup{\cite{rijke:2020})}.
\end{lem}

\textcite{riehl:2017} develop the basic theory of these synthetic categories.
As noted above, every
function has an action on morphisms and \opcit{} shows that this action preserves compositions and
identities and therefore defines a functor. They also show that $C \to D$ is then a category
whenever $D$ is, and that synthetic morphisms $\Hom[D^C]{f}{g}$ are precisely natural
transformations. One can reformulate various classical categorical notions rather directly:
\begin{defiC}[\parencite{bardomiano:2021}]
  A natural transformation $\alpha : \Hom[C^I]{\Const\prn{c}}{F}$ witnesses $c$ as the limit of
  $F : C^I$ if $\alpha$ induces an equivalence $\Hom{c'}{c} \Equiv \Hom{\Const\prn{c'}}{F}$ for all $c'$.
\end{defiC}

\begin{defi}
  An adjunction between two categories $C,D$ consists of a pair of functions $f : C \to D$ and
  $g : D \to C$ with a natural isomorphism
  $\iota : \Prod{c,d} \Hom{f\prn{c}}{d} \Equiv \Hom{c}{g\prn{d}}$.
\end{defi}

While we have given a few examples of categories above, a notable type that is \emph{not} category
is the universe $\Uni$. Maps $A : \Int \to \Uni$ are too unstructured to compose and, in particular,
correspond neither to functions $A\prn{0} \to A\prn{1}$ nor $A\prn{1} \to A\prn{0}$ (consider
$\lambda i.\, i = 0$ or $\lambda i.\, i = 1$). In Section~\ref{sec:prelims:modal-stt}, we shall discuss the
subuniverse $\Space$ constructed previously~\parencite{gratzer:2024}, which is a category of groupoids whose
morphisms correspond to functions.  To properly situate this definition, we recall what it means for
$X : A \to \Uni$ to be covariant~\parencite{riehl:2017}, giving an assignment from morphisms
$\Hom{a_0}{a_1}$ to functions $X\prn{a_0} \to X\prn{a_1}$.

\begin{nota}
  Given $X : A \to \Uni$ we write $\TotalType{X}$ for $\Sum{a : A} X\prn{a}$.
\end{nota}

\begin{defi}
  A family $X : A \to \Uni$ is \emph{covariant}
  if for every $a : \Int \to X$ and $x_0 : X\prn{a\,0}$,
  the following is contractible:
  \[
    \Con{Lift}\prn{a,x_0} =
    \Sum{x : \prn{i : \Int} \to X\prn{a\,i}}
    \Proj\prn{x\,0} = x_0
  \]

  We further say the projection $\Sum{a : A} X\prn{a} \to A$ is covariant when $X$ is. For a general
  map $\pi : X \to A$ we write $X_a$ for $\Sum{x : X} \pi\prn{x} = a$ and say $\pi$ is covariant
  when $\lambda a.\,X_a$ is.
\end{defi}

Since $\Con{Lift}\prn{a,x_0}$ is contractible it has an inhabitant $x,\_$. This yields a function
$a,x_0 \mapsto x\,1$ which defines $a_! : X\prn{a_0} \to X\prn{a_1}$. The contractibility of
$\Con{Lift}\prn{a,x_0}$ ensures that these functions compose correctly, \etc{}

\begin{lem}
  \label{lem:prelims:covariant}
  A family $\DeclVar{X}{}{A \to \UU}$ is covariant if and only if the map
  $\bar{\pi} := \lambda p.\,\prn{p\prn{0}, \Proj[1] \circ p} : \TotalType X^\Int \to \TotalType X \times_A A^\Int$
  is an equivalence.
\end{lem}

In Sections~\ref{sec:adj:pw} and \ref{sec:kan:quillen-a}, we shall briefly use a weakening of covariance:
\begin{defi}
  A family $X : A \to \Uni$ is cocartesian if
  ${\TotalType{X}}^\Int \to A^\Int \times_{A^{\brc{1}}} {\TotalType{X}}^{\brc{1}}$ has a left
  adjoint $\ell \Adjoint \bar{\pi}$ such that the unit map is invertible and therefore induces an
  equality $\ArrId{} = \bar{\pi} \circ \ell$.
\end{defi}

One can give an equivalent characterization in terms of cocartesian morphisms and show that \eg{},
every morphism in $D$ can be factored as a cocartesian morphism followed by a vertical morphism:
\begin{thm}[\textcite{buchholtz:2023}]
  If a map $\pi : D \to C$ is cocartesian then for every $c : \Int \to C$, $x_0 : C_{c(0)}$, the
  category $\Con{Lift}\prn{c,x_0}$ has an initial object.
\end{thm}

\noindent
Dually, one can consider contravariant and cartesian families and fibrations.

Finally, we note that since categories and groupoids are defined by certain
\emph{orthogonality} conditions, by \textcite{rijke:2020} they define reflective subuniverses.
\begin{prop}
  There are idempotent monads $\Rezkify,\Grpdify$ such that, \eg{}, $\Rezkify X$ is a category and
  $C^{\Rezkify X} \Equiv C^X$ when $C$ is a category.
\end{prop}

\begin{prop}[\textcite{riehl:2017}]
  \label{thm:prelims:stt-models-classical}
  When Theorem~\ref{thm:prelims:hott-models} is specialized to simplicial spaces
  \textup($\PSH{\SIMP}$\textup), the
  resulting model validates Axiom~\ref{ax:int} and in this model categories are realized by
  $\infty$-categories \textup(modeled by complete Segal spaces\textup)
  and groupoids by $\infty$-groupoids.
\end{prop}

In fact, we shall require more of $\Int$ than merely the structure of a bounded distributive
lattice. We shall, for instance, also need to know that morphisms $\Int^n \to \Int$ are particularly
rigid (certain polynomials over $\Int$). This result, along with many others, is the result of a
certain internal \emph{duality} between bounded distributive lattice algebras over $\Int$ and
certain maps $- \to \Int$, akin to the duality between a ring and the corresponding affine scheme.
To this end, we introduce some terminology:

\begin{defi}
  A finitely-presented $\Int$-algebra is a bounded distributive lattice $A$ equipped with a
  homomorphism $\theta : \Int \to A$ along with mere existence of an isomorphism
  $\iota : A \cong \Int\brk{x_1, \dots, x_n}/\prn{f_1 = g_1, \dots, f_m = g_m}$ 
  such that $\theta$ is the composite of $\iota$ and the canonical map
  $\Int \to \Int\brk{\dots}/\prn{\dots}$.
\end{defi}

\begin{defi}
  If $\theta : \Int \to A$ is a (finitely-presented) $\Int$-algebra then the spectrum $\Spec\prn{A}$
  of $A$ is the type $\Hom[\Int]{\theta}{\ArrId{}}$ of homomorphisms $A \to \Int$ in the slice
  category under $\Int$.
\end{defi}

Finally, we may state our required duality axiom.
\begin{restatable}{axiom}{dualityaxiom}\label{ax:duality}
  If $A$ is a finitely-presented $\Int$-algebra, the following evaluation map is an equivalence:
  \[
    \lambda a\,f.\,f\prn{a} : A \to \prn{\Spec\prn{A} \to \Int}
  \]
\end{restatable}

This duality axiom sharply constrains the behavior of $\Int$. For instance, it implies the
following:
\begin{lem}
  The polynomial algebra of $n$ ordered variables $\Int\brk{x_1 \ge \dots \ge x_n}$ is equivalent to
  the type of maps $\Delta^n \to \Int$ where the equivalence evaluates a given polynomial $f$ on
  $v : \Delta^n$.
\end{lem}

\begin{cor}
  $\Int$ is a category and, consequently, so too are $\Int^n$ and $\Delta^n$ for all $n$.  
\end{cor}

\begin{cor}
  If $\theta : \Delta^n \to \Delta^n$ is an isomorphism, then $\theta = \ArrId{}$. In particular,
  there are no non-trivial automorphisms of $\Int$.
\end{cor}

This duality axiom holds in the classifying topos of bounded distributive lattices,
$\PSH{\FINPOSET}$~\cite{kock:2014,blechschmidt:2023}. By Bleschmidt~\cite{blechschmidt:2023}, it in
fact holds in all subtopoi of this classifying topos which contain the generic model $\Int$
(interpreted in the classifying topos as $\Yo{\brc{0 \le 1}}$). This includes dedekind cubical
sets~\cite{gratzer:2024} and simplicial sets. Consequently, the model of Riehl and Shulman extends
to support this axiom:
\begin{prop}\label{thm:prelims:stt-models} The model of
  Theorem~\ref{thm:prelims:stt-models-classical} validates Axiom~\ref{ax:duality}
\end{prop}

\begin{rem}
  While \textcite{blechschmidt:2023} and \textcite{kock:2014} deal only with 1-topoi, semi-lattices
  in HoTT are automatically homotopy sets. Consequently, one may check the validity of
  Axiom~\ref{ax:duality} on fibrant simplicial spaces which are homotopy sets which are weakly
  equivalent to ordinary simplicial sets. In fact, Blechschmidt's argument that these results descend
  to subtopoi containing the generic model can be replayed entirely internally and checked in HoTT.
  We also note that closely related axioms and consequences are considered by \textcite{pugh:2025}
  and \textcite{cherubini:2023}.
\end{rem}

\subsection{Modal homotopy type theory}
Many theorems in category theory require the ability to quantify over the objects in a category,
\eg{}, ``if $\alpha : F \to G$ is a natural transformation of functors $\CC \to \DD$ and each
$\alpha_c$ is invertible, then $\alpha$ is invertible''. A version of this is proven by
\textcite{riehl:2017}:
$\prn{\Prod{c : C} \IsIso\prn{\lambda i.\,\alpha\,i\,c}} \to \IsIso\prn{\alpha}$, but this is subtly
different as we discuss below. In fact, as it stands we cannot directly capture the classical
statement in \STT{}.

To understand the divergence between the \STT{} and classical results, note that by working
internally to type theory when proving $\Prod{c : C} \IsIso\prn{\lambda i.\,\alpha\,i\,c}$ we
cannot assume that $c$ is just an object in $C$: since it is an arbitrary element, we have to assume
it is constructed in an arbitrary context which might contain, \eg{}, a copy of $\Int$ such that $c$
represents a synthetic morphism. In fact, if we unfold the above type into the model we find that
constructing $\Prod{c : C} \IsIso\prn{\lambda i.\,\alpha\,i\,c}$ already entails proving that the
chosen inverses are natural. A great deal of the power of simplicial type theory comes from this
implicit naturality, but it makes this particular result weaker. After all, its purpose in standard
category theory was that in this particular situation, \emph{a priori} unnatural choices of inverses
will automatically be natural. Moreover, we shall encounter theorems that are simply false when
naively translated in this way.

Accordingly, to make \STT{} practical we must extend it with \emph{modalities}: unary type
constructors distinguished by their failure to respect substitution or apply in arbitrary
contexts. For instance, we shall eventually equip \STT{} with a modality $\Modify[\GM]{-}$ which
discards all non-invertible synthetic morphisms from a type to produce its \emph{core},
which we then
use to faithfully encode pointwise invertibility (see Example~\ref{ex:prelim:pw-naturality}).

A complete reference to the modal type theory we use---\MTT{}~\parencite{gratzer:2020}---is given by
\textcite{gratzer:phd} and we record formal rules in Section~\ref{sec:mtt-rules}. Fortunately, the rules
for, \eg{}, $\Sum{}$-types are unaffected by the addition of modalities. Accordingly, for brevity we only recall the new rules which must be added to \MLTT{} to extend \HoTT{} with modalities
\`a la \MTT{}.

\MTT{} is parameterized by a \emph{mode theory}: a strict $2$-category describing the collection
of modalities (the morphisms) available along with the natural transformations between them (the
$2$-cells).  We use $\mu,\nu,\xi$ to range over modalities. In the case of simplicial type theory,
our mode theory will have only one object along with a handful of generating modalities and
$2$-cells. There are four generating modalities $\GM,\SM,\OM,\TM$ subject to the following
equations:
\begin{mathparpagebreakable}
  \GM = \GM \circ \GM = \GM \circ \SM = \GM \circ \OM = \TM \circ \GM %
  \and
  \SM = \SM \circ \SM = \SM \circ \GM = \SM \circ \OM %
  \and
  \OM \circ \OM = \ArrId{}
\end{mathparpagebreakable}

We further require the following generating 2-cells:
\begin{mathparpagebreakable}
  \epsilon : \GM \to \ArrId{} \qquad \zeta : \ArrId{} \to \SM
  \and
  \tau : \TM \cong \TM \circ \OM \quad\text{(with $\tau^{-1}$)}
  \and
  \pi_0^\TM : \TM \to \OM
  \and
  \pi_1^\TM : \TM \to \ArrId{}
\end{mathparpagebreakable}

These 2-cells are likewise subject to a number of equations. For $\epsilon$ and $\zeta$, we
require $\zeta \Whisker \SM = \SM \Whisker \zeta = \ArrId{}$ viewing all of these as 2-cells
$\SM \to \SM$ and $\GM \Whisker \zeta = \ArrId{} : \GM \to \GM$ (using $\Whisker$ to denote
whiskering).  We require the dual equations on $\epsilon$. For the remaining four 2-cells, we
require that the following diagrams commute:
\[
  \begin{tikzpicture}[diagram]
    \node (T) {$\TM$};
    \node [below = of T] (OT) {$\TM \circ \OM$};
    \node [left = 2.5cm of OT] (OC) {$\OM$};
    \node [right = 2.5cm of OT] (C) {$\ArrId{}$};
    \path[->] (T) edge node[right] {$\tau$} (OT);
    \path[->] (OT) edge node[below] {$\pi^\TM_0 \Whisker \OM$} (C);
    \path[->] (OT) edge node[below] {$\pi^\TM_1 \Whisker \OM$} (OC);
    \path[->] (T) edge node[above] {$\pi^\TM_1$} (C);
    \path[->] (T) edge node[above] {$\pi^\TM_0$} (OC);
  \end{tikzpicture}
  \qquad
  \begin{tikzpicture}[diagram]
    \node (T) {$\GM$};
    \node [below = of T] (OT) {$\TM$};
    \node [left = 2.5cm of OT] (OC) {$\OM$};
    \node [right = 2.5cm of OT] (C) {$\ArrId{}$};
    \path[->] (T) edge node[near end,right] {$\TM\Whisker\epsilon$} (OT);
    \path[->] (OT) edge node[below] {$\pi^\TM_0$} (C);
    \path[->] (OT) edge node[below] {$\pi^\TM_1$} (OC);
    \path[->] (T) edge node[above] {$\epsilon$} (C);
    \path[->] (T) edge node[above=5pt] {$\epsilon\Whisker\OM$} (OC);
  \end{tikzpicture}
\]

Each morphism $\mu$ in the mode theory induces a modal type $\Modify[\mu]{-}$. We will describe the
rules for these modal types in a moment, but first we give some idea of what they are intended to
denote. For now this is merely intuition, though the axioms and model described in
Section~\ref{sec:prelims:modal-stt} will make it so. As already mentioned, $\Modify[\GM]{-}$ removes all
non-identity synthetic morphisms from a type. $\Modify[\SM]{-}$ is the right adjoint to this
operation and so it discards all non-identity morphisms but then freely adds all morphisms so that
an $n$-simplex $\Delta^n \to \Modify[\SM]{X}$ is exactly a collection of $n$ points in
$X$.\footnote{Note that while $\Modify[\GM]{X}$ will always be an $\infty$-category, in fact an
  $\infty$-groupoid, the same is not true of $\Modify[\SM]{X}$. In particular, $\Modify[\SM]{X}$
  will hardly ever satisfy the Rezk condition.} Next, $\Modify[\OM]{-}$ sends a type to its opposite
and, in particular, reverses the directions of all synthetic morphisms. Finally, $\Modify[\TM]{-}$
sends a type to its corresponding type of \emph{twisted arrows}; we shall analyze it in more depth
in Section~\ref{sec:tw}.

The formation rule for $\Modify[\mu]{-}$ is complex: the entire point of modalities is that
$\Gamma \vdash A$ does not imply $\Gamma \vdash \Modify{A}$. Instead, \MTT{} introduces a novel form
of context operation which acts like a ``left adjoint'' to $\Modify[\mu]{-}$:
\begin{mathparpagebreakable}
  \inferrule{
    \IsCx{\Gamma}
  }{
    \IsCx{\LockCxV{\Gamma}}
  }
  \and
  \inferrule{
    \IsTy[\LockCxV{\Gamma}]{A}
  }{
    \IsTy{\Modify{A}}
  }
  \and
  \inferrule{
    \IsTm[\LockCxV{\Gamma}]{a}{A}
  }{
    \IsTm{\MkMod{a}}{\Modify{A}}
  }
\end{mathparpagebreakable}
We refer to $\brc{\mu}$ as a \emph{modal restriction}. It is helpful to compare $\Modify[\mu]{A}$
with dependent products and, therefore, to see $\LockCxV{-}$ as extending the context by something
akin to a substructural ``$\mu$ variable''~\parencite{bahr:2017,birkedal:2020}. The real force of
modalities comes through how these $\brc{\mu}$s interact with variables. In particular, it is not the
case that $\Gamma, x : A, \brc{\mu} \vdash x : A$; since $\LockCxV{-}$ is intended to model a left
adjoint, we cannot generally assume that there is a weakening substitution
$\LockCxV{\Gamma} \to \Gamma$. Instead, we alter the rule extending a context with a variable so
that each variable is annotated with a modality:
\begin{mathparpagebreakable}
  \inferrule{
    \IsCx{\Gamma}
    \\
    \IsTy[\LockCxV{\Gamma}]{A}
  }{
    \IsCx{\MECxV{\Gamma}{x}{A}}
  }
  \and
  \inferrule{
    \alpha : \mu \to \Con{mods}\prn{\Gamma_1}
  }{
    \IsTm[\Gamma_0, \DeclVar{x}{\mu}{A}, \Gamma_1]{x^\alpha}{A^\alpha}
  }
\end{mathparpagebreakable}
In the above, $x^\alpha$ is the new form of variable rule while $A^\alpha$ is an admissible
operation on the syntax which traverses the term $A$ and appropriately updates all free variables
$y^\beta$ occurring within $A$ and modifying $\beta$ appropriately using $\alpha$. In particular, if
$A$ is closed then $A^\alpha = A$. In the formal syntax, both $x^\alpha$ and $A^\alpha$ are realized
by form of substitution, see \textcite{gratzer:mtt-journal:2021} for further details.

The original context extension is given by taking $\mu = \ArrId{}$. In the second rule,
$\Con{mods}\prn{\Gamma_1}$ is the composite $\nu_0 \circ \nu_1 \circ \cdots$ of all the
$\brc{\nu_i}$s occurring in $\Gamma_1$ (and is $\ArrId{}$ if there are no such occurrences). In
other words, a variable with annotation $\mu$ can be used precisely when it occurs behind a series
of modal restrictions for which there is a 2-cell navigating from $\mu$ to this composite. It is
therefore in the variable rule where the 2-cells comes into play.
\begin{lem}
  If $\IsTy[\Gamma, \DeclVar{x}{\mu}{A}]{B}$, $\IsTm[\Gamma, \DeclVar{x}{\mu}{A}]{b}{B}$,
  and $\IsTm[\LockCxV{\Gamma}]{a}{A}$, then $\IsTy{\Sb{B}{a/x}}$ and $\IsTm{\Sb{b}{a/x}}{\Sb{B}{a/x}}$.
\end{lem}

The final piece of the puzzle is the elimination rule for modalities. Roughly, this rule says that
modal annotations are equivalent to modal types ``from the perspective of a type'', \ie{}, that giving
an element in context $\MECxV{\Gamma}{x}[\nu]{\Modify[\mu]{A}}$ is the same as giving one in
$\MECxV{\Gamma}{x}[\nu\circ\mu]{A}$. This concretely amounts to the following pattern-matching rule
which allows us to assume that $\DeclVar{x}{\nu}{\Modify[\mu]{A}}$ is of the form $\MkMod{y}$ where
$\DeclVar{y}{\nu\circ\mu}{A}$:
\begin{mathparpagebreakable}
  \inferrule{
    \IsTy[\MECxV{\Gamma}{x}[\nu]{\Modify{A}}]{B}
    \\
    \IsTm[\MECxV{\Gamma}{y}[\nu\circ\mu]{A}]{b}{\Sb{B}{\MkMod[\mu]{y}/x}}
    \\
    \IsTm[\LockCxV{\Gamma}[\nu]]{a}{\Modify{A}}
  }{
    \IsTm{\LetMod{a}[y]{b}}{\Sb{B}{a/x}}
  }
  \and
  \inferrule{
    \IsTy[\MECxV{\Gamma}{x}[\nu]{\Modify{A}}]{B}
    \\
    \IsTm[\MECxV{\Gamma}{y}[\nu\circ\mu]{A}]{b}{\Sb{B}{\MkMod[\mu]{y}/x}}
    \\
    \IsTm[\LockCxV{\Gamma}[\nu\circ\mu]]{a}{A}
  }{
    \EqTm{\prn{\LetMod{\MkMod{a}}[y]{b}}}{\Sb{b}{a/y}}{\Sb{B}{\MkMod{a}/x}}
  }
\end{mathparpagebreakable}

While these rules account for all of the necessary extensions to handle modal types, we avail
ourselves of a convenience feature as well, modal $\Prod{}$-types:
\begin{mathparpagebreakable}
  \inferrule{
    \IsTm[\MECxV{\Gamma}{x}{A}]{b}{B}
  }{
    \IsTm{\lambda x.b}{\Prod{\DeclVar{x}{\mu}{A}} B}
  }
  \and
  \inferrule{
    \IsTm{f}{\Prod{\DeclVar{x}{\mu}{A}} B}
    \quad
    \IsTm[\LockCxV{\Gamma}]{a}{A}
  }{
    \IsTm{f\prn{a}}{\Sb{B}{a/x}}
  }
\end{mathparpagebreakable}

\begin{nota}
  ``If $\DeclVar{c}{\GM}{C}$, then $\Phi\prn{c}$'' signifies
  $\Prod{\DeclVar{c}{\GM}{C}} \Phi\prn{c}$.
\end{nota}

\begin{exa}
  \label{ex:prelim:pw-naturality}
  A faithful translation of ``pointwise invertibility implies invertibility'' where
  $\DeclVar{C,D}{\GM}{\Uni}$ and $\alpha : C \times \Int \to D$ is
  $\prn{\Prod{\DeclVar{c}{\GM}{C}} \IsIso\prn{\lambda i.\,\alpha\,i\,c}}
  \to \IsIso\prn{\alpha}$
\end{exa}

Immediately from these rules, we may prove the following:
\begin{prop}[\textcite{gratzer:2020}]
  \label{thm:prelims:mtt-facts}
  \leavevmode
  \begin{itemize}
  \item $\Modify[\mu]{-}$ commutes with $\Sum{}$ and $\ObjTerm{}$
  \item $\Con{comp} : \Modify[\mu]{\Modify[\nu]{-}} \Equiv \Modify[\mu\circ\nu]{-}$ and
    $\Modify[\ArrId{}]{-} \Equiv \ArrId{}$
  \item If $\alpha : \mu \to \nu$, then there is a map $\Coe^{\alpha} : \Modify[\mu]{{-}} \to \Modify[\nu]{-^\alpha}$.
  \item $\Con{transp} : \Modify[\GM]{\Modify[\GM]{A} \to B} \Equiv \Modify[\GM]{A \to \Modify[\SM]{B}}$
  \item $\Con{transp} : \Modify[\GM]{\Modify[\OM]{A} \to B} \Equiv \Modify[\GM]{A \to \Modify[\OM]{B}}$
  \end{itemize}
  The first point yields a function $\prn{\circledast} : \Modify{A \to B} \to \Modify{A} \to \Modify{B}$.
\end{prop}

When it will not confusion, we will suppress the equivalences
$\Modify[\mu]{\Modify[\nu]{A}} \Equiv \Modify[\mu\circ\nu]{A}$ and $\Modify[\ArrId{}]{A} \Equiv
A$. Furthermore, as there is no ambiguity, we suppress $\epsilon$ (and its whiskerings) and simply
write $x$ instead of $x^\epsilon$ and similarly for $\zeta : \ArrId{} \to \SM$. Consequently, if
$\DeclVar{A}{\GM}{\Uni}$ then we are able to simply write $\Modify[\TM]{A}$ rather than
$\Modify[\TM]{A^{\TM\Whisker\epsilon}}$. By convention, we also avoid writing $x^{\ArrId{}}$.

\begin{nota}
  If $\IsTm[\LockCxV{\Gamma}]{f}{A \to B}$, we write $f^\dagger$ for the function
  $\MkMod[\mu]{f} \mathbin{\circledast} -$.
\end{nota}

\begin{rem}\label{rem:prelims:coe-natural}
  Since we shall capitalize on the fact repeatedly, we note that $\Coe^\alpha$ is always suitably
  natural. For instance, fix $\DeclVar{f}{\GM}{A \to B}$. Then we construct a path
  $\alpha : \Coe^{\pi_1^\TM} \circ f^\dagger = f \circ \Coe^{\pi_1^\TM}$ as follows:
  \[
    \alpha = \Con{funext}\prn{\lambda x.\,\LetMod[\TM]{x}[x_0]{\Refl}}
  \]
  Since this path is essentially a commuting conversion (it is given by induction to allow
  $\Coe^{\pi_1^\TM}$ and $f^\dagger$ to reduce) it is fully coherent, with higher paths being
  likewise constructed by induction and then reflexivity.
  For instance, there are two paths
  $\Coe^\alpha \circ (f \circ g)^\dagger = f \circ g \circ \Coe^\alpha$: the canonical one and the
  one which factors as 
  \[
    \Coe^\alpha \circ (f \circ g)^\dagger
    = \Coe^\alpha \circ f^\dagger \circ g^\dagger
    = f^\dagger \circ \Coe^\alpha \circ g^\dagger
    = f \circ g \circ \Coe^\alpha
  \]
  These can be shown to agree by modal induction followed by reflexivity.
\end{rem}

\begin{nota}\label{not:prelims:extract}
  We will occasionally write $\Con{extract} : \Modify[\GM]{A} \to A$ for the composite given by $\Coe$
  and the equivalence $\Modify[\ArrId{}]{A} \Equiv A$.
\end{nota}

In general, $\Modify[\mu]{-}$ need not commute with propositional equality. However, this is true in
our intended models and so we impose it as an axiom:

\begin{restatable}{axiom}{crispidaxiom}\label{ax:crispid}
  The map $\MkMod[\mu]{a} = \MkMod[\mu]{b} \to \Modify{a = b}$ sending $\Refl$ to $\MkMod{\Refl}$ is
  an equivalence for all $\DeclVar{a,b}{\mu}{A}$.%
\end{restatable}

To be very precise, this map is defined by path induction in the family of types
$\lambda\prn{x,y:\Modify{A}}.\,\LetMod{x}[a]{\LetMod{y}[b]{\Modify{a = b}}}$. By
\textcite{gratzer:normalization:2022}, there is a computational account of this principle.

\begin{cor}
  Each $\Modify[\mu]{-}$ commutes with pullbacks
  $A \times_C B = \Sum{a : A} \Sum{b : B} f\prn{a} =_C g\prn{b}$.
\end{cor}

\begin{rem}
  For readers familiar with \emph{spatial type theory}~\parencite{shulman:2018}, this modal type
  theory is an extension of spatial type theory to include two additional modalities ($\OM$,
  $\TM$). In particular, the results of \textcite{shulman:2018} that deal with $\GM$ and $\SM$ can
  be reproduced in this setting.
\end{rem}

\subsection{Modalities and simplicial type theory}
\label{sec:prelims:modal-stt}

To connect the modal and simplicial structures, we impose the following axioms motivated by the
intended model, as described in Theorem~\ref{thm:prelims:stt-models} (and more generally
$\mathcal E^{\Op{\SIMP}}$ for an $\infty$-topos $\mathcal E$); see also the work of
\textcite{myers:2023}. First, the opposite map should be an anti-equivalence of $\Int$:

\begin{restatable}{axiom}{intopaxiom}
  \label{ax:op-of-int-is-int}
  There is an equivalence $\neg : \Modify[\OM]{\Int} \to \Int$ which swaps $0$ for $1$ and $\lor$ for $\land$.
\end{restatable}
\begin{cor}
  We can extend $\neg$ to an equivalence $\neg : \Modify[\OM]{\Delta^n} \Equiv \Delta^n$.
\end{cor}

We also require two possible notions of discreteness (being $\Int$-null or $\GM$-modal)
coincide:

\begin{restatable}{axiom}{intdiscaxiom}
  \label{ax:discrete-iff-crisp}
  If $\DeclVar{A}{\GM}{\Uni}$, then $\Modify[\GM]{A} \to A$ is an equivalence \textup($A$ is \emph{discrete}\textup) if
  and only if $A \to A^\Int$ is an equivalence \textup($A$ is \emph{$\Int$-null}\textup).
\end{restatable}

We note that this axiom implies Shulman's~\parencite{shulman:2018} axioms C1. In conjunction with
the following axiom, it also implies C2.

\begin{restatable}{axiom}{intglobalaxiom}
  \label{ax:global-points}
  The canonical map $\Bool \to \Int$ is injective and induces an equivalence $\Bool \Equiv \Modify[\GM]{\Int}$.
\end{restatable}
Motivated by our intended models, we insist that equivalences are jointly detected by
$\Delta^n$:

\begin{restatable}{axiom}{cubessepaxiom}
  \label{ax:cubes-separate}
  $\DeclVar{f}{\GM}{A \to B}$ is an equivalence if and only if
    the following holds:
    \[
      \Prod{\DeclVar{n}{\GM}{\Nat}}
      \IsEquiv\prn{
        \prn{f_*}^\dagger : \Modify[\GM]{\Delta^n \to A}
        \to
        \Modify[\GM]{\Delta^n \to B}
      }
    \]
\end{restatable}
\noindent
Note that since there is a section-retraction pair $\Delta^n \to \Int^n \to \Delta^n$, we can
replace $\Delta^n$ with $\Int^n$ in the above principle. This formulation of the axiom occurred in
\textcite{gratzer:2024} where it is also necessary---\opcit{} considers a variant of $\Int$ which is
not totally ordered, and so cubes are not colimits of simplices.

One useful application is the following:
\begin{lem}
  \label{lem:prelims:covariance-simplices}
  A map $\DeclVar{\pi}{\GM}{X \to A}$ is covariant if and only if the map
  $\Modify[\GM]{X^{\Delta^n}} \to \Modify[\GM]X \times_{\Modify[\GM]A} \Modify[\GM]{A^{\Delta^n}}$
  induced by $\prn{0^*}^\dagger$ and $\prn{\pi_*}^\dagger$ is an equivalence
  for all $\DeclVar n\GM\Nat$.
\end{lem}

\subsubsection*{The axiom for $\TM$}
Finally, we add a new axiom to \STT{} that governs $\TM$. For motivation, we recall some facts
about the external definition of the twisted arrow functor $\PSH{\SIMP} \to \PSH{\SIMP}$ which
$\Modify[\TM]{-}$ is intended to internalize. Classically, $\Con{Tw} : \PSH{\SIMP} \to \PSH{\SIMP}$
is defined as follows:
\[
  \Con{Tw}\prn{X}\prn{\brk{n}} = X\prn{\Op{\brk{n}} \ast \brk{n}}
\]
Here we have written $\Op{\brk{n}} \ast \brk{n}$ instead of the equivalent $\brk{2n + 1}$ to clarify
the action of this functor on morphisms: $f \mapsto \Op{f} \ast f$. (I.e., this corresponds to the
join operation on finite linear orders.)

As this functor is defined by precomposition, it is a right adjoint whose left adjoint is defined by
left Kan extension. In particular, this left adjoint sends $\Delta^n : \PSH{\SIMP}$ to $\Delta^{2n +
1}$ (again, with the functorial action given by twisting). As such, there is a universal map
$\Mor[\eta_n]{\Delta^n}{\Con{Tw}\prn{\Delta^{2n+1}}}$ which, when transposed, is given by the identity
$\Mor{\brk{2n + 1}}{\brk{2n + 1}}$. Universality of $\eta_n$ amounts to the requirement that each
morphism $\Mor[f]{\Delta^n}{\Con{Tw}\prn{C}}$ factors as $\Con{Tw}\prn{\hat{f}} \circ \eta_n$ for
some unique $\Mor[\hat{f}]{\Delta^{2n + 1}}{C}$. Our axiom governing $\Modify[\TM]{-}$ axiomatizes
this $\eta_n$ along with the property that
$\Con{Tw}\prn{-} \circ \eta_n : \Hom{\Delta^{2n+1}}{C} \to \Hom{\Delta^n}{\Con{Tw}\prn{C}}$ is an
equivalence. With this external motivation to hand, we proceed to fix some notation and state the
axiom which governs $\Modify[\TM]{-}$.

\begin{nota}
  If $\DeclVar{n}{\GM}{\Nat}$, we have canonical maps $i_l : \Delta^n \to \Delta^{2n + 1}$,
  $i_r : \Delta^n \to \Delta^{2n + 1}$, and $i_m : \Delta^1 \to \Delta^{2n + 1}$ which picks out
  $\brc{0, \dots, n}$, $\brc{n + 1, \dots, 2n + 1}$, and $\brc{n, n + 1}$ respectively.  For
  convenience, we write
  $\bar{i}_l = i_l^\dagger \circ \neg : \Delta^n \to \Modify[\OM]{\Delta^{2n + 1}}$.
\end{nota}

Finally, in forthcoming axiom, we require a procedure which extends a map
$f : {\Delta^n \to \Delta^m}$ to a map ${\Delta^{2n + 1} \to \Delta^{2m + 1}}$ which acts
appropriately on the images of two inclusions $i_l,i_r : \Delta^n \to \Delta^{2n + 1}$. To justify
this formally, we introduce the \emph{blunt join} $X \BluntJoin Y$:
\[
  X \BluntJoin Y = X \Pushout{X \times \brc{0} \times Y} \prn{X \times \Int \times Y} \Pushout{X \times \brc{1} \times Y} Y
\]
This is the directed version of the join $X \star Y$~\parencite[Ch 6]{hottbook} such that
$X \BluntJoin Y$ is roughly $X \Coprod{} Y$ with morphisms adjoined to connect each $x:X$ to each $y:Y$.
\begin{lem}
  \label{lem:prelims:blunt-join-of-simplices}
  If $C$ is a category, then
  $C^{\Delta^{m+1+n}} \Equiv C^{\Delta^{m} \BluntJoin \Delta^n}$.
\end{lem}
\begin{defi}
  If $\DeclVar{f}{\GM}{\Delta^n \to \Delta^m}$ and we take
  $\Con{twist}\prn{f} : {\Delta^{2n + 1} \to \Delta^{2m + 1}}$ to be the map given by uniquely
  extending the map
  $f^\dagger \BluntJoin f : \Modify[\OM]{\Delta^n} \BluntJoin {\Delta^n} \to \Modify[\OM]{\Delta^m}
  \BluntJoin {\Delta^m}$ along the categorical equivalences
  $\Modify[\OM]{\Delta^i} \BluntJoin {\Delta^i} \to \Delta^{2i + 1}$.
\end{defi}

\begin{restatable}{axiom}{twarraxiom}
  \label{ax:twisted-arrow-equiv}
  For each $\DeclVar{n}{\GM}{\Nat}$, there is a \textup(necessarily unique, as $\Delta^k$ has no
  non-identity automorphisms\textup) function
  $\DeclVar{\eta_n}{\GM}{\Delta^n \to \Modify[\TM]{\Delta^{2n+1}}}$ such that the following map is an
  equivalence, for each category $\DeclVar{C}{\GM}{\Uni}$:
  \[
    \iota \defeq \lambda \MkMod[\GM]{f}.\,\MkMod[\GM]{f^\dagger \circ \eta_n}
    :
    \Modify[\GM]{\Delta^{2n + 1} \to C}
    \to
    \Modify[\GM]{\Delta^{n} \to \Modify[\TM]{C}}
  \]
  Additionally, we require that
  $\neg^\dagger = \prn{\Coe^{\tau}}^\dagger : \Modify[\TM]{\Delta^{n}} \to \Modify[\TM]{\Modify[\OM]{\Delta^{n}}}$ 
  be an equivalence and
  that the diagrams in Figure~\ref{fig:prelims:tm-diagrams} commute \textup(these are mere properties---all
  objects are sets since $\Modify[\mu]{-}$ preserves h-level\textup).
\end{restatable}
\begin{figure}
  \centering
  \[
    \begin{tikzpicture}[diagram]
      \node (A) {$\Delta^n$};
      \node [right = 3cm of A] (T) {$\Modify[\TM]{\Delta^{2n + 1}}$};
      \node [above = of T] (B0) {$\Modify[\OM]{\Delta^{2n + 1}}$};
      \node [below = of T] (B1) {$\Delta^{2n + 1}$};
      \path[->] (A) edge node[above] {$\eta_n$} (T);
      \path[->] (T) edge node[right] {$\Coe^{\pi^\TM_0}$} (B0);
      \path[->] (T) edge node[right] {$\Coe^{\pi^\TM_1}$} (B1);
      \path[->] (A) edge node[below] {$i_r$} (B1);
      \path[->] (A) edge node[above] {$\bar{i}_l$} (B0);
    \end{tikzpicture}
    \qquad
    \begin{tikzpicture}[diagram]
      \SpliceDiagramSquare{
        width = 3cm,
        nw = \Delta^n,
        sw = \Delta^m,
        ne = \Modify[\TM]{\Delta^{2n + 1}},
        se = \Modify[\TM]{\Delta^{2m + 1}},
        north = \eta_n,
        south = \eta_m,
        west = f,
        east = \Con{twist}\prn{f}^\dagger
      }
    \end{tikzpicture}
  \]

  \caption{Laws for Axiom~\ref{ax:twisted-arrow-equiv}.}
  \label{fig:prelims:tm-diagrams}
\end{figure}

One may visualize $\iota$ as ensuring that $\Modify[\GM]{\Delta^n \to \Modify[\TM]{C}}$ is
isomorphic to a $2n + 1$ simplex in $C$:
\begin{equation*}
  \begin{tikzpicture}[diagram]
    \node (A0) {$c_n$};
    \node [right = of A0] (A1) {$c_{n-1}$};
    \node [right = of A1] (A2) {$\cdots$};
    \node [right = of A2] (A3) {$c_0$};
    \node [below = 1cm of A0] (B0) {$c_{n+1}$};
    \node [right = of B0] (B1) {$c_{n+2}$};
    \node [right = of B1] (B2) {$\cdots$};
    \node [right = of B2] (B3) {$c_{2n}$};
    \path[<-] (A0) edge (A1);
    \path[<-] (A1) edge (A2);
    \path[<-] (A2) edge (A3);
    \path[->] (B0) edge (B1);
    \path[->] (B1) edge (B2);
    \path[->] (B2) edge (B3);
    \path[->] (A0) edge (B0);
  \end{tikzpicture}
\end{equation*}

Under this correspondence, $\eta$ is the unique map $\Delta^n \to \Modify[\TM]{\Delta^{2n + 1}}$
given by the identity $\ArrId{} : \Delta^{2n + 1} \to \Delta^{2n + 1}$ and is thus the universal
$n$-simplex. The map $\pi^\TM_1$ picks out the bottom row and $\pi^\TM_0$ selects the top but
\emph{twisted} so that it lands in $\Modify[\OM]{C}$ rather than $C$. This axiom will only be used
in the proof of Theorem~\ref{thm:tw:tw}, where we use $\Modify[\TM]{-}$ to construct a bifunctorial version
of $\hom$.

\begin{prop}[\textcite{gratzer:2024}]
  The model constructed in Theorem~\ref{thm:prelims:stt-models-classical} extends to a model of
  modal \HoTT{} validating our axioms.
\end{prop}
\begin{rem}
  While our previous work~\parencite{gratzer:2024} did not handle $\Modify[\TM]{-}$, the methods
  employed there scale directly to this situation. In particular, \textcite{mukherjee:2023} give an
  explicit description of the necessary twisted arrow operation and show it is a Quillen right
  adjoint as required to extend the model.
\end{rem}
\begin{rem}
  We warn the reader that the set of axioms used in \textcite{gratzer:2024} and
  \textcite{gratzer:2026} is not identical to those we have listed above. In part, this is due to
  the difference in intended model (cubical spaces versus simplicial spaces) and there the numbering
  of the axioms is different simply for ease of presentation. In light of this, we have included a
  reference chart comparing the axioms here and those of \textcite{gratzer:2024} in Appendix~\ref{sec:axioms}.
\end{rem}

With modalities to hand, a number of results from classical category theory can be proven
directly. For instance, the so-called fundamental theorem of $\infty$-category theory:

\begin{restatable}{thm}{ffesssurj}
  \label{thm:prelims:ff-ess-surj-to-equiv}
  If $\DeclVar{C,D}{\GM}{\Uni}$ are categories, then $\DeclVar{F}{\GM}{C \to D}$ is an equivalence if
  (1) the induced map $\Modify[\GM]{C} \to \Modify[\GM]{D}$ is surjective, and (2) for any
  $\DeclVar{c,c'}{\GM}{C}$ the map $\Hom{c}{c'} \to \Hom{F\prn{c}}{F\prn{c'}}$ is an equivalence.
\end{restatable}
\begin{proof}
  Suppose (1) and (2) holds. We prove that $F$ is an equivalence using Axiom~\ref{ax:cubes-separate} and
  fix $\DeclVar{n}{\GM}{\Nat}$ such that it suffices to show
  $\IsEquiv\prn{F_*^\dagger : \Modify[\GM]{\Delta^n \to C} \to \Modify[\GM]{\Delta^n \to D}}$.

  If $n = 0$, then by (1) $F_*^\dagger$ is surjective and by (2) combined with the Rezk condition,
  it is an embedding. Accordingly, $F_*^\dagger$ is an equivalence in this case. The case for $n =
  1$ is an immediate consequence of the cases for $n = 0$ along with (2). In general, since
  $C^{\Delta^n} \Equiv C^{\Delta^1} \times_{C} \dots \times_C C^{\Delta^1}$ by the Segal condition
  and likewise for $D$. More precisely, the map $\Con{Spine}_n \to \Delta^n$ is orthogonal to all
  Segal types and $\Con{Spine}_n$ is equivalent to the iterated pushout
  $\Int \Pushout{\ObjTerm{}} \dots \Pushout{\ObjTerm{}} \Int$.
  As $\Modify[\GM]{-}$ commutes with pullbacks, the case for $n \ge 2$ follows from $n=0,1$.
\end{proof}

\subsection{Basic building blocks for categories}
Finally, we recall two results from our earlier work~\parencite{gratzer:2024} that will be used repeatedly within this work
to construct new categories. The first is a construction of \emph{full subcategories} using $\SM$:

\begin{prop}\label{thm:prelims:full-subcat}
  If $\DeclVar{C}{\GM}{\Uni}$ is a category and $\DeclVar{\phi}{\GM}{\Modify[\GM]{C} \to \Prop}$ is
  a predicate, then 
  \begin{enumerate}
  \item $C_\phi = \Sum{c : C} \Modify[\SM]{\phi\prn{\MkMod[\GM]{c}}}$ is a category,

  \item the projection map $C_\phi \to C$ induces an equivalence on hom-types,

  \item $\Modify[\GM]{C_\phi} \Equiv \Sum{c : \Modify[\GM]{C}} \phi\prn{c}$, and 

  \item a map $\DeclVar{F}{\GM}{D \to C}$ factors through $C_\phi$ if and only if
    $\phi\prn{\MkMod[\GM]{F\prn d}}$ holds for all $\DeclVar{d}{\GM}{D}$.
\end{enumerate}
\end{prop}

\begin{cor}\label{cor:prelims:crisp-ff-is-ff}
  If $\DeclVar{C,D}{\GM}{\Uni}$ are categories and $\DeclVar{F}{\GM}{C \to D}$, then the canonical
  map $\Hom{c}{c'} \to \Hom{F\prn{c}}{F\prn{c'}}$ is an equivalence for all $c,c' : C$ (notice the
  lack of $\GM$!) if and only if it is an equivalence when $\DeclVar{c,c'}{\GM}{C}$.
\end{cor}

Next, we recall the construction of the category of groupoids which plays the role of the
category of sets in simplicial type theory, \eg{}, we shall use this category to define presheaves:

\begin{prop}
  \label{thm:prelims:space} There is a category $\DeclVar{\Space_i}{\GM}{\Uni[i + 1]}$ with an
  embedding $\Space_i \to \Uni[i]$ such that:
  \begin{itemize}
    \item If $X : A \to \Space_i$, then the composite $A \to \Uni[i]$ is covariant.
    \item The converse holds for $\DeclVar{A}{\GM}{\Uni[i]}$, $\DeclVar{X}{\GM}{A \to \Uni[i]}$: if
      $X$ is covariant, then $X$ factors through $\Space_i$.
  \end{itemize}
\end{prop}

\begin{cor}[Directed univalence]
  \label{cor:prelims:dua}
  $\Space^\Int \Equiv \Sum{X_0,X_1 : \Space} X_1^{X_0}$ and composition in $\Space$ is the
  composition of functions.
\end{cor}

\begin{cor}
  If $\DeclVar{X}{\GM}{\Uni}$ is a groupoid, then $X : \Space$.
\end{cor}

\begin{rem}
  We proved~\parencite{gratzer:2024} Proposition~\ref{thm:prelims:space} in a richer variation of \STT{}
  (\emph{triangulated type theory}). Since we only require the result here, we take it as an ``axiom''
  of sorts to work in a simpler type theory and note that one could extend \STT{} to triangulated
  type theory to prove this theorem outright.
\end{rem}

\section{The Yoneda embedding}
\label{sec:tw}

Within this section, we fix a category $\DeclVar{C}{\GM}{\Uni}$. Our goal is to study the type
$\PSH{C} = \Space^{\Modify[\OM]{C}}$ of presheaves on $C$. As $\Space$ is a category, so
is $\PSH{C}$, and by directed univalence:
\begin{lem}
  \label{lem:tw:hom-of-psh}
  If $F,G : \PSH{C}$ then
  $\Hom{F}{G} \Equiv \Prod{\DeclVar{c}{}{\Modify[\OM]{C}}}{F\,c \to G\,c}$
\end{lem}

\begin{rem}
  Just as with \eg{}, completeness, $\PSH{C}$ implicitly fixes a universe level such that
  $\PSH{C} = \Modify[\OM]{C} \to \Space_i$. We may regard $i$ as a parameter or simply take $i =
  0$. Occasionally, we shall need to insist that $C \Equiv C'$ where $C' : \Uni[i]$ and in such
  situations we shall say that $C$ is \emph{small}. We assume all categories are locally
  small---that each $\Hom[C]{c}{c'}$ is small.
\end{rem}

One may recast the \emph{fibrational} Yoneda lemma proven by \textcite{riehl:2017} to take advantage
of $\PSH{C}$ rather than quantifying over contravariant families as in \opcit:
\begin{lem}
  \label{lem:tw:fibrational-yoneda}
  If $F : \PSH{C}$ and $c : \Modify[\OM]{C}$ then
  $F\prn{c} \cong \Prod{c' : \Modify[\OM]{C}} \Hom[\Modify[\OM]{C}]{c}{c'} \to F\prn{c'}$
\end{lem}

For ease of reference, we also explicitly record the dual version of this lemma:
\begin{lem}\label{lem:tw:fibrational-cov-yoneda}
  If $F : C \to \Space$ and $c : C$ then
  $F\prn{c} \cong \Prod{c' : C} \Hom[C]{c}{c'} \to F\prn{c'}$
\end{lem}

\subsection{The twisted arrow category and the Yoneda embedding}

In light of this last result, the natural next step is to define a map $C \to \PSH{C}$
which sends $c : C$ to something like $\Hom{-}{c}$.\footnote{Here
  we see why $C$ must be flat: we wish to discuss both $C$ and $\Modify[\OM]{C}$. It is helpful to
  understand $\DeclVar{C}{\GM}{\Uni}$ as a \emph{closed} type which depends on nothing in the
  context and, in particular, need not be treated functorially.}
However, caution is required: $\Hom{-}{c}$ has type $C \to \Uni$ and not the
required $\Modify[\OM]{C} \to \Space$. Upon reflection, the reader should find it surprising that
$\Hom{-}{-} : C \times C \to \Uni$ at all; if all maps are functorial in \STT{} how can $\Hom{-}{-}$
be covariant in both arguments? In fact, this is a consequence of the strange behavior of synthetic
morphisms in $\Uni$. While $\Hom{-}{-}$ is functorial in both arguments, the lack of directed
univalence for $\Uni$ makes this useless. This strangeness ensures that $\Hom{-}{-}$ does not
restrict to a function into $\Space$.

What is required instead is a function $\Phi : \Modify[\OM]{C} \times C \to \Space$ such that
$\Phi\prn{\MkMod[\OM]{c},-} = \Hom{c}{-}$ whenever $\DeclVar{c}{\GM}{C}$,
\ie{}, a function that agrees on
objects with $\Hom{-}{-}$ and has the same functoriality in the second argument, but takes
$\Modify[\OM]{C}$ as its first argument. In fact, it is highly non-obvious where such a function
should come from; \textcite[p.~xii]{riehl:2022} specifically highlight this construction as
remarkably subtle in $\infty$-category theory. It is for this reason that we introduced
$\Modify[\TM]{-}$. Recall the visualization of $\Modify[\GM]{\Delta^n \to \Modify[\TM]{C}}$:
\begin{equation}
  \label{eq:tw:visual}
  \begin{tikzpicture}[diagram,baseline=(link)]
    \node (A0) {$c_n$};
    \node [right = of A0] (A1) {$c_{n-1}$};
    \node [right = of A1] (A2) {$\cdots$};
    \node [right = of A2] (A3) {$c_0$};
    \node [below = 1cm of A0] (B0) {$c_{n+1}$};
    \node [right = of B0] (B1) {$c_{n+2}$};
    \node [right = of B1] (B2) {$\cdots$};
    \node [right = of B2] (B3) {$c_{2n}$};
    \path[<-] (A0) edge (A1);
    \path[<-] (A1) edge (A2);
    \path[<-] (A2) edge (A3);
    \path[->] (B0) edge (B1);
    \path[->] (B1) edge (B2);
    \path[->] (B2) edge (B3);
    \path[->] (A0) edge coordinate[pos=.5] (link) (B0);
  \end{tikzpicture}
\end{equation}
The projection to $\Modify[\OM]{C}$ gives the top row and the map to $C$ yields the bottom. This
visualization for $n$-simplices is very similar to that of
$C^{\Int} = \Sum{c_0,c_1} \Hom{c_0}{c_1}$, but the top row has been twisted to ensure that
one restriction lands in $\Modify[\OM]{C}$ as required for a bifunctorial version of $\Hom{-}{-}$:

\begin{restatable}{thm}{twlfib}
  \label{thm:tw:tw}
  If $\DeclVar{C}{\GM}{\Uni}$ is a category, then the following holds:
  \begin{itemize}
  \item The map
    $\gl{\Coe^{\pi_0^\TM},\Coe^{\pi_1^\TM}} : \Modify[\TM]{C} \to \Modify[\OM]{C} \times C$
    straightens to $\Phi : \Modify[\OM]{C} \times C \to \Space$.
  \item For every $\DeclVar{c}{\GM}{C}$, the map
    $\alpha_c : \Hom{\Hom{c}{-}}{\Phi\prn{\MkMod[\OM]{c},-}}$ induced by the covariant Yoneda lemma
    (Lemma~\ref{lem:tw:fibrational-cov-yoneda})
    applied to $\iota\prn{\MkMod[\GM]{\ArrId{c}}} : \Phi\prn{\MkMod[\OM]{c},c}$ is
    an equivalence.
  \item For every $\DeclVar{c}{\GM}{C}$, the map
    $\beta_c : \Hom{\Hom[\Modify[\OM]{C}]{c}{-}}{\Phi\prn{-,c}}$ induced by the Yoneda lemma
    (Lemma~\ref{lem:tw:fibrational-yoneda})
    applied to $\iota\prn{\MkMod[\GM]{\ArrId{c}}} : \Phi\prn{\MkMod[\OM]{c},c}$ is
    an equivalence.
  \end{itemize}
\end{restatable}

\begin{lem}
  \label{lem:tw:simple-coherence}
  Given $\DeclVar{f}{\GM}{\Delta^{1} \to C}$ and let
  $\bar{f} = \iota\prn{\MkMod[\GM]{f}} : \Modify[\TM]{C}$, then there exist  paths:
  \begin{gather*}
    \theta\prn{f}_0 :
    \prn{\Coe^{\pi_0^{\TM}} \circ \Con{extract}\prn{\bar{f}}}\prn{\ast}
    =
    \MkMod[\OM]{f\prn{0}}
    \quad
    \theta\prn{f}_1 :
    \prn{\Coe^{\pi_1^{\TM}} \circ \Con{extract}\prn{\bar{f}}}\prn{\ast}
    =
    f\prn{1}
  \end{gather*}
  These paths are natural in $C$ so that \eg{}, the two paths of the following shape induced by
  $\theta\prn{g \circ f}_1$ and $\theta\prn{f}_1$ with the naturality of $\Coe^{\pi_1^\TM}$ at $g$ agree:
  \[
    \prn{\Coe^{\pi_1^{\TM}} \circ \Con{extract}\prn{\iota\prn{\MkMod[\GM]{g \circ f}}}}\prn{\ast}
    =
    g\prn{f\prn{1}}
  \]
  Here $\ast : \Delta^0$ is the unique element of the unit type and we remind the reader that
  $\iota$ comes from Axiom~\ref{ax:twisted-arrow-equiv} while $\Con{extract}$ appears in
  Notation~\ref{not:prelims:extract}.
\end{lem}
\begin{proof}
  We show the second, as they are symmetric. We define $\theta_1$ using the naturality of
  $\Coe^{\pi_1^{\TM}}$ and the behavior of $\Coe^{\pi_1^{\TM}}$ on $\eta$
  from Figure~\ref{fig:prelims:tm-diagrams}:
  \begin{align*}
    &\prn{\Coe^{\pi_1^{\TM}} \circ \Con{extract}\prn{\iota\prn{\MkMod[\GM]{f}}}}\prn{\ast}
    \\
    &= \prn{\Coe^{\pi_1^{\TM}} \circ f^\dagger \circ \eta_0}\prn{\ast}
    \\
    &= \prn{f \circ \Coe^{\pi_1^{\TM}} \circ \eta_0}\prn{\ast} && \text{By naturality, Remark \ref{rem:prelims:coe-natural}}
    \\
    &= f\prn{1} && \text{By the first diagram in Figure~\ref{fig:prelims:tm-diagrams}}
  \end{align*}

  To prove that $\theta_1$ is natural in $C$, we observe that the terms agree up to a commuting
  conversion of elimination rules for modal types. Accordingly, we may prove that these two paths
  agree by induction on $\eta_0\prn{\ast}$ and then reflexivity.
\end{proof}

\begin{proof}[Proof of Theorem~\ref{thm:tw:tw}]
  We begin by showing that the map
  $\gl{\Coe^{\pi^\TM_0},\Coe^{\pi^\TM_1}} : \Modify[\TM]{C} \to \Modify[\OM]{C} \times C$ is a
  covariant family. By Lemma~\ref{lem:prelims:covariance-simplices} it suffices to show that the following map induced by
  $\brc{0} : \Delta^0 \to \Delta^n$ is an equivalence:
  \[
    \epsilon : \Modify[\GM]{\Modify[\TM]{C}^{\Delta^n}} \to
    \prn{
      \DelimMin{1}
      \Modify[\GM]{\Modify[\TM]{C}}
      \times_{\Modify[\GM]{\Modify[\OM]{C}} \times \Modify[\GM]{C}}
      \prn{\Modify[\GM]{\Modify[\OM]{C}^{\Delta^n}} \times \Modify[\GM]{C^{\Delta^n}}}
    }
  \]
  For convenience, we begin by applying a few modal transformations (in particular, using
  $\Con{transp} : \Modify[\GM]{A \to \Modify[\OM]{B}} \Equiv \Modify[\GM]{\Modify[\OM]{A} \to B}$)
  such that it suffices to show that the following map is an equivalence:
  \[
    \epsilon' :
    \Modify[\GM]{\Modify[\TM]{C}^{\Delta^n}} \to
    \prn{
      \DelimMin{1}
      \Modify[\GM]{\Modify[\TM]{C}}
      \times_{\Modify[\GM]{C} \times \Modify[\GM]{C}}
      \prn{\Modify[\GM]{C^{\Modify[\OM]{\Delta^n}}} \times \Modify[\GM]{C^{\Delta^n}}}
    }
  \]

  To prove this, we shall construct a commutative diagram:
  \begin{equation}
    \DiagramSquare{
      width = 7cm,
      sw = \Modify[\GM]{\Modify[\TM]{C}^{\Delta^n}},
      nw = \Modify[\GM]{C^{\Delta^{2n+1}}},
      se =
      \Modify[\GM]{\Modify[\TM]{C}}
      \times_{\Modify[\GM]{C} \times \Modify[\GM]{C}}
      \prn{\Modify[\GM]{C^{\Modify[\OM]{\Delta^n}}} \times \Modify[\GM]{C^{\Delta^n}}},
      ne =
      \Modify[\GM]{C^{\Delta^1}}
      \times_{\Modify[\GM]{C} \times \Modify[\GM]{C}}
      \prn{\Modify[\GM]{C^{\Modify[\OM]{\Delta^n}}} \times \Modify[\GM]{C^{\Delta^n}}},
      north = \phi,
      south = \epsilon',
      west = \iota,
      east = \prn{\iota,\ArrId{}}
    }
    \label{diag:key-diagram}
  \end{equation}
  We define $\phi$ momentarily, but we first remark that $\prn{\iota,\ArrId{}}$ is
  well-formed because of Lemma~\ref{lem:tw:simple-coherence}, which ensures that applying $\iota$ and then
  evaluating commutes appropriately with projection.

  Note also that the two vertical maps are equivalences because $\iota$ is an
  equivalence. Accordingly, by 3-for-2, to show $\epsilon'$ is an equivalence,
  it suffices to ensure that $\phi$ is an equivalence making the diagram commute.\footnote{%
    We emphasize that the filler for this square is irrelevant. Any filler
    suffices to show that $\epsilon'$ is an equivalence, which in turn implies that $\epsilon$ is an
    equivalence as required.}
  We now define $\phi$ as follows:
  \[
    \phi\prn{\MkMod[\GM]{f}} \defeq
    \prn{
      \MkMod[\GM]{\Restrict{f}{n \le n + 1}},
      \prn{
        \MkMod[\GM]{\Restrict{f}{0 \le \dots \le n} \circ \neg},
        \MkMod[\GM]{\Restrict{f}{n + 1 \le \dots \le 2n + 1}}
      },
      \Refl
    }
  \]
  $\phi$ is given by restricting along a categorical equivalence
  ($\Modify[\OM]{\Delta^{n}} \BluntJoin \Delta^n \to \Delta^{2n + 1}$),
  so it is an equivalence.

  Next, we note that all four of the maps in this diagram are weakly natural in $C$. For the bottom
  and top maps, this is an easy observation---the top is given by restriction and the bottom uses
  restriction along with $\Coe^{\pi_1^{\TM}}$, which is also natural
  by Remark~\ref{rem:prelims:coe-natural}. For the left-hand map, this is a consequence of the naturality
  of $\iota$. For the right-hand map, the only wrinkle is the paths used to witness that $\iota$
  commutes with evaluating on projections. This requires a filler for a certain path, but this is
  precisely the naturality coherence supplied by the latter part of Lemma~\ref{lem:tw:simple-coherence}.

  Finally, we argue that the diagram commutes.  Fix
  $\MkMod[\GM]{f} : \Modify[\GM]{\Delta^{2n + 1} \to C}$. We wish to show that
  $\epsilon'\prn{\iota\prn{\MkMod[\GM]{f}}}
  = \prn{\iota,\ArrId{}}\prn{\phi\prn{\MkMod[\GM]{f}}}$. By naturality, however, we
  may use $f$ to reduce to the case where $C = \Delta^{2n + 1}$ and $f = \ArrId{}$. In this case,
  everything involved is a set and so it suffices to argue the diagram commutes when we replace each
  pullback with a simple product. With this in place, we now calculate:
  \begin{align*}
    \epsilon'\prn{\iota\prn{\MkMod[\GM]{\ArrId{}}}}
    &= \epsilon'\prn{\MkMod[\GM]{\eta_n}}
    \\
    &=
      \prn{
        \MkMod[\GM]{\eta_n\,0},
        \prn{
          \Con{transp}\prn{\MkMod[\GM]{\Coe^{\pi_0^{\TM}}\circ \eta_n}},
          \MkMod[\GM]{\Coe^{\pi_1^{\TM}} \circ \eta_n}
        }
      }
    \\
    &=
      \prn{
        \MkMod[\GM]{\eta_n\,0},
        \prn{
          \Con{transp}\prn{\MkMod[\GM]{\bar{i}_l}},
          \MkMod[\GM]{i_r}
        }
      }
    \\
    &=
      \prn{
        \MkMod[\GM]{i_m\prn{\eta_0\,\ast}},
        \prn{
          \MkMod[\GM]{i_l \circ \neg},
          \MkMod[\GM]{i_r}
        }
      }
    \\
    &=
      \prn{\iota,\ArrId{}}\prn{
        \MkMod[\GM]{i_m},
        \prn{
          \MkMod[\GM]{i_l \circ \neg},
          \MkMod[\GM]{i_r}
        }
      }
    \\
    &=\prn{\iota,\ArrId{}}\prn{\phi\prn{\MkMod[\GM]{\ArrId{}}}}
  \end{align*}

  This completes the first step of the argument. The second is to show that for each
  $\DeclVar{c}{\GM}{C}$ the map
  $\alpha_c : \Hom[C \to \Space]{\Hom{c}{-}}{\Phi\prn{\MkMod[\OM]{c},-}}$ is an isomorphism. Passing
  to total spaces, it suffices to show the following map is an equivalence:
  \[
    \tilde{\alpha}_c = \lambda \prn{d,f}.\,\prn{d,f_*\prn{\iota\prn{\MkMod[\GM]{\ArrId{c}}}}}
    : \Sum{d : C} \Hom{c}{d} \to \Sum{d : C} \Phi\prn{\MkMod[\OM]{c},d}
  \]
  In the above, $f_*$ is the covariant transport operation on $\Phi\prn{\MkMod[\OM]{c},-}$.
  Since both sides of this map are categories, it suffices to show that this map is
  fully faithful and essentially surjective.

  In fact, $\tilde\alpha_c$ is an equivalence on objects. To this end, we observe that if
  $\DeclVar{f}{\GM}{\Hom{c}{d}}$ for some
  $\DeclVar{d}{\GM}{C}$, then we can construct the transport $f_*$ alternatively as follows.
  Define a path $h : \Delta^1 \to \Modify[\TM]{C}$ by $h \defeq\iota{\prn{\MkMod[\GM]{\lambda \_, \_, k.\, f\prn{k}}}}$,
  i.e.,
  $h$ corresponds to the following doubly degenerate 3-simplex in $C$:
  \[
    \begin{tikzcd}
      c\ar[d,"\ArrId{c}"'] & c\ar[l,"\ArrId{c}"'] \\
      c\ar[r,"f"'] & d
    \end{tikzcd}
  \]
  We then consider the morphism $\lambda i.\,\prn{\pi^\TM_1\prn{h\,i}, h\,i}$ in
  $\Sum{d : C} \Phi\prn{\MkMod[\OM]{c},d}$. Using the definition of $\iota$ and the naturality of
  $\eta$, this is a morphism from $\prn{c,\iota\prn{\MkMod[\GM]{\ArrId{c}}}}$ to
  $\prn{d, \iota\prn{\MkMod[\GM]{f}}}$. Moreover, the naturality of $\Coe^{\pi_1^{\TM}}$ ensures
  that it lies over $f$ in $C$. Consequently,
  $\tilde{\alpha}_c\prn{d,f} = \prn{d,\iota\prn{\MkMod[\GM]{f}}}$ when restricted to
  $\DeclVar{d}{\GM}{C}$ and $\DeclVar{f}{\GM}{\Hom{c}{d}}$, which is an equivalence because $\iota$
  is invertible.

  For fully-faithfulness, it suffices to show that the following map is an equivalence:
  \[
    \tilde{\alpha}_c :
    \Modify[\GM]{\Int \to \Sum{d : C} \Hom{c}{d}}
    \to
    \Modify[\GM]{\Int \to \Sum{d : C} \Phi\prn{\MkMod[\OM]{c},d}}
  \]
  However, as both sides are the total spaces of covariant families, it suffices to show that the
  following map is an equivalence:
  \[
    \tilde{\alpha}_c :
    \Modify[\GM]{\Sum{d : \Int \to C} \Hom{c}{d\,0}}
    \to
    \Modify[\GM]{\Sum{d : \Int \to C} \Phi\prn{\MkMod[\OM]{c},d\,0}}
  \]
  The conclusion now follows from the previous case.
\end{proof}

\begin{nota}
  We write $\Phi_D : \Modify[\OM]{D} \times D \to \Space$ for the same construction applied to some
  category $D$. Within this section, we continue to write $\Phi$ as shorthand for $\Phi_D$.
\end{nota}
\begin{nota}
  \label{nota:tw:nat-tw-hom}
  If $\DeclVar{F}{\GM}{C \to D}$ we write
  $\Phi_F : \Prod{c_0 : \Modify[\OM]{C}, c_1 : C} \Phi_C\prn{c_0,c_1} \to \Phi_D\prn{F\prn{c_0},F\prn{c_1}}$
  for the map induced by the following commutative cube:
  \[
    \begin{tikzpicture}[diagram]
      \node (P)  at (0,0)      {$\Phi\prn{c_0,c_1}$};
      \node (TC) at (6,0)      {$\Modify[\TM]{C}$};
      \node (B)  at (3,-1.75)  {$\Phi\prn{F^\dagger\prn{c_0}, F\prn{c_1}}$};
      \node (TP) at (10,-1.75) {$\Modify[\TM]{D}$};
      \node (U1) at (0,-3.5)   {$\Unit$};
      \node (OC) at (6,-3.5)   {$\Modify[\OM]{C} \times C$};
      \node (U2) at (3,-5.25)  {$\Unit$};
      \node (OP) at (10,-5.25) {$\Modify[\OM]{D} \times D$};
      \path[->] (P) edge (TC);
      \path[->] (P) edge (U1);
      \path[->] (TC) edge (TP);
      \path[->] (TC) edge (OC);
      \path[->] (TP) edge (OP);
      \path[->] (U1) edge node[pos=0.3,above] {$\gl{c_0,c_1}$} (OC);
      \path (U1) edge[equality] (U2);
      \path[->] (OC) edge node[above right] {$F^\dagger \times F$} (OP);
      \path[->] (U2) edge node[below] {$\gl{F^\dagger\prn{c_0}, F\prn{c_1}}$} (OP);
      \path[->,dotted] ($(P)+(0.55,-0.35)$) edge node[pos=0.55,above right] {$\Phi_{F}$} ($(B)+(-0.2,0.3)$);
      \draw[white,line width=4pt] (B.east) -- (TP.west);
      \draw[white,line width=4pt] (B.south) -- (U2.north);
      \path[->] (B) edge (TP);
      \path[->] (B) edge (U2);
      \node at (1.6,-0.25)  {$\lrcorner$};  %
      \node at (0.55,-1.35) {$\lrcorner$};  %
      \node at (6.75,-0.9)  {$\lrcorner$};  %
      \node at (3.7,-2.4)   {$\lrcorner$};  %
    \end{tikzpicture}
  \]
  The uniqueness of covariant transport shows that the following square commutes
  where $\DeclVar{c}{\GM}{C}$ and $c' : C$:
  \[
    \DiagramSquare{
      width = 5cm,
      ne = \Hom{c}{c'},
      se = \Hom{F\prn{c}}{F\prn{c'}},
      nw = \Phi\prn{\MkMod[\OM]{c},c'},
      sw = \Phi\prn{\MkMod[\OM]{F\prn{c}},F\prn{c'}},
    }
  \]
\end{nota}

\begin{cor}
  \label{cor:tw:tw-of-op}
  If $\DeclVar{c_0}{}{\Modify[\OM]{C}}$ and $c_1 : C$, then
  $\Phi\prn{c_0,c_1} = \Phi_{\Modify[\OM]{C}}\prn{\MkMod[\OM]{\MkMod[\OM]{c_1}},c_0}$.
\end{cor}
\begin{proof}
  Passing to total spaces, it suffices to find an equivalence $\Modify[\TM]{C} \to
  \Modify[\TM]{\Modify[\OM]{C}}$ fitting into the following diagram:
  \[
    \begin{tikzpicture}[diagram]
      \node (T) {$\Modify[\TM]{C}$};
      \node[right = 5cm of T] (OT) {$\Modify[\TM]{\Modify[\OM]{C}}$};
      \node[below right = 1.75cm and 2.5cm of T] (C) {$\Modify[\OM]{C} \times C$};
      \path[->] (T) edge (C);
      \path[->] (OT) edge (C);
      \path[->] (T) edge (OT);
    \end{tikzpicture}
  \]

  The map $\Coe^\tau$ precisely satisfies this role: it is invertible because the 2-cell $\tau$
  is an isomorphism and it fits into the commuting diagram because of the corresponding diagram
  in the mode theory.
\end{proof}

\subsection{The Yoneda lemma}

With a bi-functorial version of $\Hom{-}{-}$ to hand, we can now straightforwardly define the Yoneda
embedding $\Yo$ and leverage Lemma~\ref{lem:tw:fibrational-yoneda} into a result about $\Yo$:
\begin{defi}[Yoneda]
   $\Yo = \lambda c.\,\Phi\prn{-,c} : C \to \PSH{C}$.
\end{defi}
\begin{lem}
  \label{lem:tw:weak-yoneda}
  $\Hom{\Yo{c}}{X} \cong X\prn{\MkMod[\OM]{c}}$ for all $X : \PSH{C}$ and $\DeclVar{c}{\GM}{C}$.
\end{lem}
\begin{proof}
    Since $c$ is $\GM$-annotated, using Theorem~\ref{thm:tw:tw} and Corollary~\ref{cor:tw:tw-of-op} we have the following
    identification
    $\Hom[\Modify[\OM]{C}]{\MkMod[\OM]{c}}{-} = \Phi\prn{-,c}$. Moreover, by Lemma~\ref{lem:tw:hom-of-psh} we
    additionally have the following:
    \[
      \Prod{c' : \Modify[\OM]{C}} \Phi\prn{c',c} \to X\prn{c'} \cong \Hom{\Yo{c}}{X}
    \]
    The conclusion now follows by Lemma~\ref{lem:tw:fibrational-yoneda}.
\end{proof}

A great deal of category theory is contained within Lemma~\ref{lem:tw:weak-yoneda}.
It shows that $\Yo$ is
fully-faithful on $\GM$-annotated elements of $C$ and that $C$ is a full subcategory of $\PSH{C}$:
\begin{lem}
  \label{lem:tw:repr}
  $\Yo : C \to \PSH{C}$ induces an equivalence $C \Equiv \PSH{C}_{\Con{isRepr}}$ where
  $\Con{isRepr} = \lambda X.\,\Sum{c : C} X = \Yo{c}$.\footnote{Note that $\Con{isRepr}\prn{X}$ is a
  proposition due to Lemma~\ref{lem:tw:weak-yoneda} and Corollary \ref{cor:prelims:crisp-ff-is-ff}.}
\end{lem}
While Lemma~\ref{lem:tw:weak-yoneda} follows directly from
Lemma~\ref{lem:tw:fibrational-yoneda}, the above consequence can only be expressed once there exists a
\emph{category} of presheaves---something missing from \textcite{riehl:2017}. This opens up a new
proof strategy: to prove a result of $C$, we first prove that it holds for $\Space$, then
$\PSH{C}$, then that it restricts to the full subcategory. For instance, we may prove the
aforementioned characterization of natural isomorphisms:
\begin{thm}
  \label{thm:tw:nat-iso}
  If $\DeclVar{C,D}{\GM}{\Uni}$ are categories, $\DeclVar{F,G}{\GM}{C \to D}$, and
  $\DeclVar{\alpha}{\GM}{\Hom{F}{G}}$, then $\Prod{c : C} \IsIso\prn{\alpha\,c}$ if
  $\Prod{\DeclVar{c}{\GM}{C}} \IsIso\prn{\alpha\,c}$.
\end{thm}
\begin{proof}
  Note that this theorem is trivial for $C = \Delta^0$ and for $C = \Delta^1, D = \Space$ it is
  a consequence of Corollary~\ref{cor:prelims:dua}. The Segal condition for $\Space$ then implies the theorem
  for $C = \Delta^n,D = \Space$.

  By Lemma~\ref{lem:tw:repr}, it suffices to assume $D = \PSH{D_0}$.  By
  Axiom~\ref{ax:discrete-iff-crisp} and Theorem~\ref{thm:prelims:mtt-facts} it suffices to show that
  $\prn{\Sum{c : C} \IsIso\prn{\alpha\,c}} \to \prn{\Sum{c : C}
    \Modify[\SM]{\IsIso\prn{\alpha\,c}}}$ is an equivalence. By Axiom~\ref{ax:cubes-separate}, it
  suffices to prove for all $n$:
  \[
    \IsEquiv\Bigl(\Modify[\GM]{\prn{\Sum{c : C} \IsIso\prn{\alpha\,c}}^{\Delta^n}}
     \to \Modify[\GM]{\prn{\Sum{c : C} \Modify[\SM]{\IsIso\prn{\alpha\,c}}}^{\Delta^n}}\Bigr)
  \]
  Unfolding and commuting $\GM$ with $\Sum{}$, it suffices to show that
    for every $\DeclVar{c}{\GM}{\Delta^n \to C}$ the following holds:
    \[
      \Sum{\DeclVar{\sigma}{}{\Delta^n}} \IsIso\prn{\alpha\prn{c\,\sigma}}
      \Equiv
      \Sum{\DeclVar{\sigma}{\GM}{\Delta^n}} \IsIso\prn{\alpha\prn{c\,\sigma}}
    \]
  Replacing $\alpha$ with $\alpha \circ c$, however, reduces us
  to the already proven case of $C = \Delta^n, D = \Space$.
\end{proof}

Lemma~\ref{lem:tw:weak-yoneda} is already powerful. However, it does not capture that this
equivalence is \emph{natural} in both $c$ and $X$---or, more precisely, since $c$ is $\GM$-annotated
and the equivalence is in $\Uni$, the naturality it yields is trivial. We are able to prove a far
stronger version of the Yoneda lemma that (1) does not need to assume that $\DeclVar{c}{\GM}{C}$,
and (2) yields the desired functoriality in both $c$ and $X$. To do so, we replace $\Hom{-}{-}$
with $\Phi$:

\begin{restatable}[Functorial Yoneda lemma]{thm}{functorialyoneda}
  \label{thm:tw:strong-yoneda}
  There is a natural isomorphism
  $\Phi_{\PSH{C}}\prn{\Yo^\dagger\prn{-}, -} \cong \Con{eval} : \Modify[\OM]{C} \times \PSH{C} \to \Space$.
\end{restatable}
\begin{rem}
  This result uses a handful of results from Section~\ref{sec:adj}. These forward references are justified:
  we do not use Theorem~\ref{thm:tw:strong-yoneda} till Section~\ref{sec:adj:psh}. We present the proof here for
  conceptual coherence.
\end{rem}
\begin{proof}
  The central difficulty in this proof is to find a map
  $\Phi_{\PSH{C}}\prn{\Yo^\dagger\prn{-}, -} \to \Con{eval}$ which can then be checked to be an
  equivalence. To construct this map, we use the presentation of
  $\Modify[\OM]{C} \times \PSH{C} \to \Space$ as covariant families over
  $\Modify[\OM]{C} \times \PSH{C}$. In particular, we consider the following pullback diagrams:
  \[
    \begin{tikzpicture}[diagram]
      \node (PhiPsh) {$\Modify[\TM]{\PSH{C}}$};
      \node [below = 2cm of PhiPsh] (PshPsh) {$\Modify[\OM]{\PSH{C}} \times \PSH{C}$};
      \node [pullback,left = 3.6cm of PhiPsh] (V) {$V$};
      \node [left = 3.6cm of PshPsh] (CPsh) {$\Modify[\OM]{C} \times \PSH{C}$};
      \node [pullback,left = 3.6cm of V] (PhiC) {$\Modify[\TM]{C}$};
      \node [left = 3.6cm of CPsh] (CC) {$\Modify[\OM]{C} \times C$};
      \path[->] (PhiPsh) edge (PshPsh);
      \path[->] (V) edge (PhiPsh);
      \path[->] (V) edge (CPsh);
      \path[->] (PhiC) edge node[above] {$v$} (V);
      \path[->] (PhiC) edge (CC);
      \path[->] (CPsh) edge node[below] {$\Yo^\dagger \times \ArrId{}$} (PshPsh);
      \path[->] (CC) edge node[below] {$\ArrId{} \times \Yo$} (CPsh);
    \end{tikzpicture}
  \]
  \[
    \begin{tikzpicture}[diagram]
      \node (Spt) {$\Sum{A : \Space} A$};
      \node [below = 2cm of Spt] (S) {$\Space$};
      \node [pullback,left = 3.75cm of Spt] (V) {$W$};
      \node [left = 3.75cm of S] (CPsh) {$\Modify[\OM]{C} \times \PSH{C}$};
      \node [pullback,left = 3.75cm of V] (PhiC) {$\Modify[\TM]{C}$};
      \node [left = 3.75cm of CPsh] (CC) {$\Modify[\OM]{C} \times C$};
      \path[->] (Spt) edge (S);
      \path[->] (V) edge (Spt);
      \path[->] (V) edge (CPsh);
      \path[->] (PhiC) edge node[above] {$w$} (V);
      \path[->] (PhiC) edge (CC);
      \path[->] (CPsh) edge node[below] {$\Con{eval}$} (S);
      \path[->] (CC) edge node[below] {$\ArrId{} \times \Yo$} (CPsh);
    \end{tikzpicture}
  \]

  The claim is then that $V \Equiv W$. To show this, we argue as follows. Over $\Modify[\OM]{C} \times \PSH{C}$,
  consider the \emph{free} covariant family $Z$ of the composite
  $\Modify[\TM]{C} \to \Modify[\OM]{C} \times C \to \Modify[\OM]{C} \times \PSH{C}$. Consider then
  the fibered map $\bar{v}$  (respectively $\bar{w}$) from $Z$ to $V$ (respectively $W$), induced by
  $v$ (respectively $\bar{w}$). We will show that both $\bar{v}$ and $\bar{w}$ are equivalences. The
  conclusion follows: $\bar{v} \circ \bar{w}^{-1}$ is the desired equivalence.

  Using \textcite[Theorem 2.41]{rijke:2020} and \textcite[Subsection~5.2.5]{buchholtz:2023} the free covariant fibration $Z : \Modify[\OM]{C}
  \times \PSH{C} \to \Uni$ over said composite is given as follows. For $\DeclVar{c}{\GM}{\Modify[\OM]{C}}$ and $\DeclVar{X}{\GM}{\PSH{C}}$ by
  Corollary~\ref{cor:adj:push-pull-formula} we have:
  \begin{align*}
    &Z\prn{c, X} =
    \\
    &\ \Grpdify\prn{\Sum{c_0,c_1} \Hom[\Modify[\OM]{C}]{c_0}{c} \times \Hom[\PSH{C}]{\Yo{c_1}}{X} \times \Phi\prn{c_0,c_1}}
  \end{align*}
  To show that \eg{}, $\bar{v}$ is an equivalence, we must show that the following map induced by
  $v$ is an equivalence:
  \[
    Z\prn{c,X} \to \Phi\prn{\Yo^\dagger\prn{c},X}
  \]
  We may use Theorem~\ref{thm:tw:nat-iso} and assume that there exists $\DeclVar{c'}{\GM}{C}$ such
  that $c = \MkMod[\OM]{c'}$ and that $\DeclVar{X}{\GM}{\PSH{C}}$. Moreover, since the right-hand
  side is a groupoid, $\bar{v}$ is uniquely induced by extending the canonical map of the following
  type:
  \begin{align*}
    &\bar{v}_0 : \prn{\Sum{c_0,c_1} \Hom[\Modify[\OM]{C}]{c_0}{c} \times \Hom[\PSH{C}]{\Yo{c_1}}{X} \times \Phi\prn{c_0,c_1}}
    \to \Phi\prn{\Yo^\dagger\prn{c},X}
    \\
    &\bar{v}_0\prn{c_0,c_1,f,x,t} = \Phi\prn{\Yo^\dagger\prn{f},x}\prn{\Phi_{\Yo}\prn{t}}
  \end{align*}

  Intuitively, the map $\Phi\prn{{\Yo^\dagger{f}},x}$ acts on
  $\Phi_{\Yo}\prn{t} : \Phi\prn{{\Yo}^\dagger\prn{c_0}, \Yo{c_1}}$ by precomposing with
  ${\Yo}^\dagger\prn{f} : \Hom[\Modify[\OM]{\PSH{C}}] {{\Yo}^\dagger\prn{c_0}} {{\Yo}^\dagger\prn{c}}$ and
  postcomposing with $x : \Hom[\PSH{C}]{\Yo{c_1}}{X}$, hence giving an
  element of $\Phi\prn {\Yo^\dagger\prn{c},X}$.
    
  One may check directly that the
  assignment $x \mapsto \eta\prn{c,c',\ArrId{},x',\ArrId{}}$ is a quasi-inverse to this map where
  $x' : \Hom{\Yo\prn{c'}}{X}$ corresponds to $x : \Phi\prn{\Yo^\dagger\prn{c},X}$
  under the equivalence $\alpha$ induced by Theorem~\ref{thm:tw:tw}.

  The only difficulty here is to show that $\prn{c_0,c_1,f,x,t}$ is equal to
  $\prn{c,c,\ArrId{},\alpha^{-1}\prn{t},\alpha\prn{\ArrId{}}}$. 
  To this end, we construct a series of maps of the following type:
  \[
    \Int \to {\Sum{c_0,c_1} \Hom{c_0}{c} \times \Hom{\Yo{c_1}}{X} \times \Phi\prn{c_0,c_1}}
  \]
  Our goal is for these maps to ``link'' the two endpoints of our equation, as all such maps are
  converted into equalities by $\Grpdify$. Specifically, we have a series of maps:
  \begin{align*}
    &\prn{c_0,c_1,f,x,t}
    \\
    &\sim \prn{c,c_1,\ArrId{},x,\Phi\prn{f,\ArrId{}}\prn{t}}
    \\
    &= \prn{c,c_1,\ArrId{},x,\alpha\prn{g}} && \text{where } g = \alpha_c^{-1}\prn{\Phi\prn{f,\ArrId{}}\prn{t}} : \Hom{c}{c_1}
    \\
    &=\prn{c,c_1,\ArrId{},x,\Phi\prn{\ArrId{},g}\prn{\alpha\prn{\ArrId{}}}}
    \\
    &\sim \prn{c,c,\ArrId{},x \circ \Yo{g},\alpha\prn{\ArrId{}}}
  \end{align*}

  Finally, we show $x \circ \Yo{g} : \Hom{\Yo{c}}{X}$ is equivalent to
  $\alpha^{-1}\prn{\Phi\prn{\Yo{f},x}\prn{\Phi_{\Yo}\prn{t}}}$, completing our proof. To see this last equality, we note
  that
  $\Yo{g} = \alpha_{\Yo{c}}^{-1}\prn{\Phi_{\Yo}\prn{\Phi\prn{f,\ArrId{}}\prn{t}}}$---\ie{},
  $\alpha$ is suitably natural as per the square of Notation~\ref{nota:tw:nat-tw-hom}---so $x \circ \Yo{g}$ is in
  turn equal to $\alpha_{\Yo{c}}^{-1}\prn{\Phi\prn{\ArrId{},x}\prn{\Phi_{\Yo}\prn{\Phi\prn{f,\ArrId{}}\prn{t}}}}$
  which in turn is equivalent to
  $\alpha_{\Yo{c}}^{-1}\prn{\Phi\prn{\Yo{f},x}\prn{\Phi_{\Yo}\prn{t}}}$ by uniqueness of covariant
  transports.

  The case for $w$ is similar. Once more we fix $\DeclVar{c}{\GM}{C}$ along with
  $\DeclVar{X}{\GM}{\PSH{C}}$ and we must show that the following canonical map is an equivalence:
  \begin{align*}
    &\bar{w} : \Grpdify \prn{\Sum{c_0,c_1} \Hom[\Modify[\OM]{C}]{c_0}{c} \times \Hom[\PSH{C}]{\Yo{c_1}}{X} \times \Phi\prn{c_0,c_1}}
    \to X\prn{c}
    \\
    &\bar{w}\prn{\eta\prn{c_0,c_1,f,x,t}} = x\,c\,\prn{\Phi\prn{f,\ArrId{}}\prn{t}}
  \end{align*}
  Just as before, we consider a putative quasi-inverse sending $\prn{c,x}$ to
  $\prn{c,c,\ArrId{},x',\alpha\prn{\ArrId{}}}$ where $x' : \Hom{\Yo{c}}{X}$ is
  uniquely determined by the property that $x'\,c\,\prn{\alpha\prn{\ArrId{}}} = x$. We
  directly check one of the two equalities and verify the remaining one through a similar zig-zag
  type argument to the above. Indeed, the same argument as the above leads us the check the
  following identity where
  $c_0 : \Modify[\OM]{C}$, $c_1 : C$, $x : \Hom{\Yo{c_1}}{X}$, $t : \Phi\prn{c_0,c_1}$, and $f : \Hom{c_0}{c}$:
  \[
    x\,c\,\prn{\Phi\prn{f,\ArrId{}}\prn{t}}
    =
    (x \circ \Yo{\alpha^{-1}\prn{\Phi\prn{f,\ArrId{}}\prn{t}}})\,c\,\prn{\alpha\prn{\ArrId{}}}
  \]

  Unfolding the right-hand side, we see that it is equal to the following
  \[
    x\,c\,\prn{\Phi\prn{\ArrId{},\alpha^{-1}\prn{\Phi\prn{f,\ArrId{}}\prn{t}}}\,\prn{\alpha\prn{\ArrId{}}}}
  \]
  Finally, we note that 
  $\Phi\prn{\ArrId{},\alpha^{-1}\prn{\Phi\prn{f,\ArrId{}}\prn{t}}}\,\prn{\alpha\prn{\ArrId{}}}$ can be
  identified with $\Phi\prn{f,\ArrId{}}\prn{t}$ by definition of $\alpha$.
\end{proof}

\begin{rem}
  Inspecting the above proof, we find that it guarantees that our equivalence
  $\Con{ev} \to \Phi_{\PSH{C}}\prn{\Yo^{\dagger}\prn{-},-}$ is given by $\Phi_{\Yo}$
  when restricted to representable presheaves.
\end{rem}

\section{Revisiting adjunctions}
\label{sec:adj}

With presheaves and the Yoneda embedding available, we now revisit the theory of adjoint functors
introduced by \textcite{riehl:2017} in \STT{}. They define a pair of functions $f : C \to D$ and
$g : D \to C$ to be adjoint when equipped with
$\iota : \Prod{c,d} \Hom{f\prn{c}}{d} \Equiv \Hom{c}{g\prn{d}}$. While they produce several
equivalent reformulations using a unit and counit natural transformations, no non-trivial examples
of adjunctions are given---unsurprisingly, since concrete examples of categories in \STT{} are relatively
recent. Even with $\Space$ available it is quite difficult to produce examples of such
adjunctions.

It is far more feasible to construct only $f$ and then show that
$\Phi\prn{f^\dagger\prn{-},d} : \PSH{C}$ is representable for every $\DeclVar{d}{\GM}{D}$. This is
comparable to Theorem~\ref{thm:tw:nat-iso}: we wish to give a functorial definition of either $f$ or $g$
and a \emph{non-functorial} definition of the other, and then show that this can be upgraded to a
full adjunction. In this section, we show that this is indeed possible, and we observe that a number
of important adjunctions and results are then immediately within reach. In particular, we shall use
this technique to prove that $\PSH{C}$ is cocomplete and, moreover, is the free cocompletion of $C$.

\subsection{Pointwise adjunctions to adjunctions}
\label{sec:adj:pw}

Let us begin by formalizing the notion of pointwise adjoints:
\begin{defi}
  We say that $\DeclVar{f}{\GM}{C \to D}$ is a pointwise left adjoint if the following type is
  inhabited:
  \[
    \Prod{\DeclVar{d}{\GM}{D}} \Con{isRepr}\prn{\Phi\prn{f^\dagger\prn{-}, d}}
  \]
  Dually, $f$ is a pointwise right adjoint if $f^\dagger : \Modify[\OM]{C} \to \Modify[\OM]{D}$ is a
  pointwise left adjoint.
\end{defi}

Our main theorem relies on two crucial preliminary results. The first shows that any pointwise left
adjoint $f$ gives rise to a function in the other direction picking out the various (necessarily
unique) representing objects for $\Phi\prn{f^\dagger\prn{-}, d}$.
\begin{lem}
  If $\DeclVar{f}{\GM}{C \to D}$ is a pointwise left adjoint, then the type of morphisms
  $\DeclVar{g}{\GM}{D \to C}$ equipped with a natural isomorphism
  $\iota : \Phi\prn{f^\dagger\prn{-}, -} \cong \Yo \circ g$ is contractible.
\end{lem}
\begin{proof}
  Since $\Yo$ is an embedding, this type is a proposition. It therefore suffices to show that it is
  inhabited. By assumption, $\bar{g} = \Phi\prn{f^\dagger\prn{-}, d}$ is representable for all
  $\DeclVar{d}{\GM}{D}$, and thus it factors through $\PSH{C}_{\Con{isRepr}}$. Post-composing
  with the equivalence $\PSH{C}_{\Con{isRepr}} \Equiv C$ yields the desired $g : D \to C$.
\end{proof}

Using this, we prove a universal case of the theorem improving a pointwise adjoint to an adjoint:
every $\DeclVar{g}{\GM}{D \to C}$ that is a cartesian fibration~\parencite{buchholtz:2023}
such that the fiber
over every $\DeclVar{c}{\GM}{C}$ has an initial object~\parencite{bardomiano:2021} admits a left
adjoint.

\begin{lem}
  \label{lem:adj:functoriality-of-init-objects}
  If $\DeclVar{g}{\GM}{D \to C}$ is cartesian and for
  each $\DeclVar{c}{\GM}{C}$ the fiber $D_c$ has an initial object, then there exists $f : C \to D$
  and a path $p : \Prod{c : C} f\prn{g\prn{c}} = c$ such that $f\prn{c}$ is initial in $D_c$ for all $c : C$.
\end{lem}

\begin{proof}
  Note that $\Con{hasInitialObj}\prn{D_c}$ is a proposition and, therefore, by
  Axiom~\ref{ax:discrete-iff-crisp} we may assume $\Modify[\GM]{\Con{hasInitialObj}\prn{D_c}}$ holds
  for each $\DeclVar{c}{\GM}{C}$. With this observation to hand, we can show that $g$ is a pointwise
  right adjoint: if $\DeclVar{c}{\GM}{C}$, $\DeclVar{d}{}{D}$:
  \begin{align*}
    \Phi_C\prn{\MkMod[\OM]{c},g\prn{d}}
    &\Equiv \Hom{c}{g\prn{d}}
    \\
    &\Equiv \Hom{\ObjInit{D_c}}{d} && \text{$g$ is cartesian}
    \\
    &\Equiv \Phi_D\prn{\MkMod[\OM]{\ObjInit{D_c}},d}
  \end{align*}

  In this last step, we use our observation that $\Modify[\GM]{\Con{hasInitialObj}\prn{D_c}}$ and,
  crucially, that not only $\Con{hasInitialObj}\prn{D_c}$ holds. In particular, we rely on the fact
  that $\DeclVar{\ObjInit{D_c}}{\GM}{D_c}$.

  Accordingly, we obtain a function $\DeclVar{f}{\GM}{C \to D}$ which sends $\DeclVar{c}{\GM}{C}$ to
  $\ObjInit{D_c}$. It remains to show that $f\prn{c}$ is initial in $D_c$ for all $c : C$. Since $D
  = \Sum{c:C}D_c$, this amounts to the following map being an equivalence:
  \[
    \prn{\Sum{d : D} \Sum{x : D^\Int} x\,1 = d \times x\,0 = f\prn{g\prn{d}} \times g \circ x =
    \ArrId{g\prn{d}}}\to D
  \]
  This simultaneously shows that (1) $f\prn{c}$ lives in the expected fiber and (2) that it is
  initial in this fiber. To prove this, we use Theorem~\ref{thm:prelims:ff-ess-surj-to-equiv} which
  allows us to reduce to the $\GM$-annotated case, where the conclusion follows from the fact that
  $f\prn{c}$ is then initial in $D_c$.
\end{proof}

\begin{thm}
  \label{thm:adj:pw-to-full-adj}
  Pointwise right adjoints are right adjoints.
\end{thm}
\begin{proof}
  Given a map $\DeclVar{g}{\GM}{D \to C}$, consider the cartesian family
    \[
      \pi : \COMMA{C}{g} = \prn{\Sum{c : C}\Sum{d : D} \Hom{c}{g\prn{d}}} \to C
    \]
  Since $g$ is a
  pointwise right adjoint, each fiber of $\pi$ over $\DeclVar{c}{\GM}{C}$ has an initial object. We
  then apply Lemma~\ref{lem:adj:functoriality-of-init-objects} to obtain $\bar{f} : C \to
  \COMMA{C}{g}$ where---without loss of generality---we may assume that $\pi \circ \bar{f}$ is
  definitionally the identity.
  Finally, the composite $f = \Proj[2] \circ \bar{f}$ is the desired left adjoint to $g$:
  \begin{align*}
    &\Hom[C]{c}{g\prn{d}}
    \\
    &\quad\Equiv \Sum{\alpha : \Hom[C]{c}{g\prn{d}}}
      \Hom[\COMMA{C}{g}_c]{\bar{f}\prn{c}}{\prn{c, d, \alpha}}
    \\
    &\quad\Equiv
      \Sum{\alpha : \Hom[C]{c}{g\prn{d}}} \Sum{\beta : \Hom[D]{f\prn{c}}{d}}
      g\prn{\beta} \circ \Proj[3]\prn{\bar{f}\prn{c}} = \alpha
    \\
    &\quad\Equiv \Hom[D]{f\prn{c}}{d}
  \end{align*}
  The first step uses the initiality of $\bar{f}\prn{c}$ in the fiber over $c$, the second
  unfolds the definition of a morphism in $\COMMA{C}{g}$, and the final step is a standard
  contractibility of singletons argument.
\end{proof}

\subsection{Examples of adjunctions}

We take advantage of Theorem~\ref{thm:adj:pw-to-full-adj} to produce vital examples of adjoints. The
most important is the following:

\begin{thm}
  \label{thm:adj:precomp-is-right-adj} If $\DeclVar{f}{\GM}{D \to C}$ and $D$ is small, then
  $\PSH{f} \defeq (f^\dagger)^* : \PSH{C} \to \PSH{D}$ is a right adjoint with left adjoint $f_!$.
\end{thm}
\begin{proof}
  For notational simplicity, we replace $C$ and $D$ with $\Modify[\OM]{C}$ and $\Modify[\OM]{D}$. By
  Theorem~\ref{thm:adj:pw-to-full-adj}, it suffices to assume $\DeclVar{X}{\GM}{D \to \Space}$ and
  to construct $f_!\prn{X} : C \to \Space$ along with a natural bijection $\Prod{Y}
  \Hom{f_!\prn{X}}{Y} \Equiv \Hom{X}{f^*\prn{Y}}$. This is an immediate consequence of
  \textcite[Theorem 2.41]{rijke:2020} after taking the localization of the composite
  $\Sum{c : C} X\prn{c} \to C \to D$ against the map $\brc{0} \to \Int$.
\end{proof}

\begin{cor}
  The left adjoint $f_! \Adjoint \PSH{f}$ satisfies $f_! \circ \Yo \cong \Yo \circ f$.
\end{cor}

\begin{cor}
  $\Space$ is small cocomplete: $\Const : \Space \to \Space^C$ is a right adjoint with left
  adjoint $\Colim{}$ for small categories $\DeclVar{C}{\GM}{\Uni}$. Explicitly, if
  $\DeclVar{X}{\GM}{C \to \Space}$, then $\Colim{C} X = \Grpdify\prn{\Sum{c : C} X\prn{c}}$.
\end{cor}
\begin{rem}
  One can prove $\Space$ is complete (that $\Const \Adjoint \Lim{}$) by a result of
  \textcite{gratzer:2024}. In particular, they show $\Prod{c : C} X\prn{c} : \Space$ whenever
  $X : C \to \Space$ and $\DeclVar{C}{\GM}{\Uni}$. Corollary~\ref{cor:prelims:dua} and Lemma~\ref{lem:tw:hom-of-psh} then
  imply that $\Hom{A}{\Prod{c : C} X\prn{c}} \Equiv \Hom{\Con{const}\,A}{X}$.
\end{rem}

\begin{lem}
  \label{lem:adj:eval-at-pt-is-radj}
  If $\DeclVar{c}{\GM}{C}$ then $\PSH{c} : \PSH{C} \to \Space$ is a left adjoint.
\end{lem}
\begin{proof}
  For simplicity, we once more replace $C$ with $\Modify[\OM]{C}$ and invoke
  Theorem~\ref{thm:adj:pw-to-full-adj}. We then assume that we are given $\DeclVar{X}{\GM}{\Space}$ and
  define $c_*X$ to be $\lambda c'.\,X^{\Hom{c'}{c}}$. This is covariant in $c'$
  \parencite{gratzer:2024}. It then suffices to show that the map $X^{\Hom{c}{c}} \to X$ given by
  evaluation at the identity map induces an equivalence
  $\Hom[C \to \Space]{Y}{c_*X} \Equiv \Hom[\Space]{Y\prn{c}}{X}$. This, in turn, is a consequence of
  the fact that the map $Y\prn{c} \to \Grpdify\prn{\Sum{c' : C} \Hom{c'}{c} \times Y\prn{c'}}$ is an
  equivalence, so we may decompose the evaluation map as follows:
  \begin{align*}
    &\Hom[C \to \Space]{Y}{c_*X}
    \\
    &\Equiv \Prod{c' : C} Y\prn{c'} \to X^{\Hom{c'}{c}}
    \\
    &\Equiv \Prod{c' : C} \Hom{c'}{c} \times Y\prn{c'} \to X
    \\
    &\Equiv \Grpdify \prn{\Sum{c' : C} \Hom{c'}{c} \times Y\prn{c'}} \to X
    \\
    &\Equiv Y\prn{c} \to X \qedhere
  \end{align*}
\end{proof}

Combining Theorem~\ref{thm:adj:precomp-is-right-adj} and Lemma~\ref{lem:adj:eval-at-pt-is-radj}, we obtain the
following characterization of $f_!X$:

\begin{cor}
  \label{cor:adj:push-pull-formula}
  If $\DeclVar{X}{\GM}{\PSH{C}}$, $\DeclVar{f}{\GM}{C \to D}$, and $\DeclVar{d}{\GM}{D}$ then we may
  explicitly identify
  $\prn{f_!X}\,d$ as the following type:
  \[
    \Grpdify \prn{\Sum{c : \Modify[\OM]{C}} X\prn{c} \times \Hom{f^\dagger c}{d}}
  \]

  Moreover, the counit map $\epsilon : \Hom{f_!f^*Y}{Y}$ is equivalent at $d$ to the
  map assigning $\eta\prn{c,x,\alpha}$ to $Y\prn{\alpha}\prn{x}$---the canonical map induced by
  viewing the above expression as a colimit of a functor $\SLICE{\Modify[\OM]{C}}{d} \to \Space$
  and equipping $Y$ with the canonical cocone.
\end{cor}
\begin{proof}
  We first observe that $\PSH{d}f_!\prn{X} = \prn{f_!X}\,d$. Transposing, we have
  $\Hom{\PSH{d}f_!X}{Z} \Equiv \Hom{X}{\PSH{f}d_*Z}$ for every $\DeclVar{Z}{\GM}{\Space}$. We
  calculate $\Hom{X}{\PSH{f}d_*Z}$ using the definition of $d_*$ provided above:
  \begin{align*}
    &\Hom[\PSH{C}]{X}{f^*d_*Z}
    \\
    &\Equiv \Prod{\DeclVar{c}{}{\Modify[\OM]{C}}} X\prn{c} \to \Hom{f^\dagger c}{d} \to Z
    \\
    &\Equiv \prn{\Sum{\DeclVar{c}{}{\Modify[\OM]{C}}} X\prn{c} \times \Hom{f^\dagger c}{d}} \to Z
  \end{align*}
  Therefore $\PSH{d}f_!X$ satisfies the universal property of
  $\Grpdify \prn{\Sum{\DeclVar{c}{}{\Modify[\OM]{C}}} X\prn{c} \times \Hom{f^\dagger c}{d}}$

  The second claim then follows by unfolding the explicit description of the transportation
  equivalences given above.
\end{proof}

\begin{rem}
  Ideally, a fully functorial description of $f_!\prn{X}$ directly and without the indirection of
  the above approach. One path to such a description is to use the recently-constructed category of
  categories in simplicial type theory~\cite{gratzer:2026}. One could then use the explicit
  description of the free cocartesian fibration~\cite{buchholtz:2023} along with the left adjoint
  $\vrt{-} : \Cat \to \Space$ to define $f_!$ as follows:
  \[
    f_!\prn{X}\prn{d} = \vrt{\Sum{c : \Modify[\OM]{C}} X\prn{c} \times \Hom{f^\dagger d}{c}}
  \]
  It is unclear to us whether $\Grpdify = \vrt{-}$ in this particular case (\ie{}, whether our
  result above could be strengthened to remove $\GM$-annotations) and we leave this question for
  future work.
\end{rem}

\begin{rem}\label{rem:adj:right-kan}
  Far less effort is required for the case of right Kan extensions $f^* \Adjoint f_*$. In this case,
  we observe that if $X : C \to \Space$ then $\lambda d.\,\Prod{c : C} f\,c = d \to X\prn{c}$
  induces a map $D \to \Space$ without any $\GM$-annotations on $X$. One may directly check that it
  is satisfies the expected transposition equivalence and thus serves as a right adjoint.
\end{rem}

The following lemma does not require Theorem~\ref{thm:adj:pw-to-full-adj}, but is merely a consequence of
manipulating natural transformations:
\begin{lem}
  \label{lem:adj:comp}
  If $f : C \to D$ is an adjoint so is $f_* : C^A \to D^A$.
\end{lem}

\begin{cor}
  \label{cor:adj:limits-are-pointwise} If $C$ is (co)complete so is $C^D$ and (co)limits are
  computed pointwise. In particular, $\PSH{C}$ is (co)complete.
\end{cor}
\begin{cor}
  \label{cor:adj:yo-preserves-limits} The Yoneda embedding preserves all limits.
\end{cor}
\begin{proof}
  If $\DeclVar{F}{\GM}{I \to C}$ and $\Lim{} F$ exists, then functoriality of $\Yo$ induces a map
  $\Yo{\Lim{} F} \to \Lim{}\prn{\Yo \circ F} = \lambda c_0.\,\Prod{c : C}\,\Phi\prn{c_0,c}$, so it
  suffices to check that this map is invertible at all $\DeclVar{c}{\GM}{C}$. Unfolding, we must
  argue that $\Hom{c_0}{\Lim{} F} \to \Prod{c : C} \Hom{c_0}{F\,c}$ is an equivalence, but this
  is immediate by the definition of limit.
\end{proof}

Finally, we show the full subcategories $\Space_{\le n}$ of $\Space$ defined by $n$-truncated types
form reflective subcategories of $\Space$. The idea is simple: use the truncation HITs. However, it
is not automatic that they restrict to $\Trunc{-}_{n} : \Space \to \Space_{\le n}$. We prove this
alongside with the reflectivity of $\Space_{\le n}$ using Theorem~\ref{thm:adj:pw-to-full-adj}:
\begin{cor}
  The inclusion $\Space_{\le n} \to \Space$ is a right adjoint.
\end{cor}
\begin{cor}
  $\Space_{\le n}$ is (co)complete.
\end{cor}
The same methodology applies to the subcategory of modal types associated to an idempotent
monad~\parencite{rijke:2020}.

\begin{exa}[Isbell conjugation]
  If $\DeclVar{C}{\GM}{\Uni}$, then the Isbell conjugation map $\phi$ is a left adjoint:
  \begin{align*}
    &\phi : \PSH{C} \to \Modify[\OM]{C \to \Space}
    \\
    &\phi\prn{X} = \MkMod[\OM]{\lambda c.\,\Phi\prn{X,\Yo{c}}}
  \end{align*}
\end{exa}

\subsection{The universal property of presheaf categories}
\label{sec:adj:psh}

Next, we generalize Theorem~\ref{thm:adj:precomp-is-right-adj} to show that if $\DeclVar{f}{\GM}{C \to E}$
where $C$ is a small category and $E$ is a cocomplete category, then
$\Phi\prn{f^\dagger\prn{-},-} : E \to \PSH{C}$ is a right adjoint loosely following the argument
given by \textcite{cisinski:2019}. We begin with a few general lemmas. In what follows, fix $C$
and $E$ as above.

First, as a corollary of the proof of Theorem~\ref{thm:adj:precomp-is-right-adj}:
\begin{lem}
  The colimit of $\Yo : C \to \PSH{C}$ is $\ObjTerm{\PSH{C}} = \lambda \_.\,\ObjTerm{}$.
\end{lem}

From the above, and further inspection of colimits, we are able to derive a result of independent
interest: Every presheaf is the colimit of representable presheaves.
\begin{lem}[Density of $\Yo$]
  \label{lem:adj:density}
  If $\DeclVar{X}{\GM}{\PSH{C}}$, then
  $X \cong \Colim{\Modify[\OM]{\TotalType{X}}} \Yo \circ \pi^\dagger$,
  where $\TotalType X=\Sum{c:\Modify[\OM]C}X\prn c$.
\end{lem}
\begin{proof}
  We begin with the following computation where $\pi : \TotalType{X} \to \Modify[\OM]{C}$
  and $\pi^\dagger_! : \Space^{\TotalType{X}} \to \PSH C$:
  \[
    \pi^\dagger_!\ObjTerm{}
    \cong
    \pi^\dagger_!\prn{\Colim{\Modify[\OM]{\TotalType{X}}} \Yo}
    \cong
    \Colim{\Modify[\OM]{\TotalType{X}}} \pi^\dagger_! \circ \Yo
    \cong
    \Colim{\Modify[\OM]{\TotalType{X}}} \Yo \circ \pi^\dagger
  \]
  We have used the fact that $\pi^\dagger_!$, a left adjoint, commutes with
  colimits~\parencite{bardomiano:2021}.
  To show $\pi^\dagger_!\ObjTerm{} \cong X$, we note that for all $Z : \PSH{C}$:
  \begin{align*}
    \Hom{\pi^\dagger_!\ObjTerm{}}{Z}
    &\Equiv \Hom[{\TotalType{X}} \to \Space]{\ObjTerm{}}{Z \circ \pi}
    \\
    &\Equiv \Prod{\prn{c, x} : \Sum{c : \Modify[\OM]{C}} X\prn{c}} Z\prn{c}
    \\
    &\Equiv \Prod{c : \Modify[\OM]{C}} X\prn{c} \to Z\prn{c}
    \\
    &\Equiv \Hom{X}{Z}
  \end{align*}
  The conclusion now follows from the Yoneda lemma.
\end{proof}

\begin{lem}
  $\Nv_f = \Phi\prn{f^\dagger\prn{-},-} : E \to \PSH{C}$ is a right adjoint
  with left adjoint $f_! \Adjoint \Nv_f$.
\end{lem}
\begin{proof}
  We will prove that $\Nv_f$ is a pointwise right adjoint. Accordingly, fixing
  $\DeclVar{X}{\GM}{\PSH{C}}$ we must construct $\DeclVar{e}{}{\Modify[\OM]{E}}$ such that
  $\Phi\prn{e,-} \cong \Phi\prn{\MkMod[\OM]{X},\Nv_f\prn{-}}$. Since
  $X \cong \Colim{\Modify[\OM]{\TotalType{X}}} \Yo \circ \pi^\dagger$ and $E$ is cocomplete, by the
  dual of Corollary~\ref{cor:adj:yo-preserves-limits} it suffices to assume
  $\MkMod[\OM]{X} = \Yo^\dagger\prn{c}$ with $\DeclVar{c}{}{\Modify[\OM]{C}}$.\footnote{Note the
    lack of $\GM$-annotation here: we must ensure that we are functorial in $c$ in order to obtain a
    \emph{diagram} in $E$.} Finally, take $e = f^\dagger\prn{c}$ and
  $\Phi\prn{f^\dagger\prn{c},-} \cong \Phi\prn{\Yo^\dagger\prn{c},\Nv_f\prn{-}}$ by
  Theorem~\ref{thm:tw:strong-yoneda}.
\end{proof}

We are now able to prove, as promised, the universal property of $\PSH{C}$. If we write
$\CoCont{\PSH{C}}{E}$ for the full subcategory of $\PSH{C} \to E$ spanned by functors preserving
all colimits, then $\Yo^* : \CoCont{\PSH{C}}{E} \to \prn{C \to E}$ is an equivalence. To prove this,
we essentially argue that there is a map sending $f$ to the left adjoint to $\Nv_f$ and that this is
the inverse to $\Yo^*$.

\begin{restatable}{thm}{universalpropofpsh}
  \label{thm:adj:cc}
  $\Yo^* : \CoCont{\PSH{C}}{E} \to \prn{C \to E}$ is an equivalence.
\end{restatable}
\begin{proof}
  We use Theorem~\ref{thm:prelims:ff-ess-surj-to-equiv}. If $\DeclVar{f}{\GM}{C \to E}$, then
  $f_! : \PSH{C} \to E$ satisfies $f_! \circ \Yo = f$ and so $\Yo^*$ is essentially surjective:
  \[
    \Phi\prn{\prn{f_!\circ \Yo}^\dagger\prn{-},-} \cong \Phi\prn{\Yo^\dagger\prn{-},\Nv_f\prn{-}}
    \cong \Nv_f = \Phi\prn{f^\dagger\prn{-},-}
  \]

  Next, we argue that if $\DeclVar{F}{\GM}{\CoCont{\PSH{C}}{E}}$, then $\prn{F\circ\Yo}_! \cong F$,
  so that $\Yo^*$ is a bijection on $\GM$-elements. Let us write $f = F \circ \Yo$. We first
  construct a map $\alpha : \Hom{f_!}{F}$, check that it is an isomorphism on representable functors
  $\Yo{c}$, and conclude that it is an isomorphism in general as both $f_!$ and $F$ preserve
  colimits. To construct $\alpha$, we observe that by transposition it suffices to construct a map
  $\Hom[\PSH{C} \to \PSH{C}]{\ArrId}{\Nv_f\prn{F\prn{-}}}$. After uncurrying, it suffices to
  construct a map $\Hom[\Modify[\OM]{C} \times \PSH{C} \to \Space]{\epsilon}{\Phi\prn{f^\dagger{-},F\prn{-}}}$
  where $\epsilon$ denotes evaluation.
  In light of the identification $\epsilon \cong \Phi\prn{\Yo^\dagger\prn{-},-}$ given in
  Theorem~\ref{thm:tw:strong-yoneda}, we finally choose $\Phi_F$.

  To check that this induces an isomorphism at $\Yo{c}$ where $\DeclVar{c}{\GM}{C}$, we first note
  that the $f_!\prn{\Yo{c}} = f\prn{c}$ and under this identification the unit map
  $\eta : \Hom{\Yo{c}}{\Hom{f^\dagger\prn{-}}{f\prn{c}}}$ is given by $\Phi_f$. 
  Moreover, the following diagram commutes up to homotopy by definition of transposition $f_! \Adjoint \Nv_f$:
  \[
    \DiagramSquare{
      width = 6cm,
      nw = {\Yo{c}},
      ne = {\Phi\prn{\Yo^\dagger\prn{-},\Yo{c}}},
      sw = {\Nv_f\prn{f_!\prn{\Yo{c}}}},
      se = {\Nv_f\prn{F\prn{\Yo{c}}}},
      north = \Phi_{\Yo},
      west = \eta,
      east = \Phi_{F},
      south = \Nv_f\prn{\alpha_{\Yo{c}}},
    }
  \]
  Unraveling the top series of maps, we see that it sends $\DeclVar{f}{\GM}{\Hom{d}{c} \cong \Phi\prn{d,c}}$
  to $F\prn{f} : \Hom{f\prn{d}}{f\prn{c}}$ and, consequently, that $\alpha_{\Yo{c}}$ is homotopic to
  the identification $f_!\prn{\Yo{c}} \cong f\prn{c}$ and, in particular, an isomorphism.

  Finally, we now show that $\Yo^*$ is fully faithful. To show that it is fully faithful, we must
  show that if $\DeclVar{f,g}{\GM}{C \to E}$, then $\Hom{f_!}{g_!} \cong \Hom{f}{g}$. Both sides are
  groupoids, so it suffices to consider $\GM$-annotated elements. If
  $\DeclVar{\alpha}{\GM}{\Hom{f}{g}}$, then by transposing we may regard $\alpha$ as an element of
  $\Modify[\GM]{C \to E^\Int}$ and the previous observation ensures that this type is equivalent to
  $\Modify[\GM]{\CoCont{\PSH{C}}{E^\Int}}$ which yields the desired conclusion after transposing.
\end{proof}

\section{The theory of Kan extensions}
\label{sec:kan}
A unifying concept in category theory are \emph{Kan extensions}, which are universal extensions of
functors along functors on the same domain. Mac Lane, one of the founders of category theory,
famously stated: ``The notion of Kan extensions subsumes all the other fundamental concepts of
category theory,'' such as (co)limits and adjunctions~\parencite{maclane:1978,riehl:2014}.

\begin{defi}[Kan extensions]
  Given a map $f : C \to D$ and a category $E$, the left (right) Kan extension $\Lan{f}$
  ($\Ran{f}$) is the left (right) adjoint to $f^* : E^D \to E^C$.
\end{defi}

While the definition makes sense in general, to use the results of the previous sections, we shall
assume $\DeclVar{f}{\GM}{C \to D}$ and $\DeclVar{E}{\GM}{\Uni}$.  In Section~\ref{sec:kan:existence} we
show that Kan extensions exist whenever $E$ is (co)complete and in
Sections~\ref{sec:kan:cofinal} and \ref{sec:kan:quillen-a} we put this to work by deducing two important results:
Quillen's theorem A and the properness of cocartesian fibrations. Our arguments for the existence of
Kan extensions and Quillen's theorem A adapt the (model-agnostic) $\infty$-categorical arguments of
\textcite{ramzi:bousfield-kan}.

\subsection{Existence and characterization of Kan extensions}
\label{sec:kan:existence}

We can prove that Kan extensions can be computed in an expected way. For $\DeclVar{d}{}{D}$, we
write $\SLICE{C}{d} \defeq C \times_D \SLICE{D}{d}$ and
$\COSLICE{C}{d} \defeq C \times_D \COSLICE{D}{d}$.
We assume that $C$ and $D$ are both small so each $\SLICE Cd$ is also small.
By Theorem~\ref{thm:adj:precomp-is-right-adj} and Lemmas \ref{lem:adj:comp} and
\ref{cor:adj:push-pull-formula}:

\begin{lem}
  \label{lem:kan:lan-exists-for-psh}
  If $E = \PSH{A}$ for some category $\DeclVar A\GM\Uni$, then $\Lan{f}$ exists.
  Moreover, if $\DeclVar{X}{\GM}{C \to E}$ and $\DeclVar{d}{\GM}{D}$, then
  $\Lan{f}\,X\,d = \Colim{} \prn{\SLICE{C}{d} \to C \to E}$, that is,
  $\Grpdify\prn{\Sum{\prn{c,\_} : \SLICE{C}{d}} X\prn{c}}$ in the category of spaces.
\end{lem}

By the second half of Lemma \ref{cor:adj:push-pull-formula}, the counit map
$\epsilon : \Hom{\Lan{f} f^*Y}{Y}$ specializes at $d$ to become the map $\Hom{\Lan{f}\,Y\,d}{Y}$
induced by the cocone $\SLICE{C}{d} \to \SLICE{D}{d} \to D \to E$ where the final arrow is $Y$.
This yields more generally:
\begin{thm}
  \label{thm:kan:left-kan-exists}
  If $E$ is cocomplete, then $\Lan{f}$ exists, and if $\DeclVar{X}{\GM}{C \to E}$,
  $\DeclVar{d}{\GM}{D}$, then $\Lan{f}\,X\,d = \Colim{} \prn{\SLICE{C}{d} \to C \to E}$
  with a similar characterization of the counit.
\end{thm}
\begin{proof}
  It suffices to argue that precomposition is pointwise a right adjoint and so we fix
  $\DeclVar{X}{\GM}{C \to E}$. By Theorem~\ref{thm:adj:cc}, we may view $X$ as the composition
  $\bar X \circ \Yo$, where $\bar X : \PSH{C} \to E$ is the left adjoint to $\Nv_X$.
  Next, we observe by~Lemma~\ref{lem:kan:lan-exists-for-psh} that $\Yo: C \to \PSH C$
  admits an extension to $D$ along $f$, namely $\Lan{f} \Yo : D \to \PSH C$,
  and we claim that
  $\bar{X} \circ \Lan{f} \Yo$ is our desired extension of $f$.
  Fixing $Z : D \to E$, we calculate:
  \begin{align*}
  \Hom[D \to E]{\bar{X} \circ \Lan{f} \Yo}{Z}
  &\Equiv \Hom[D \to \PSH{C}]{\Lan{f} \Yo}{\Nv_X \circ Z}
  \\*
  &\Equiv \Hom[C \to \PSH{C}]{\Yo}{\Nv_X \circ Z \circ f}
  \\*
  &\Equiv \Hom[C \to E]{\bar X\circ\Yo}{Z \circ f}
  \\*
  &= \Hom[C \to E]{X}{Z \circ f}
\end{align*}
  The expected colimit formula continues to hold as a consequence of
  Lemma~\ref{lem:kan:lan-exists-for-psh} and the cocontinuity of $\bar{X}$.
\end{proof}

By duality, we obtain the following variant:
\begin{thm}
  \label{thm:kan:right-kan-exists}
  If $E$ is complete, then $\Ran{f}$ exists and is specified by the dual limit formula:
  $\Ran{f}\,X\,d = \Lim{} \prn{\COSLICE{C}{d} \to C \to E}$ and a dual characterization
  of the limit.
\end{thm}

\begin{rem}
  Given that limits are easier to describe in $\Space$ than colimits---the latter involving
  $\Grpdify$ and the former needing only $\Prod{}$-types, we can often give more explicit functorial
  descriptions of $\Ran{f}$. For instance, if $\DeclVar{f}{\GM}{C \to D}$ we may directly verify that
  $\Ran{f}\prn{\Const\,A} = \lambda d.\,\Prod{c : \COSLICE{C}{d}} A$ by showing the latter (1)
  indeed a presheaf valued in $\Space$ and (2) satisfies the correct universal property.
\end{rem}

\subsection{Final and initial functors}
\label{sec:kan:cofinal}

It is frequently useful to show that the limit of a complex diagram $D$ can be calculated by first
restricting to a simpler diagram using $f : C \to D$ and calculating the limit there \eg{},
restricting from $\mathbb{Z}$ to $\mathbb{Z}_{\le 0}$. When this approach is valid, $f$ is said to
be initial:
\begin{defi}
  A functor $\DeclVar{f}{\GM}{C \to D}$ is \emph{initial} if for every
  $\DeclVar{X}{\GM}{D \to \Space}$ the map $\Lim{D} X \to \Lim{C} X \circ f$ is an
  equivalence. A map is \emph{final} if its opposite is initial.
\end{defi}
While this definition is asymmetrical in its treatment of initiality and finality,
we shall restore the symmetry as a consequence of Quillen's Theorem~A in the next section,
see~Corollary~\ref{cor:kan:right-cofinal-char}.

Recall that $\Lim{D} X = \Prod{d : D} X\prn{d}$ and so the definition of
initiality equivalently states that the
restriction map $\prn{\Prod{d : D} X\prn{d}} \to \prn{\Prod{c : C} X\prn{f\prn{c}}}$ is an
equivalence.

\begin{exa}
  The $\brc{0}$/$\brc{1}$ inclusion ${\Unit \to \Int}$ is initial/final.
\end{exa}

\begin{lem}
  \label{lem:kan:left-cofinal-preserves-limit}
  If $\DeclVar{f}{\GM}{C \to D}$ is initial and
  $\DeclVar{X}{\GM}{D \to E}$, then $\Lim{C} \prn{X \circ f}$ and $\Lim{D} X$ both exist
  whenever either exists and are canonically isomorphic.
\end{lem}
\begin{proof}
  By Corollary~\ref{cor:adj:yo-preserves-limits}, we replace $E$ with $\PSH{E}$ and by
  Corollary~\ref{cor:adj:limits-are-pointwise} we reduce to $\Space$ where the result is immediate.
\end{proof}

\begin{restatable}{lem}{locop}
  \label{lem:app:localization-op}
  If $\DeclVar{C}{\GM}{\Uni}$, then $\Grpdify{C} \Equiv \Grpdify{\Modify[\OM]{C}}$.
\end{restatable}
\begin{proof}
  We observe that $\Grpdify{C} \Equiv \Modify[\GM]{\Grpdify{C}}$ and likewise for
  $\Modify[\OM]{C}$. Accordingly, we note that:
  \begin{gather*}
    \Modify[\GM]{
      \Modify[\OM]{C} \to \Modify[\GM]{X}
    }
    \Equiv
    \Modify[\GM]{
      \Grpdify{\Modify[\OM]{C}} \to \Modify[\GM]{X}
    }
    \\
    \Modify[\GM]{
      C \to \Modify[\GM]{X}
    }
    \Equiv
    \Modify[\GM]{
      \Grpdify{C} \to \Modify[\GM]{X}
    }
    \\
    \Modify[\GM]{
      C \to \Modify[\GM]{X}
    }
    \Equiv
    \Modify[\GM]{
      \Modify[\OM]{C} \to \Modify[\GM]{X}
    }
  \end{gather*}
  Finally, the result follows from a simple Yoneda argument.
\end{proof}

\begin{lem}
  \label{lem:kan:localization-is-cofinal}
  For every $C$, the canonical map $C \to \Grpdify{C}$ is both initial and final.
\end{lem}
\begin{proof}
  By Lemma~\ref{lem:app:localization-op}, it suffices to argue that this map is initial. To this
  end, we must show the following map to be an equivalence for every $X : \Grpdify{C} \to \Space$:
  \[
    \prn{\Prod{d : \Grpdify{C}} X\prn{d}} \to \prn{\Prod{c : C} X\prn{\eta\prn{c}}}
  \]
  However, $X\prn{d}$ is discrete for every $d : \Grpdify{C}$ and so this is simply the universal
  property of $\Grpdify$.
\end{proof}
\begin{cor}
  \label{lem:kan:limit-over-contractible}
  If $\Grpdify{C} = \ObjTerm{}$, then $\Lim{C} A = A$ for $A : \Space$.
\end{cor}

\subsection{Quillen's Theorem A}
\label{sec:kan:quillen-a}

Our next goal is to prove the $\infty$-categorical version of Quillen's theorem A. Unlike
traditional proofs, we follow \textcite{ramzi:bousfield-kan} and rely on having already established
the basic apparatus of Kan extensions to simplify our argument.

\begin{defi}
  A functor $\DeclVar{f}{\GM}{C \to D}$ is Quillen final if
  $\Grpdify(\COSLICE{C}{d}) \Equiv \Unit$ for all $\DeclVar{d}{\GM}{D}$
\end{defi}

Our main result of this section shall be the following, though we shall only arrive at its proof at
some preliminary work:
\begin{thm}[Quillen's Theorem A]
  \label{thm:kan:theorem-a}
  A functor $\DeclVar{f}{\GM}{C \to D}$ is final if and only if
  it is Quillen final.
\end{thm}

\begin{rem}
  This result shows that, in particular, finality doesn't depend on the particular universe $\Space$
  chosen.
\end{rem}

\begin{lem}
  \label{lem:kan:theorem-a-for-psh}
  If $\DeclVar{f}{\GM}{C \to D}$ is Quillen final and $\DeclVar{X}{\GM}{D \to \PSH{A}}$,
  then $\Colim D X \Equiv \Colim C X \circ f$.
\end{lem}
\begin{proof}
  This statement is pointwise, so we quickly reduce to $\Space$ instead of $\PSH{A}$. In this
  situation, we wish to show that the following commutes:
  \[
    \begin{tikzcd}
      \Space^D \ar[rr,"f^*"]\ar[dr,"\Colim D"'] && \Space^C\ar[dl,"\Colim C"] \\
      & \Space &
    \end{tikzcd}
  \]
  Note that all three morphisms are left adjoints,
  and so it suffices to compare their right adjoints:
  the constant functors $\Delta_C$ and $\Delta_D$,
  along with the right Kan extension $\Ran f$.
  We next note that there is at least a comparison map
  $\Delta_D \to \Ran{f} \circ \Delta_C$ given by transposing the
  identity map  $f^* \circ \Delta_D \to \Delta_C$.
  We must argue that this map is pointwise invertible,
  and so we reduce to considering $\DeclVar{X}{\GM}{\Space}$
  and $\DeclVar{d}{\GM}{D}$,
  and we must show the following, using~Theorem~\ref{thm:kan:right-kan-exists}
  (or the explicit construction of Remark~\ref{rem:adj:right-kan}):
  $X \Equiv \Lim{\COSLICE{C}{d}} X$.
  This now follows from our assumption and~Lemma~\ref{lem:kan:limit-over-contractible}.
\end{proof}

\begin{lem}
  \label{lem:kan:theorem-a}
  If $\DeclVar{f}{\GM}{C \to D}$ is Quillen final, $E$ a cocomplete category, and
  $\DeclVar{X}{\GM}{D \to E}$, then $\Colim D X \Equiv \Colim C X \circ f$.
\end{lem}
\begin{proof}
  We reduce to the case where $E = \PSH{D}$ (and therefore Lemma~\ref{lem:kan:theorem-a-for-psh}) by
  factoring $X$ as $\bar{X} \circ \Yo$ and noting that $\bar{X}$ preserves colimits by construction.
\end{proof}

\begin{proof}[Proof of Theorem~\ref{thm:kan:theorem-a}]
  To see that Quillen finality implies finality, we apply
  Lemma~\ref{lem:kan:theorem-a} to $\Modify[\OM]{\Space}$,
  and calculate:
  \[
    \Lim{\Modify[\OM]D} X
    \Equiv \Colim D X^\dagger
    \Equiv \Colim C X^\dagger \circ f
    \Equiv \Lim{\Modify[\OM]C} X \circ f^\dagger
  \]
  For the reverse, suppose that $f$ is final. We note that by the dual of
  Lemma~\ref{lem:kan:left-cofinal-preserves-limit} (again applied to $\Modify[\OM]{\Space}$), the
  canonical map $\Colim{C} X \circ f \to \Colim{D} X$ is an equivalence for any
  $\DeclVar{X}{\GM}{D \to \Space}$.  Fix $\DeclVar d\GM D$ and choose
  $X = \hom(d,-) = \Phi\prn{\MkMod[\OM]d,-}$ such that the colimits in question are precisely
  $\Grpdify{\COSLICE{D}{d}}$ and $\Grpdify{\COSLICE{C}{d}}$,
  using~Theorem~\ref{thm:adj:precomp-is-right-adj}.  This completes the proof since
  $\Grpdify{\COSLICE{D}{d}} = \ObjTerm{}$.
\end{proof}

\begin{cor}
  \label{cor:kan:right-cofinal-char}
  A functor $\DeclVar{f}{\GM}{C \to D}$ is final if and only if, for every
  $\DeclVar{X}{\GM}{D \to \Space}$ the map $\Colim{D} X \to \Colim{C} X \circ f$
  is an equivalence.
\end{cor}

This restores the symmetry between initial and final functors, as promised.  We offer another
symmetric definition of initiality and finality, informed by \textcite[Ch.~8]{cisinski:2024}.
\begin{defi}[Covariant equivalences]
  Fix $\DeclVar{p}{\GM}{C \to A}$ and $\DeclVar{q}{\GM}{D \to A}$ between categories $A,C,D$. Let
  $\DeclVar{f}{\GM}{C \to D}$ be a fibered map as follows:
  \[
    \begin{tikzpicture}[diagram]
      \node (C) {$C$};
      \node [right = 4cm of C] (D) {$D$};
      \node [below right = 2cm and 2cm of C] (A) {$A$};
      \path[->] (C) edge node[above] {$f$} (D);
      \path[->] (C) edge node[below left] {$p$} (A);
      \path[->] (D) edge node[below right] {$q$} (A);
    \end{tikzpicture}
  \]
  We call $f$ a \emph{covariant equivalence (between $p$ and $q$)} if for all families $\DeclVar{X}{\GM}{A \to \Space}$
  reindexing gives rise to an equivalence, \ie,:
  \[
    f^* : \prn{\DelimMin{1}\Prod{a:A} D_a \to X_a} \to \prn{\DelimMin{1}\Prod{a:A} C_a \to X_a}
  \]
  Dually, $f$ is called \emph{contravariant equivalence} precomposition with respect to all contravariant families is an equivalence.
\end{defi}

\begin{lem}\label{lem:cov-eq-comp}
  Let $f$ as below be a covariant between $p$ and $q$. For any functor
  $\DeclVar{r}{\GM}{A \to B}$, $f$ is also a covariant equivalence between $r \circ p$ and
  $r \circ q$:
  \[
    \begin{tikzpicture}[diagram]
      \node (C) {$C$};
      \node [right = 4cm of C] (D) {$D$};
      \node [below right = 2cm and 2cm of C] (A) {$A$};
      \node [below = of A] (B) {$B$};
      \path[->] (C) edge node[above] {$f$} (D);
      \path[->] (C) edge node[below left] {$p$} (A);
      \path[->] (D) edge node[below right] {$q$} (A);
      \path[->] (A) edge node[right] {$r$} (B);
    \end{tikzpicture}
  \]
\end{lem}

\begin{proof}
  We get the following induced square:
  \[
    \begin{tikzpicture}[diagram]
      \node (NW) {$\big( \prod_{b:B} D_b \to X_b \big)$};
      \node [right = 5.5cm of NW] (NE) {$\big( \prod_{b:B} C_b \to X_b \big)$};
      \node [below = of NW] (SW) {$\big( \prod_{a:A} D_a \to X_{r(a)} \big)$};
      \node [below = of NE] (SE) {$\big( \prod_{a:A} C_a \to X_{r(a)} \big)$};
      \path[->] (NW) edge node[above] {$f^*$} node[below] {$\simeq$} (NE);
      \path[->] (NW) edge node[left] {$\simeq$} (SW);
      \path[->] (NE) edge node[right] {$\simeq$} (SE);
      \path[->] (SW) edge node[below] {$f^*$} (SE);
    \end{tikzpicture}
  \]
  The upper horizontal map is an equivalence by the preconditions. The goal is to show that the
  lower horizontal map is an equivalence, too. But this follows from $3$-for-$2$ for equivalences.
\end{proof}

\begin{lem}[Characterizations of initiality]
  \label{lem:kan:characterization-initiality}
  Let $\DeclVar{f}{\GM}{C \to D}$ be a functor. Then the following are equivalent:
  \begin{enumerate}[label=\textup{(\arabic*)}]
  \item\label{it:prop:char-left-cof-i} $f$ is initial.
  \item\label{it:prop:char-left-cof-ii} Let $\DeclVar{X}{\GM}{A \to \Space}$ be a family with
    associated left fibration $\DeclVar{\pi}{\GM}{\TotalType{X} \to A}$. Then any square of the
    following form has a filler $\overline{\varphi}$, uniquely up to homotopy:
    \[
      \begin{tikzpicture}
        \node (nw) {$C$};
        \node[below = 1.5cm of nw] (sw) {$D$};
        \node[right = 4cm of nw] (ne) {$\TotalType{X}$};
        \node[right = 4cm of sw] (se) {$A$};
        \path[->] (nw) edge node[above] {$\varphi$} (ne);
        \path[->] (nw) edge node[left] {$f$} (sw);
        \path[->] (ne) edge node[right] {$\pi$} (se);
        \path[->] (sw) edge node[below] {$\alpha$} (se);
        \path[->, unique exists] (sw) edge node[upright desc] {$\overline{\varphi}$} (ne);
      \end{tikzpicture}
    \]
  \item\label{it:prop:char-left-cof-iii} For any family $\DeclVar{X}{\GM}{A \to \Space}$ the
    following square is a pullback:
    \[
      \DiagramSquare{
        nw = \TotalType{X}^D,
        ne = \TotalType{X}^C,
        sw = A^D,
        se = A^C,
        nw/style = pullback,
        height = 1.5cm,
      }
    \]
  \item\label{it:prop:char-left-cof-iv} $f$ is a covariant equivalence with respect to any
    $\DeclVar{\alpha}{\GM}{D \to A}$.
  \end{enumerate}
  The analogous characterization holds for contravariant equivalences and final functors.
\end{lem}
\begin{proof}
  Conditions~\ref{it:prop:char-left-cof-ii} and~\ref{it:prop:char-left-cof-iii} are readily seen
  to be equivalent by commuting $\Prod{}$ and $\Sum{}$. Condition~\ref{it:prop:char-left-cof-iv}
  unfolds to the following: for any $\DeclVar{X}{\GM}{A \to \Space}$ reindexing along $f$ is an
  equivalence, namely
  \[
    f^* :
    \prn{\Prod{a:A} D_a \to X_a}
    \xrightarrow{\sim}
    \prn{\Prod{a:A} (\Sum{d:D_a} C_{a,d}) \to X_a}
  \]
  This, again, is readily seen to be equivalent to~\ref{it:prop:char-left-cof-ii}.

  We turn to the implication
  $\text{\ref{it:prop:char-left-cof-iv}} \implies \text{\ref{it:prop:char-left-cof-i}}$.
  But this is clear,
  since~\ref{it:prop:char-left-cof-i} says that $f$ is a covariant equivalence with respect to
  itself and $\ArrId{D}$.

  For the converse direction
  $\text{\ref{it:prop:char-left-cof-i}} \implies \text{\ref{it:prop:char-left-cof-iv}}$
  we use the insight
  just made together with Lemma~\ref{lem:cov-eq-comp}.
\end{proof}

The following alternative characterization is also often useful:

\begin{restatable}{lem}{charcofinal}
  \label{prop:kan:charcofinal}
  $\DeclVar{f}{\GM}{C \to D}$ is initial \textup(resp. final\textup) if and only if for every
  covariant \textup(resp. contravariant\textup) family $\DeclVar{\pi}{\GM}{X \to Y}$, $f$ is left
  orthogonal to $\pi$, \ie{}, $\IsEquiv\prn{X^D \to X^C \times_{Y^C} Y^D}$.
\end{restatable}
\begin{proof}
  Immediate by Proposition~\ref{it:prop:char-left-cof-iv} for the initial case
  and by duality and~Corollary~\ref{cor:kan:right-cofinal-char} for the final case.
\end{proof}

\begin{cor}
  \label{cor:initial-covariant-ofs}
  There is an orthogonal factorization system in the sense of \cite{rijke:2020} with the left class given by the initial functors and the right class by covariant fibrations.
\end{cor}

As another consequence we get the dual of~Theorem~\ref{thm:kan:theorem-a}:
\begin{cor}
  \label{thm:kan:theorem-a-dual}
  A functor $\DeclVar{f}{\GM}{C \to D}$ is initial if and only if
  $\Grpdify(\SLICE{C}{d}) \Equiv \Unit$ for all $\DeclVar{d}{\GM}{D}$
  \textup(\emph{Quillen initial}\textup).
\end{cor}

We demonstrate the utility of Theorem~\ref{thm:kan:theorem-a} by giving a new and far simpler
proof that cocartesian fibrations are \emph{proper}.
\begin{defi}
  \label{def:proper-maps}
  A functor $\DeclVar{\pi}{\GM}{E \to B}$ between categories is \emph{proper} if for all pullbacks
  (of $\GM$-functors) of the following form, $v$ is final if $u$ is final:
  \[
    \begin{tikzpicture}[diagram]
      \SpliceDiagramSquare{
        height = 1.35cm,
        ne = E,
        se = B,
        nw = E',
        sw = B',
        east = \pi,
        west = \pi',
        nw/style = pullback,
      }
      \node [pullback, left = of nw] (A) {$E''$};
      \node [left = of sw] (B) {$B''$};
      \path[->] (A) edge node[above] {$v$} (nw);
      \path[->] (B) edge node[below] {$u$} (sw);
      \path[->] (A) edge (B);
    \end{tikzpicture}
  \]
  We call $\pi$ \emph{smooth} if $\pi^\dagger : \Modify[\OM]{E} \to \Modify[\OM]{B}$ is proper.
\end{defi}
\begin{lem}
  Smooth and proper functors are closed under composition and pullback.\footnote{The definition of
    properness is formulated specifically to bake in the latter.}
\end{lem}
\begin{thm}
  \label{thm:kan:cocartesian-is-proper}
  Cartesian fibrations are smooth and cocartesian fibrations are proper.
\end{thm}
\begin{proof}
  It suffices to treat the proper case. Fix a cocartesian fibration $\DeclVar{\pi}{\GM}{E \to B}$
  and note that since cocartesian fibrations are stable under pullbacks, it suffices to show that $v$ is
  final in the following pullback diagram if $u$ is final:
  \[
    \DiagramSquare{
      height = 1.5cm,
      width = 3cm,
      nw = { A \times_B E},
      sw = A,
      ne = E,
      se = B,
      nw/style = pullback,
      north = v,
      south = u,
      east = \pi,
    }
  \]
  We now use Theorem~\ref{thm:kan:theorem-a}. For $\DeclVar{e}{\GM}{E}$ we compute the fiber:
  \begin{align*}
    & (A \times_B E) \times_E \COSLICE{E}{e}
    \\
    &\quad
    \Equiv A \times_B \COSLICE{E}{e}
    \\
    &\quad
    \Equiv A \times_B \prn{\Sum{\DeclVar{b'}{}{B},\DeclVar{f}{}{\Hom{\pi\prn{e}}{b'}}}{\prn{E_{b'}}^{\Int}}}
    && \text{$\pi$ is cocartesian}
    \\
    &\quad
    \Equiv \Sum{\DeclVar{(a,f)}{}{A \times_B \COSLICE{B}{\pi\prn{e}}}}{\COSLICE{\prn{E_{u\prn{a}}}}{f_!e}}
  \end{align*}
  Applying $\Grpdify$ to each fiber yields $\Grpdify{\COSLICE{\prn{E_{u\prn{a}}}}{f_!e}} \Equiv \Unit$ (as
  coslices have initial elements) and $\Grpdify(A \times_B \COSLICE{B}{\pi\prn{e}}) \Equiv \Unit$
  since $u$ is final by assumption. This implies that applying $\Grpdify$ to the entire
  $\Sum{}$-type produces $\Unit$~\parencite{rijke:2020}.
\end{proof}

\begin{cor}
  If $\DeclVar{\pi}{\GM}{E \to B}$ is cocartesian and $\DeclVar{X}{\GM}{E \to D}$, then the left Kan
  extension $\Lan{\pi} X$ sends $\DeclVar{b}{\GM}{B}$ to $\Colim{}\prn{E_b \to E \to D}$.
\end{cor}

\subsection{Smooth and proper base change}

We want to show that smooth and proper functors satisfy the \emph{Beck--Chevalley condition}. We follow \cite[Section 8.4]{cisinski:2024}; see also \cite{anel:2024} for a general discussion.

First, using the initial-covariant factorization (see \ref{cor:initial-covariant-ofs}) of a functor
with small fibers, we recall that for every $\DeclVar{u}{\GM}{A \to B}$ the left adjoint $u_! :
\Space^A \to \Space^B$ to precomposition $u^* : \Space^B \to \Space^A$ constructed in
Theorem~\ref{thm:adj:precomp-is-right-adj} admits an explicit description. For $\DeclVar{F}{\GM}{A
\to \Space}$, consider its unstraightening:
\[
  \DiagramSquare{
    nw = E,
    ne = \Space_*,
    sw = A,
    se = \Space,
    west = \phi,
    south = F,
    nw/style = pullback,
  }
\]
Next, denote by $E \stackrel{j}{\to} X \stackrel{\psi}{\to} B$ the initial-covariant factorization
of $u \circ \phi$. Then $u_!(F)$ is the straightening of $\psi$:
\[
  \begin{tikzpicture}[diagram]
    \node (E) {$E$};
    \node [pullback, right = of E] (X) {$X$};
    \node [right = of X] (Sp) {$\Space_*$};
    \node [below = of E] (A) {$A$};
    \node [below = of X] (B) {$B$};
    \node [below = of Sp] (S) {$\Space$};
    \path[->] (E) edge node[above] {$j$} (X);
    \path[->] (X) edge (Sp);
    \path[->] (E) edge node[left] {$\phi$} (A);
    \path[->] (X) edge node[right] {$\psi$} (B);
    \path[->] (Sp) edge (S);
    \path[->] (A) edge node[below] {$u$} (B);
    \path[->] (B) edge node[below] {$u_!(F)$} (S);
  \end{tikzpicture}
\]
In particular, the unit $\DeclVar{\eta}{\GM}{F \to u^*u_!F}$ corresponds under directed univalence
to the map $\widetilde{\eta}$:
\[
  \begin{tikzpicture}[diagram]
    \node (E) {$E$};
    \node [pullback, right = of E] (Y) {$Y$};
    \node [right = of Y] (X) {$X$};
    \node [below = of Y] (A) {$A$};
    \node [below = of X] (B) {$B$};
    \path[->, dashed] (E) edge node[above] {$\widetilde{\eta}$} (Y);
    \path[->] (E) edge[bend left = 35] node[upright desc] {$j$} (X);
    \path[->] (E) edge node[upright desc] {$\phi$} (A);
    \path[->] (Y) edge (X);
    \path[->] (Y) edge node[right] {$u^*\psi$} (A);
    \path[->] (X) edge node[right] {$\psi$} (B);
    \path[->] (A) edge node[above] {$u$} (B);
  \end{tikzpicture}
\]

Our goal is to consider the behavior of this adjunction as $u$ varies. We begin by noting that  any
commuting square of categories induces a \emph{Beck--Chevalley square}:
\[
  \DiagramSquare{
    nw = A',
    ne = A,
    sw = B',
    se = B,
    north = u,
    west = f',
    east = f,
    south = v,
  }
  \qquad\rightsquigarrow\qquad
  \begin{tikzpicture}[diagram,baseline=(sw.base)]
    \SpliceDiagramSquare{
      nw = \Space^{A'},
      ne = \Space^{A},
      sw = \Space^{B'},
      se = \Space^{B},
      north = u^*,
      west = f'_!,
      east = f_!,
      south = v^*,
      north/style = <-,
      south/style = <-,
    }
    \path ($(nw)!0.3!(se)$) edge[two cell] node[above right] {$\beta$} ($(nw)!0.7!(se)$);
  \end{tikzpicture}
\]

The 2-cell $\DeclVar{\beta}{\GM}{f_!'u^* \Rightarrow v^*f_!}$ (a morphism in $\Space^A \to
\Space^{B'}$) can be explicitly described as the transpose of the map
$u^*\eta : u^* \to u^*f^*f_! \Equiv f'^* v^* f_!$.

Chasing through the involved constructions analogously to the above, we give a
description of the Beck--Chevalley map as in \cite[Section 8.4]{cisinski:2024}. Let
$\DeclVar{F}{\GM}{A \to \Space}$ and $\varphi$ its unstraightening:
\[
  \DiagramSquare{
    nw = X,
    ne = \Space_*,
    sw = A,
    se = \Space,
    west = \varphi,
    south = F,
    nw/style = pullback,
  }
\]
Recall, that $f_!F$ is given by
\[
  \begin{tikzpicture}[diagram]
    \node (X) {$X$};
    \node [pullback, right = of X] (Y) {$Y$};
    \node [right = of Y] (Sp) {$\Space_*$};
    \node [below = of X] (A) {$A$};
    \node [below = of Y] (B) {$B$};
    \node [below = of Sp] (S) {$\Space$};
    \path[->] (X) edge node[above] {$j$} (Y);
    \path[->] (Y) edge (Sp);
    \path[->] (X) edge node[left] {$\varphi$} (A);
    \path[->] (Y) edge node[left] {$\psi$} (B);
    \path[->] (Sp) edge (S);
    \path[->] (A) edge node[above] {$f$} (B);
    \path[->] (B) edge node[above] {$f_!F$} (S);
  \end{tikzpicture}
\]
The relevant unit $\DeclVar{\eta}{\GM}{F \to f^*f_!F}$ is given by the mediating map:
\[
  \begin{tikzpicture}[diagram]
    \node (X) {$X$};
    \node [pullback, right = of X] (Z) {$Z$};
    \node [right = of Z] (Y) {$Y$};
    \node [below = of Z] (A) {$A$};
    \node [below = of Y] (B) {$B$};
    \path[->, dashed] (X) edge node[above] {$\widetilde{\eta}$} (Z);
    \path[->] (X) edge[bend left = 35] node[above] {$j$} (Y);
    \path[->] (X) edge node[below left] {$\varphi$} (A);
    \path[->] (Z) edge (Y);
    \path[->] (Z) edge (A);
    \path[->] (Y) edge node[right] {$\psi$} (B);
    \path[->] (A) edge node[above] {$f$} (B);
  \end{tikzpicture}
\]
Pulling back along $u$ and $v$, resp., and factoring the upper horizontal map yields the square
\[
  \begin{tikzpicture}[diagram]
    \node (AX) {$A' \times_A X$};
    \node [right = 5cm of AX] (BY) {$B' \times_B Y$};
    \node [above right = 1.5cm and 2.5cm of AX] (Yp) {$Y'$};
    \node [below = of AX] (Ap) {$A'$};
    \node [below = of BY] (Bp) {$B'$};
    \path[->] (Yp) edge node[above right] {$r$} (BY);
    \path[->] (AX) edge node[above left] {$j'$} (Yp);
    \path[->, dashed] (AX) edge node[above] {$j '':={f' \times_fj}$} (BY);
    \path[->] (AX) edge node[left] {$u^*\varphi$} (Ap);
    \path[->] (BY) edge node[right] {$v^*\psi$} (Bp);
    \path[->] (Ap) edge node[above] {$f'$} (Bp);
  \end{tikzpicture}
\]
where $j'$ is initial and $r$ is a covariant fibration. We consider the composite $\DeclVar{\psi' :=
v^*\psi \circ r}{\GM}{Y' \to B}$ whence
\[
  \begin{tikzpicture}[diagram]
    \node [pullback] (Yp) {$Y'$};
    \node [right = 2.5cm of Yp] (Sp1) {$\Space_*$};
    \node [right = 3cm of Sp1] (BY) {$B' \times_B Y$};
    \node [pullback, right = 3cm of BY] (Y) {$Y$};
    \node [right = of Y] (Sp2) {$\Space_*$};
    \node [below = of Yp] (B1) {$B$};
    \node [below = of Sp1] (S1) {$\Space_*$};
    \node [below = of BY] (Bp) {$B'$};
    \node [below = of Y] (B2) {$B$};
    \node [below = of Sp2] (S2) {$\Space$};
    \draw ($(BY.south east) + (-.1cm,-.3cm)$) -- ++(0.2cm,0) -- ++(0,0.2cm);
    \path[->] (Yp) edge (Sp1);
    \path[->] (Yp) edge node[left] {$\psi'$} (B1);
    \path[->] (Sp1) edge (S1);
    \path[->] (BY) edge (Y);
    \path[->] (BY) edge node[left] {$v^*\psi$} (Bp);
    \path[->] (Y) edge (Sp2);
    \path[->] (Y) edge node[left] {$\psi$} (B2);
    \path[->] (Sp2) edge (S2);
    \path[->] (B1) edge node[below] {$f'_!u^*(F)$} (S1);
    \path[->] (Bp) edge node[above] {$v$} (B2);
    \path[->] (Bp) edge[bend right = 25] node[below] {$v^*f_!F$} (S2);
    \path[->] (B2) edge node[above] {$f_!F$} (S2);
  \end{tikzpicture}
\]
Finally, the Beck--Chevalley transformation corresponds to the fibered map $\DeclVar{r}{\GM}{Y' \to_B B' \times_B Y}$ over $B$.

\begin{thm}[Smooth base change, \protect{\cite[Theorem 8.4.1]{cisinski:2024}}]\label{thm:smooth-proper:smooth-base-change}
  Consider a pullback square
  \[
    \DiagramSquare{
      nw = A',
      ne = A,
      sw = B',
      se = B,
      north = u,
      west = f',
      east = f,
      south = v,
      nw/style = pullback,
    }
  \]
  of small categories where $v$ is smooth. Then, the Beck--Chevalley transformation is an equivalence
\[ f'_! u^* \simeq v^* f_!.\]
\end{thm}

\begin{proof}
We will show that the map $r$ from the preceding description is an equivalence by showing that it is also initial (and it is a covariant fibration by construction). We have the following commutative cube:
\[
  \begin{tikzpicture}[diagram]
    \node (AX) {$A' \times_A X$};
    \node [right = 5cm of AX] (BY) {$B' \times_B Y$};
    \node [below right = 1.5cm and 2.5cm of AX] (X) {$X$};
    \node [right = 5cm of X] (Y) {$Y$};
    \node [below = 3cm of AX] (Ap) {$A'$};
    \node [right = 5cm of Ap] (Bp) {$B'$};
    \node [below right = 1.5cm and 2.5cm of Ap] (A) {$A$};
    \node [right = 5cm of A] (B) {$B$};
    \node at ($(AX)!0.15!(A)$) {$\lrcorner$};
    \node at ($(BY)!0.15!(B)$) {$\lrcorner$};
    \node at ($(Ap)!0.15!(B)$) {$\lrcorner$};
    \path[->] (AX) edge node[upright desc] {$j''$} (BY);
    \path[->] (AX) edge (X);
    \path[->] (AX) edge node[upright desc, pos=0.75] {$u^*\phi$} (Ap);
    \path[->] (BY) edge node[upright desc] {$v'$} (Y);
    \path[->] (BY) edge node[upright desc, pos=0.75] {$v^*\psi$} (Bp);
    \path[->] (Y) edge node[upright desc, pos=0.3] {$\psi$} (B);
    \path[->] (Ap) edge node[upright desc, pos=0.7] {$f'$} (Bp);
    \path[->] (Ap) edge node[upright desc] {$u$} (A);
    \path[->] (Bp) edge node[upright desc] {$v$} (B);
    \path[->] (A) edge node[upright desc, pos=0.6] {$f$} (B);
    \path[->] (X) edge[commutative diagrams/crossing over] node[upright desc, pos=0.4] {$j$} (Y);
    \path[->] (X) edge[commutative diagrams/crossing over] node[upright desc, pos=0.3] {$\phi$} (A);
  \end{tikzpicture}
\]
By the pullback lemma, the top square is a pullback, too. Since $v$ is smooth and $j$ is initial, $j''$ is initial. But $j'' = r \circ j'$ so $r$ is also initial due to cancelation.
\end{proof}

One can also prove the analogous \emph{proper} base change formula as in \cite[Theorem 8.4.6]{cisinski:2024}. We need the following lemmas:

\begin{lem}
  \label{lem:smooth-proper:colim-pres-lan}
  Let $u : I \to J$ be functors between categories and $E$ be cocomplete. Let $D : I \to E$ be a diagram. Then $\Colim{I} D \simeq \Colim{J} u_!(D)$.
\end{lem}

\begin{proof}
  This is a standard proof that follows e.g.~from the universal properties. Denote by $\Delta_I : E \to I \to E$ and $\Delta_J : E \to J \to E$ the constant functors. For $e : E$ one checks $u^*(\Delta_J(e)) \equiv \Delta_I$, thus:
  \begin{align*}
    &\hom_E(\Colim{J}u_!(D), e)  \\
    &\simeq \hom_{E^J}(u_!(D), \Delta_J e)  \\
    &\simeq \hom_{E^J}(D, u^*(\Delta_J e)) \\
    &\simeq \hom_{E^J}(D, \Delta_I e) \\
    &\simeq \hom_E(\Colim{I}D,e) \qedhere
  \end{align*}
\end{proof}

\begin{lem}
 \label{lem:smooth-proper:bcc-pasting}
 Consider the following diagram of ($\GM$) categories and functors with small fibers:
 \[
   \begin{tikzpicture}[diagram]
     \node [pullback] (A2) {$A''$};
     \node [pullback, right = of A2] (A1) {$A'$};
     \node [right = of A1] (A0) {$A$};
     \node [below = of A2] (B2) {$B''$};
     \node [below = of A1] (B1) {$B'$};
     \node [below = of A0] (B0) {$B$};
     \path[->] (A2) edge node[above] {$u'$} (A1);
     \path[->] (A1) edge node[above] {$u$} (A0);
     \path[->] (A2) edge node[left] {$f''$} (B2);
     \path[->] (A1) edge node[left] {$f'$} (B1);
     \path[->] (A0) edge node[right] {$f$} (B0);
     \path[->] (B2) edge node[below] {$v'$} (B1);
     \path[->] (B1) edge node[below] {$v$} (B0);
   \end{tikzpicture}
 \]
 The Beck--Chevalley transformations $\DeclVar{\beta}{\GM}{f'_!u^* \Rightarrow v^*f_!}$ and
 $\DeclVar{\beta'}{\GM}{f_!''u'^* \Rightarrow v'^* f_!'}$ then yield a composite Beck--Chevalley
 transformation $\beta''$ \footnote{Classically, this discussion could be stated more generally in
 the context of mates of an adjunction, but we refrain from this here; although see~\textcite{buchholtz:2023}[Appendix~A.2].} for the outer square:
 \[ f''_! (uu')^* \Rightarrow (vv')^*f_!. \]
 If both $\beta$ and $\beta'$ are isomorphisms, then so is $\beta''$.
\end{lem}

\begin{proof}
  Consider the individual Beck--Chevalley transformations:
  \[
    \begin{tikzpicture}[diagram]
      \node (A2) {$\Space^{A''}$};
      \node [right = 2.5cm of A2] (A1) {$\Space^{A'}$};
      \node [right = 2.5cm of A1] (A0) {$\Space^{A}$};
      \node [below = of A2] (B2) {$\Space^{B''}$};
      \node [below = of A1] (B1) {$\Space^{B'}$};
      \node [below = of A0] (B0) {$\Space^{B}$};
      \path[->] (A1) edge node[above] {$(u')^*$} (A2);
      \path[->] (A0) edge node[above] {$u^*$} (A1);
      \path[->] (A2) edge node[left] {$f''_!$} (B2);
      \path[->] (A1) edge node[left] {$f'_!$} (B1);
      \path[->] (A0) edge node[right] {$f_!$} (B0);
      \path[->] (B1) edge node[below] {$(v')^*$} (B2);
      \path[->] (B0) edge node[below] {$v^*$} (B1);
      \path ($(A2)!0.3!(B1)$) edge[two cell] node[above right] {$\beta'$} ($(A2)!0.7!(B1)$);
      \path ($(A1)!0.3!(B0)$) edge[two cell] node[above right] {$\beta$} ($(A1)!0.7!(B0)$);
    \end{tikzpicture}
  \]
We define the desired transformation $\DeclVar{\beta''}{\GM}{f_!''(uu')^* \Rightarrow (vv')^*f_!}$ by:
\[ f_!''(uu')^* \simeq f_!''u'^*u^* \stackrel{\beta'u^*}{\Longrightarrow} v'^* f'_! u^* \stackrel{v'^*\beta}{\Longrightarrow}
v'^* v^* f_! \simeq (vv')^*f_!\]
Since natural isomorphisms are closed under whiskering and composition, if both $\beta$ and $\beta'$ are isomorphisms, then so is $\beta''$.
\end{proof}

\begin{lemC}[\protect{\cite[Theorem 8.4.5]{cisinski:2024}}]
\label{lem:smooth-proper:proper-colimit-formula}
  Let $\DeclVar{f}{\GM}{A \to B}$ be a proper functor between small categories. Then, for any $\DeclVar{b}{\GM}{B}$ and $\DeclVar{F}{\GM}{A \to \Space}$ there is an equivalence
  \[ \Colim{A_b} i^*F \stackrel{\simeq}{\longrightarrow} b^*f_!(F),\]
  where $\DeclVar{i}{\GM}{A_b \to A}$. 
\end{lemC}

\begin{proof}
 Let $\DeclVar{p}{\GM}{X \to A}$ be the left fibration associated to $F$. We get the following diagram:

 \[
   \begin{tikzpicture}[diagram]
     \node [pullback] (Xb) {$X_b$};
     \node [pullback, right = 2.5cm of Xb] (Xsl) {$\SLICE{X}{b}$};
     \node [right = 2.5cm of Xsl] (X) {$X$};
     \node [pullback, below = of Xb] (Ab) {$A_b$};
     \node [pullback, below = of Xsl] (Asl) {$\SLICE{A}{b}$};
     \node [below = of X] (A) {$A$};
     \node [below = of Ab] (U) {$\Unit$};
     \node [below = of Asl] (Bsl) {$\SLICE{B}{b}$};
     \node [below = of A] (B) {$B$};
     \path[->] (Xb) edge node[above] {$k_b$} (Xsl);
     \path[->] (Xsl) edge node[above] {$u_b$} (X);
     \path[->] (Xb) edge (Ab);
     \path[->] (Xsl) edge (Asl);
     \path[->] (X) edge node[right] {$p$} (A);
     \path[->] (Ab) edge node[above] {$i_b$} (Asl);
     \path[->] (Asl) edge node[above] {$v_b$} (A);
     \path[->] (Ab) edge (U);
     \path[->] (Asl) edge node[right] {$\SLICE{f}{b}$} (Bsl);
     \path[->] (A) edge node[right] {$f$} (B);
     \path[->] (U) edge node[above] {$(b,\ArrId{B})$} (Bsl);
     \path[->] (Bsl) edge node[above] {$w_b$} (B);
     \path[->] (U) edge[bend right = 25] node[below] {$b$} (B);
   \end{tikzpicture}
 \]
As $w_b$ is a right fibration it is final, then so is $v_b$ by properness of $f$.
By Theorem~\ref{thm:smooth-proper:smooth-base-change} applied to the lower right square, we get an isomorphism:
\[ {\left(\SLICE{f}{b}\right)}_! (v_b^*(F)) \simeq w_b^* (f_!(F))\]
By Lemma~\ref{lem:smooth-proper:colim-pres-lan} we get $\Colim{\SLICE{A}{b}} v_b^*(F) \simeq \Colim{\SLICE{B}{b}}(\SLICE{f}{b})_! v^*b(F)$, and by the above Beck--Chevalley equivalence $\Colim{\SLICE{B}{b}} (\SLICE{f}{b})_! v^*_b(F) \simeq \Colim{\SLICE{B}{b}} w_b^* f_!(F)$.
Consider the induced diagram:
\[
  \begin{tikzpicture}[diagram]
    \node (TAb) {$\Space^{A_b}$};
    \node [right = 2.5cm of TAb] (TAsl) {$\Space^{\SLICE{A}{b}}$};
    \node [right = 2.5cm of TAsl] (TA) {$\Space^{A}$};
    \node [below = of TAb] (S) {$\Space$};
    \node [below = of TAsl] (TBsl) {$\Space^{\SLICE{B}{b}}$};
    \node [below = of TA] (TB) {$\Space^{B}$};
    \path[->] (TAb) edge node[left] {$\Colim{A_b}$} (S);
    \path[->] (TAsl) edge node[above] {$i_b^*$} (TAb);
    \path[->] (TAsl) edge node[right] {${\SLICE{f}{b}}_!$} (TBsl);
    \path[->] (TA) edge[bend right = 35] node[upright desc] {$i^*$} (TAb);
    \path[->] (TA) edge node[above] {$v_b^*$} (TAsl);
    \path[->] (TA) edge node[right] {$f_!$} (TB);
    \path[->] (TBsl) edge node[above] {$(b,\ArrId{b})^*$} (S);
    \path[->] (TB) edge[bend left = 35] node[upright desc] {$b^*$} (S);
    \path[->] (TB) edge node[above] {$w_b^*$} (TBsl);
  \end{tikzpicture}
\]
By Lemma~\ref{lem:smooth-proper:bcc-pasting}, we get a natural transformation for the composite square
\[ \Colim{A_b} i^*F \Rightarrow b^* f_!(F).\]

By properness of $f$ and finality of $(b,\ArrId{b})$ and applying Lemma~\ref{lem:kan:theorem-a-for-psh} and Corollary~\ref{cor:kan:right-cofinal-char}, we obtain the following commutative diagrams:
\[
  \begin{tikzpicture}[diagram]
    \node (TAb) {$\Space^{A_b}$};
    \node [right = 2.5cm of TAb] (TAsl) {$\Space^{\SLICE{A}{b}}$};
    \node [right = 2.5cm of TAsl] (TA) {$\Space^A$};
    \node [below = of TAsl] (S) {$\Space$};
    \path[->] (TAb) edge[bend right = 15] node[below left] {$\Colim{A_b}$} (S);
    \path[->] (TAsl) edge node[above] {$i_b^*$} (TAb);
    \path[->] (TAsl) edge node[left] {$\Colim{\SLICE{A}{b}}$} (S);
    \path[->] (TA) edge node[above] {$v_b^*$} (TAsl);
    \path[->] (TA) edge[bend left = 15] node[below right] {$\Colim{A}$} (S);
  \end{tikzpicture}
\]
\[
  \begin{tikzpicture}[diagram]
    \node (S0) {$\Space$};
    \node [right = 2.5cm of S0] (TBsl) {$\Space^{\SLICE{B}{b}}$};
    \node [right = 2.5cm of TBsl] (TB) {$\Space^B$};
    \node [below = of TBsl] (S) {$\Space$};
    \path (S0) edge[equality, bend right = 15] (S);
    \path[->] (TBsl) edge node[above] {$(b,\ArrId{b})^*$} (S0);
    \path[->] (TBsl) edge node[left] {$\Colim{\SLICE{B}{b}}$} (S);
    \path[->] (TB) edge node[above] {$w_b^*$} (TBsl);
    \path[->] (TB) edge[bend left = 15] node[below right] {$\Colim{B}$} (S);
  \end{tikzpicture}
\]

This yields the following square:
\[
  \begin{tikzpicture}[diagram]
    \node (TL) {$\Colim{A_b} i^*F$};
    \node [right = 6cm of TL] (TR) {$b^* f_!(F)$};
    \node [below = 1.5cm of TL] (ML) {$\Colim{\SLICE{A}{b}} v_b^*F$};
    \node [below = 1.5cm of TR] (MR) {$\Colim{\SLICE{B}{b}} (\SLICE{f}{b})_! v_b^*F$};
    \node [below = 1.5cm of ML] (BL) {$\Colim{\SLICE{B}{b}} (\SLICE{f}{b})_! v^*b$};
    \node [below = 1.5cm of MR] (BR) {$\Colim{\SLICE{B}{b}} w_b^* f_!(F)$};
    \path[->] (TL) edge (TR);
    \path[->] (TL) edge node[left] {$\simeq$} (ML);
    \path[->] (TR) edge node[right] {$\simeq$} (MR);
    \path[->] (ML) edge node[left] {$\simeq$} (BL);
    \path[->] (MR) edge node[right] {$\simeq$} (BR);
    \path[->] (BL) edge node[below] {$\simeq$} (BR);
  \end{tikzpicture}
\]
By 2-out-of-3, the top horizontal map is also an equivalence, as postulated.
\end{proof}

Finally, we can prove the proper base change theorem:

\begin{thm}[Proper base change, \protect{\cite[Theorem 8.4.6]{cisinski:2024}}]
  Consider a pullback square 
  \[
    \DiagramSquare{
      nw = A',
      ne = A,
      sw = B',
      se = B,
      north = u,
      west = f',
      east = f,
      south = v,
      nw/style = pullback,
    }
  \]
  of small categories where $f$ is proper. Then, the Beck--Chevalley transformation is an equivalence
  \[ f'_!u^* \simeq v^* f_!.\]
  \end{thm}

  \begin{proof}
    For $\DeclVar{b}{\GM}{B}$ consider the following cube:
    \[
      \begin{tikzpicture}[diagram]
        \node (Apb) {$A'_{b'}$};
        \node [right = 4cm of Apb] (Ab) {$A_b$};
        \node [below left = 1.25cm and 2cm of Apb] (Ap) {$A'$};
        \node [right = 4cm of Ap] (A) {$A$};
        \node [below = 2.5cm of Apb] (U1) {$\ObjTerm{}$};
        \node [right = 4cm of U1] (U2) {$\ObjTerm{}$};
        \node [below = 2.5cm of Ap] (Bp) {$B'$};
        \node [right = 4cm of Bp] (B) {$B$};
        \node at ($(Apb)!0.2!(U2)$) {$\lrcorner$};
        \node [rotate=-90] at ($(Apb)!0.2!(Bp)$) {$\lrcorner$};
        \node [rotate=-90] at ($(Ab)!0.2!(B)$) {$\lrcorner$};
        \node at ($(Ap)!0.2!(B)$) {$\lrcorner$};
        \path[->] (Apb) edge node[above] {$\simeq$} (Ab);
        \path[->] (Apb) edge node[above left] {$i'$} (Ap);
        \path[->] (Apb) edge (U1);
        \path[->] (Ab) edge node[above left] {$i$} (A);
        \path[->] (Ab) edge (U2);
        \path[->] (Ap) edge node[left, pos=0.3] {$f'$} (Bp);
        \path (U1) edge[equality] (U2);
        \path[->] (U1) edge node[below right] {$b'$} (Bp);
        \path[->] (U2) edge node[below right] {$b$} (B);
        \path[->] (Bp) edge node[below, pos=0.6] {$v$} (B);
        \path[->] (A) edge[commutative diagrams/crossing over] node[left, pos=0.3] {$f$} (B);
        \path[->] (Ap) edge[commutative diagrams/crossing over] node[above, pos=0.7] {$u$} (A);
      \end{tikzpicture}
    \]
The front face is a pullback by assumption. The side faces are pullback by construction. Then, by pullback cancellation, so is the back face.

Applying Lemma~\ref{lem:smooth-proper:proper-colimit-formula} we obtain the following diagram whose vertical maps are equivalences:
\[
  \begin{tikzpicture}[diagram]
    \node (TL) {$\Colim{A'_{b'}} i'^*(u^*F)$};
    \node [right = 5cm of TL] (TR) {$\Colim{A_b} i^*F$};
    \node [below = 1.5cm of TR] (R1) {$f_!(F)(b)$};
    \node [below = 1.5cm of R1] (R2) {$f_!(F)(v(b'))$};
    \node [below = 1.5cm of R2] (BR) {$v^*f_!(F)(b)$};
    \node [left = 5cm of BR] (BL) {$f'_!(u^*F)(b')$};
    \path[->] (TL) edge (TR);
    \path[->] (TL) edge node[left] {$\simeq$} (BL);
    \path[->] (TR) edge node[right] {$\simeq$} (R1);
    \path (R1) edge[equality] (R2);
    \path (R2) edge[equality] (BR);
    \path[->] (BL) edge node[above] {$\simeq$} (BR);
  \end{tikzpicture}
\]
By the 2-out-of-3 property, the horizontal map is an equivalence.
  \end{proof}

\section{Conclusions and future work}
\label{sec:conc}

We have introduced and studied the impact of the $\infty$-categorical Yoneda embedding in
\STT{}. This includes the development of classical concepts (Kan extensions, adjoints, (co)limits,
\etc{}), all in the synthetic $\infty$-categorical setting. While some of the basic theory had been
investigated in \STT{} already, we were able to produce the first non-trivial concrete examples of,
\eg{}, adjunctions (Theorem~\ref{thm:adj:precomp-is-right-adj}) and give several more refined versions of
existing theorems (Theorem~\ref{thm:tw:nat-iso}) which more closely match their standard counterparts.

\subsection{Related work}

There are several closely related type-theoretic approaches to synthetic ($\infty$-)category
theory. We may roughly divide these into (1) directed type theory, where every type is a category
but various operations ($\Prod{}$) must be restricted, and (2) variations on simplicial type
theory. For instance, many directed type theories have been proposed and studied over the years~%
\parencite{%
  licata:2011,%
  warren:2013,%
  nuyts:2015,%
  kavvos:directed:2019,%
  buchholtz:2019,%
  kolomatskaia:2023,%
  weaver:2020,%
  weaver:phd,%
  ahrens:2023,%
  north:2018,%
  nuyts:2020,%
  neumann:2024,%
  }.
In general, while these type theories are a promising approach to formalize category theory in type
theory, none of them have thus far received as much attention as \STT{} and, consequently, none have
developed category theory to the extent of this work. Furthermore, it is substantially harder to
design a directed type theory in this style (as it is a more radical alteration of the basic rules
of type theory) and most proposals handle only $1$-category theory rather than
$\prn{\infty,1}$-categories. We note, however, that some of these type theories do include a version
of Theorem~\ref{thm:tw:strong-yoneda} in the form of directed path
induction~\parencite{north:2018,nuyts:2020,neumann:2024}. Given, however, that few of our arguments
rely on types which are not categories, we expect many of them to transfer to sufficiently rich
future variants of directed type theory.

Other variations of simplicial type theory have been considered in the literature. For instance,
several papers use additional judgmental structure (extension types) to get more definitional
equalities around
hom-types~\parencite{riehl:2017,bardomiano:2021,weinberger:phd,buchholtz:2023,weinberger:twosided:2024,weinberger:sums:2024}
at the cost of making the interval a second-class type similar to two-level type
theory~\parencite{annenkov:2023,voevodsky:2013}. Other versions have favored a cubical
interval~\parencite{gratzer:2024} or even a cubical interval atop a cubical version of
\HoTT{}~\parencite{weaver:2020,weaver:phd}. Aside from the addition of modalities, our version of
\STT{} is deliberately minimalistic: we use only ordinary \HoTT{} with a handful of
postulates. Accordingly, our results can be interpreted into essentially any incarnation of modal
\STT{} and does not rely on extra definitional equalities.

Finally, there are many attempts to formulate more conceptual and synthetic foundations for
$\infty$-category theory which do not rely on type theory. For instance, the $\infty$-cosmos program
of \textcite{riehl:2022} aims to give a systematic account of the formal category theory and
model-independence using 2-category theory. On the other hand, most practitioners in the field
attempt to give looser ``model independent'' arguments which avoid relying on explicit computations
as much as possible. We have successfully translated some of these arguments into our framework,
proving that this informal discipline is effective (\eg{}, Section~\ref{sec:kan}). More recently,
\textcite{cisinski:2024} have begun to redevelop $\infty$-category theory in a deliberately informal
and high-level language, splitting the difference between a formal theory like \STT{} and the usual
``model-independent'' discipline of practitioners. We expect that their arguments can be translated
into \STT{} and we have shown that some of their primitive axioms are \emph{provable}
in \STT{} (\eg{},
Theorems~\ref{thm:prelims:space} and \ref{thm:prelims:full-subcat} and Lemma~\ref{lem:adj:functoriality-of-init-objects}).

\subsection{Future work}

Many promising avenues for future work remain to be explored. While we have focused on presheaf
categories and immediate consequences of their theory, we plan to port other foundational results
from category theory (presentable and accessible categories, Bousfield localizations, topos theory,
\etc{}) into \STT{}.  It would also be desirable to adapt more parts of the internal
$\infty$-category theory and $\infty$-topos theory of Martini and
Wolf~\parencite{martini:cocartesianfibrations:2022,martini:colimitscocompletions:2021,martini:phd,martini:presentable:2022,martini:propermorphisms:2023,martini:yonedaslemma:2021,wolf:phd}
to \STT{}.
Additionally, we hope to extend a proof assistant like Agda~\parencite{agda}
with the necessary support for modalities to give machine-checked versions of the proofs in this
paper. On the foundational side, \STT{} presently relies on a handful of axioms
(Appendix~\ref{sec:axioms}) and therefore satisfies only normalization and not canonicity. In future work,
we hope to examine which of these principles can be given computational interpretations and to what
extent one can `compute' with synthetic $\infty$-categories.

\appendix
\section{The formal rules of \texorpdfstring{\MTT{}}{MTT}}
\label{sec:mtt-rules}

The formal syntax of \MTT{} is comprised of four judgments: $\IsCx{\Gamma}$,
$\IsSb{\delta}{\Delta}$, $\IsTm{a}{A}$, and $\IsTy{A}$. We list the relevant novel rules for these
judgments below:%
\def\MathparLineskip{\lineskip=11pt}
\begin{mathparpagebreakable}
  \JdgFrame{\IsCx{\Gamma}}
  \\
  \inferrule{ }{
    \IsCx{\EmpCx}
  }
  \and
  \inferrule{
    \IsCx{\Gamma}
  }{
    \IsCx{\LockCx{\Gamma}}
  }
  \and
  \inferrule{
    \IsCx{\Gamma}
    \\
    \IsTy[\LockCx{\Gamma}]{A}
  }{
    \IsCx{\MECx{\Gamma}{A}}
  }
  \\
  \inferrule{
    \IsCx{\Gamma}
  }{
    \EqCx{\LockCx{\Gamma}[\ArrId{}]}{\Gamma}
    \\
    \EqCx{\LockCx{\LockCx{\Gamma}[\mu]}[\nu]}{\LockCx{\Gamma}[\mu \circ \nu]}
  }
  \\
  \JdgFrame{\IsSb{\delta}{\Delta}}
  \\
  \inferrule{ }{
    \IsSb{\EmpSb}{\EmpCx}
  }
  \and
  \inferrule{
    \IsCx{\Gamma}
    \\
    \IsTy[\LockCx{\Gamma}]{A}
  }{
    \IsSb[\MECx{\Gamma}{A}]{\Wk}{\Gamma}
  }
  \and
  \inferrule{
    \IsSb{\delta}{\Delta}
    \\
    \IsTm[\LockCx{\Gamma}]{a}{A}
  }{
    \IsSb{\ESb{\delta}{a}}{\MECx{\Delta}{A}}
  }
  \\
  \inferrule{
    \IsSb{\delta}{\Delta}
  }{
    \IsSb[\LockCx{\Gamma}]{\LockSb{\delta}}{\LockCx{\Delta}}
  }
  \\
  \inferrule{
    \IsCx{\Gamma}
    \\
    \Mor[\alpha]{\mu}{\nu}
  }{
    \IsSb[\LockCx{\Gamma}[\nu]]{\KeySb{\alpha}{\Gamma}}{\LockCx{\Gamma}[\mu]}
  }
  \\
  \inferrule{
    \IsSb{\gamma}{\EmpCx}
  }{
    \EqSb{\EmpSb}{\gamma}{\EmpCx}
  }
  \and
  \inferrule{
    \IsSb{\delta}{\Delta}
    \\
    \IsTm[\LockCx{\Gamma}]{a}{A}
  }{
    \EqSb{\Wk \circ \prn{\ESb{\delta}{a}}}{\delta}{\Delta}
  }
  \and
  \inferrule{
    \IsSb{\delta}{\MECx{\Delta}{A}}
  }{
    \EqSb{\ESb{\prn{\Wk \circ \delta}}{\Sb{\Var}{\delta}}}{\delta}{\MECx{\Delta}{A}}
  }
  \and
  \inferrule{
    \IsSb{\delta}{\Delta}
  }{
    \EqSb{\LockSb{\delta}[\ArrId{}]}{\delta}{\Delta}
    \\
    \EqSb[
      \LockCx{\Gamma}[\nu\circ\mu]
    ]{
      \LockSb{\delta}[\nu\circ\mu]
    }{
      \LockSb{\LockSb{\delta}[\nu]}
    }{
      \LockCx{\Delta}[\nu\circ\mu]
    }
  }
  \and
  \inferrule{
    \IsCx{\Gamma}
  }{
    \EqSb[\LockCx{\Gamma}]{\KeySb{\ArrId{}}{\Gamma}}{\ArrId{}}{\LockCx{\Gamma}}
    \\
    \EqSb[
      \LockCx{\Gamma}[\xi]
    ]{
      \KeySb{\alpha}{\Gamma} \circ \KeySb{\beta}{\Gamma}
    }{
      \KeySb{\alpha \circ \beta}{\Gamma}
    }{
      \LockCx{\Gamma}[\mu]
    }
  }
  \and
  \inferrule{
    \IsSb{\delta}{\Delta}
    \\
    \mu \le \nu
  }{
    \EqSb[
      \LockCx{\Gamma}[\nu]
    ]{
      \KeySb{\alpha}{\Delta} \circ \LockCx{\delta}
    }{
      \LockSb{\delta} \circ \KeySb{\alpha}{\Gamma}
    }{
      \LockCx{\Delta}[\mu]
    }
  }
  \and
  \inferrule{
    \IsCx{\Gamma}
    \\
    \Mor[\alpha]{\mu_0}{\nu_0}
    \\
    \Mor[\beta]{\mu_1}{\nu_1}
  }{
    \LockCx{\Gamma}[\nu_1 \circ \nu_0] \vdash
    \substack{
      \LockCx{\KeySb{\beta}{\Gamma}}[\mu_0]
      \circ
      \KeySb{\alpha}{\Gamma}
      \\
      =
      \KeySb{\beta \HComp \alpha}{\Gamma}
    }
    : \LockCx{\Gamma}[\mu_1 \circ \mu_0]
  }
  \\
  \JdgFrame{\IsTy{A}}
  \\
  \inferrule{
    \IsTy[\LockCx{\Gamma}]{A}
  }{
    \IsTy{\Modify{A}}
  }
  \and
  \inferrule{
    \IsSb{\delta}{\Delta}
    \\
    \IsTy[\LockCx{\Delta}]{A}
  }{
    \EqTy{\Sb{\Modify{A}}{\delta}}{\Modify{\Sb{A}{\LockSb{\delta}}}}
  }
  \\
  \JdgFrame{\IsTm{a}{A}}
  \\
  \inferrule{
    \IsTy[\LockCx{\Gamma}]{A}
  }{
    \IsTm[\LockCx{\MECx{\Gamma}{A}}]{\Var}{\Sb{A}{\LockSb{\Wk}}}
  }
  \and
  \inferrule{
    \IsTm[\LockCx{\Gamma}]{a}{A}
  }{
    \IsTm{\MkMod{a}}{\Modify{A}}
  }
  \and
  \inferrule{
    \IsTy[\MECx{\Gamma}[\nu]{\Modify{A}}]{B}
    \\
    \IsTm[\MECx{\Gamma}[\nu\circ\mu]{A}]{b}{\Sb{B}{\ESb{\Wk}{\MkMod[\mu]{\Var}}}}
    \\
    \IsTm[\LockCx{\Gamma}[\nu]]{a}{\Modify{A}}
  }{
    \IsTm{\LetMod{a}{b}}{\Sb{B}{\ESb{\ISb}{a}}}
  }
  \\
  \inferrule{
    \IsTy[\MECx{\Delta}[\nu]{\Modify{A}}]{B}
    \\
    \IsTm[\MECx{\Delta}[\nu\circ\mu]{A}]{b}{\Sb{B}{\ESb{\Wk}{\MkMod[\mu]{\Var}}}}
    \\
    \IsTm[\LockCx{\Delta}[\nu]]{a}{\Modify{A}}
    \\
    \IsSb{\delta}{\Delta}
  }{
    \IsTm{
      \substack{
        \LetMod{\Sb{a}{\LockSb{\delta}[\nu]}}{\Sb{b}{\ESb{\prn{\delta \circ \Wk}}{\Var}}}
        \\
        {}={}
        \Sb{\prn{\LetMod{a}{b}}}{\delta}
      }
    }{
      \Sb{B}{\ESb{\delta}{a}}
    }
  }
  \and
  \inferrule{
    \IsTm[\LockCx{\Gamma}]{a}{A}
    \\
    \IsSb{\delta}{\Delta}
  }{
    \EqTm{\Sb{\MkMod{a}}{\delta}}{\MkMod{\Sb{a}{\LockSb{\delta}}}}{\Modify{\Sb{A}{\LockSb{\delta}}}}
  }
  \and
  \inferrule{
    \IsSb{\delta}{\Delta}
    \\
    \IsTm[\LockCx{\Gamma}]{a}{\Sb{A}{\LockSb{\delta}}}
    \\
    \IsTy[\LockCx{\Delta}]{A}
  }{
    \EqTm[\LockCx{\Gamma}]{\Sb{\Var}{\LockCx{\ESb{\delta}{a}}}}{a}{\Sb{A}{\LockSb{\delta}}}
  }
  \and
  \hspace*{-5pt}\inferrule{
    \IsTy[\MECx{\Gamma}[\nu]{\Modify{A}}]{B}
    \\
    \IsTm[\MECx{\Gamma}[\nu\circ\mu]{A}]{b}{\Sb{B}{\ESb{\Wk}{\MkMod[\mu]{\Var}}}}
    \\
    \IsTm[\LockCx{\Gamma}[\nu]]{a}{\Modify{A}}
  }{
    \EqTm{\prn{\LetMod{\MkMod{a}}{b}}}{\Sb{b}{\ESb{\ISb}{a}}}{\Sb{B}{\ESb{\ISb}{\MkMod{a}}}}
  }
\end{mathparpagebreakable}

\section{The complete list of axioms}
\label{sec:axioms}

\univalenceaxiom*
\intaxiom*
\dualityaxiom*
\crispidaxiom*
\intopaxiom*
\intdiscaxiom*
\intglobalaxiom*
\cubessepaxiom*
\twarraxiom*

\begin{figure}
  \centering
  \begin{tabular}{| c | c |}
    The present work & \textcite{gratzer:2024}
    \\\hline
    Axiom \ref{ax:univalence} & Axiom 4
    \\
    Axiom \ref{ax:int} & Axiom 1
    \\
    Axiom \ref{ax:duality} & Axiom 10
    \\
    Axiom \ref{ax:crispid} & Axiom 5
    \\
    Axiom \ref{ax:op-of-int-is-int} & Axiom 2
    \\
    Axiom \ref{ax:discrete-iff-crisp} & Axiom 6
    \\
    Axiom \ref{ax:global-points} & Axiom 7
    \\
    Axiom \ref{ax:cubes-separate} & Axiom 8
    \\
    Axiom \ref{ax:twisted-arrow-equiv} & No equivalent
    \\
    No equivalent & Axiom 3
    \\
    No equivalent & Axiom 9
  \end{tabular}

  \caption{A comparison of our axioms versus those of \textcite{gratzer:2024}.}
\end{figure}

\printbibliography

\end{document}